\documentclass[11pt,letterpaper]{article}
\usepackage{amsmath,amssymb,amsthm}
\usepackage{cite}
\usepackage{array}
\newcommand{\tablegap}{\\[8pt]}
\newcommand{\cell}[1]{\begin{tabular}{@{}c@{}}#1\end{tabular}}
\newcolumntype{L}[1]{>{\raggedright\arraybackslash}m{#1}}
\newcolumntype{C}[1]{>{\centering\arraybackslash}m{#1}}
\usepackage{tikz}
\usepackage[T1]{fontenc}
\usepackage[utf8]{inputenc}
\usepackage[letterpaper,margin=1in]{geometry}
\usepackage{lmodern,microtype,graphicx,adjustbox,authblk,needspace}
\usepackage[colorlinks=true,linkcolor=blue,citecolor=blue,urlcolor=blue]{hyperref}

\allowdisplaybreaks[1]
\hypersetup{pdftitle={Breaking the Multiplicative Overhead in Quantum Entropy Estimation}}
\newtheorem{theorem}{Theorem}[section]
\newtheorem{proposition}[theorem]{Proposition}
\newtheorem{lemma}[theorem]{Lemma}
\theoremstyle{definition}
\newtheorem{definition}[theorem]{Definition}
\newcommand{\Tr}{\operatorname{Tr}}

\newcommand{\Ot}{\widetilde O}
\newcommand{\OtOmega}{\widetilde\Omega}
\newcommand{\ket}[1]{\lvert #1\rangle}
\newcommand{\bra}[1]{\langle #1\rvert}
\newcommand{\eps}{\varepsilon}
\newcommand{\norm}[1]{\lVert #1\rVert}
\newcommand{\abs}[1]{\lvert #1\rvert}

\begin{document}
\title{Breaking the Multiplicative Overhead in Quantum Entropy Estimation}
\author[1]{Junxiang Huang}
\author[2]{Chenyang Li}
\author[2]{Lu-Fan Zhang}
\author[2]{Yusen Wu\thanks{Corresponding author: \href{mailto:yusen.wu@bnu.edu.cn}{yusen.wu@bnu.edu.cn}.}}
\author[1,3]{Yukun Zhang\thanks{Corresponding author: \href{mailto:yukun.zhang@maths.ox.ac.uk}{yukun.zhang@maths.ox.ac.uk}.}}
\affil[1]{Center on Frontiers of Computing Studies, School of Computer Science, Peking University, Beijing 100871, China}
\affil[2]{School of Artificial Intelligence, Beijing Normal University, Beijing 100875, China}
\affil[3]{Mathematical Institute, University of Oxford, Oxford OX2 6GG, United Kingdom}
\date{}
\maketitle
\begin{abstract}
The von Neumann, Tsallis, and R\'enyi entropies are fundamental measures of quantum information. A common estimation strategy repeats a globally accurate spectral transformation at every statistical readout step, multiplying the worst-case costs of the two procedures and leaving gaps to known query lower bounds. We overcome this overhead by organizing the approximation and readout across spectral scales. We construct multi-level algorithms that normalize each spectral contribution locally and adapt its estimation precision to its magnitude; variable-time estimation further accounts for the probability of reaching expensive spectral tests. Given controlled access to a purification-preparation unitary and its inverse, we obtain the following bounds for additive error $\varepsilon$, rank upper bound $R$, and fixed order $\alpha$. Tsallis entropy estimation has near-optimal rank-independent query complexity $\widetilde O(\varepsilon^{-1/(\alpha-1)})$ for $1<\alpha<2$, when the allowed dimension and rank accommodate the lower-bound instances, and $\widetilde O(1/\varepsilon)$ for $\alpha\ge2$, where $\widetilde O$ suppresses logarithmic factors. The former removes a factor $1/\varepsilon$ from the previous upper bound; the latter extends the known integer-order scaling to noninteger orders. We also improve the von Neumann entropy bound from $\widetilde O(R/\varepsilon^2)$ to $\widetilde O(R/\varepsilon)$ and obtain $\widetilde O(R/\varepsilon)$ for noninteger R\'enyi orders $1<\alpha<3$, with further bounds at other orders. More generally, our functional-estimation theorem replaces the global product by a sum governed by local approximation costs, function magnitudes, and spectral masses. We further estimate fixed logarithmic moments and entropy variance using $\widetilde O(R/\varepsilon)$ queries, providing efficient access to fluctuation parameters that enter finite-blocklength quantum compression and pure-state entanglement conversion. More broadly, our framework extends beyond entropy to a broad class of density-matrix functionals, offering a systematic approach toward optimal query complexity in quantum spectral estimation.
\end{abstract}
\par\smallskip\noindent\textbf{Keywords:}
Quantum entropy, multi-level estimation, purified query access, polynomial approximation, amplitude estimation.
\par\clearpage
\tableofcontents\clearpage

\section{Introduction}\label{sec:intro}
The von Neumann, Tsallis, and R\'enyi entropies describe complementary aspects of the spectrum of a quantum state. The von Neumann entropy quantifies mixedness and the entanglement of a bipartite pure state, while the order parameter in the Tsallis and R\'enyi families changes the weight assigned to large and small eigenvalues~\cite{Watrous18,Renyi61,Tsallis88}. Their estimation is therefore a basic instance of extracting useful information from a quantum state without reconstructing its density operator. We study the number of coherent queries needed for this task, with particular attention to the interaction between the desired accuracy and an upper bound on the rank.

Our input is a unitary that prepares a purification of an unknown state, together with its inverse and controlled versions. Although this oracle does not reveal the eigenbasis, it permits a block encoding of the density operator~\cite{GilyenLi20}. Quantum singular value transformation (QSVT) can then approximate a function of that operator~\cite{QSVT19}, and amplitude estimation can recover its expectation in the input state. This combination underlies rank-sensitive entropy estimators~\cite{WangEntropies24,WZL24} and dimension-independent estimators for fixed positive powers~\cite{LiuWang26}. Nevertheless, gaps remain between these query upper bounds and the available lower bounds~\cite{Lifting25,BoundsList25,WangFramework26}. We focus on a specific algorithmic source of the gap: every statistical query invokes a spectral transform accurate enough for the most demanding part of the spectrum. We reduce this overhead by allocating the spectral processing and readout costs jointly across scales.

To make the issue concrete, we consider Tsallis entropy of fixed order $1<\alpha<2$. Up to an affine conversion, the target is the power sum $\Tr(\rho^\alpha)$, which can be written as the expectation of $g(x)=x^{\alpha-1}$ under the eigenvalue-weighted spectral distribution of $\rho$. Eigenvalues smaller than a cutoff $\tau$ contribute at most $\tau^{\alpha-1}$, because their total probability mass is at most one. Thus $\tau$ of order $\varepsilon^{1/(\alpha-1)}$ suffices for additive error $\varepsilon$. A cutoff-based transform with query cost $\widetilde O(1/\tau)$, followed by additive amplitude estimation with $\widetilde O(1/\varepsilon)$ calls to that transform, has total cost $\widetilde O(\varepsilon^{-1-1/(\alpha-1)})$. This calculation reproduces the precision dependence of the earlier rank-independent bound~\cite[Theorem~III.2]{LiuWang26} and isolates the product that we seek to avoid. The transform is designed for the most expensive spectral scale, but the readout repeatedly pays that cost while estimating the entire moment.

\begin{quote}
\emph{Can quantum entropy estimation overcome the multiplicative cost of spectral processing and statistical readout?}
\end{quote}

Multiplicative lower bounds elsewhere in quantum computation make this a substantive question. For Boolean functions on disjoint input blocks, the general adversary bound composes multiplicatively and characterizes bounded-error quantum query complexity~\cite{HoyerLeeSpalek07,Reichardt11}. Thus an OR of $m$ independent $k$-bit parities requires $\Theta(\sqrt m\,k)$ queries, and the depth-two formula $\operatorname{AND}_m\circ\operatorname{OR}_n$ requires $\Theta(\sqrt{mn})$ queries. The phenomenon also persists through recursive composition: balanced ternary-majority formulas of depth $d$ have query complexity $\Theta(2^d)$ on $3^d$ inputs~\cite[Theorem~1.1]{ReichardtSpalek12}. Moreover, composition can remain an obstruction after a problem is encoded into another computational task. In sparse Hamiltonian simulation, Low embeds $\operatorname{PARITY}_n\circ\operatorname{OR}_m$ into the dynamics and obtains the corresponding $\Omega(n\sqrt m)$ lower bound~\cite[Section~5]{LowSpectral19}. Across nested search, parity, recursive formula evaluation, and simulation, these products arise from the composition of independent input problems. They motivate examining the dependence between the two tasks in entropy estimation.

For entropy estimation, we answer the question affirmatively by exploiting a dependence absent from those independent-block examples: the local function values and the probabilities with which they are observed come from the same spectrum. Fine spectral resolution can be expensive precisely where the function value, spectral mass, or both are small. A global transform followed by a uniformly accurate readout loses this relation when it maximizes the two costs separately. We instead choose the spectral approximation and statistical precision jointly. For the Tsallis profile above, a band of scale $\delta$ has function magnitude $O(\delta^{\alpha-1})$. Dividing by that magnitude changes the normalized accuracy required to estimate the band's contribution to error $e$ from $e$ to order $e/\delta^{\alpha-1}$. Consequently, the local query cost is $\widetilde O(\delta^{\alpha-2}/e)$. Summing geometrically spaced bands down to $\tau$ gives $\widetilde O(\varepsilon^{-1/(\alpha-1)})$, rather than $\widetilde O(\varepsilon^{-1-1/(\alpha-1)})$. Thus local normalization with fixed-time readout already removes the extra $1/\varepsilon$ factor in this example.

For von Neumann entropy, the relevant additional information is the mass of the small-eigenvalue region. If the rank is at most $R$, a band near $\delta$ and the paths that reach it have probability at most $O(R\delta)$. We construct a coherent sequence of spectral tests in which a path stops once its scale has been identified. Variable-time amplification and estimation~\cite{Ambainis12,CGJ19} allow the decreasing survival probability to enter the statistical cost, rather than charging every path for the deepest test. Since the resolution cost grows as $1/\delta$ while the rank-based mass bound decreases as $R\delta$, this refinement yields a total query bound of $\widetilde O(R/\varepsilon)$.

We develop these observations into a functional-estimation theorem for $\Tr[\rho g(\rho)]$. The theorem replaces the product of globally maximized costs by a sum that retains the local approximation cost, function magnitude, and spectral probability bounds. Beyond the three entropy formulas, it gives a criterion for determining when scalar spectral estimation can avoid the overhead of repeatedly implementing a uniformly accurate operator transform. As one concrete consequence, the same theorem gives $\widetilde O(R/\varepsilon)$ estimation of $\Tr[\rho(-\log\rho)^k]$ for every fixed positive integer $k$, and $\widetilde O(R/\varepsilon)$ estimation of entropy variance. The latter is a fluctuation quantity appearing in second-order source coding and pure-state entanglement conversion~\cite{DattaLeditzky15,AbdelhadiRenes20}. Section~\ref{sec:log-moments} derives these estimates, providing coherent-query access to the spectral parameters that enter those coding results.

Multi-level estimation itself has important precedents, particularly the recent framework of Chen \emph{et al.} for classical distributions~\cite{Multilevel26}; variable-time quantum-state reductions also appear in Wang, Zhang, and Li~\cite{WZL24}. Building on these methods, we give the density-spectrum construction and probability-sensitive cost analysis that establish the entropy improvements above. Under arbitrary purified access, this requires a coherent decomposition that preserves the unknown eigenbasis, a signed readout that does not discard branch mass, and an analysis that distinguishes the mass reaching a test from the mass selected by it. The distinction is essential for R\'enyi entropy, because a power-sum bound can make a selected band rare without making its earlier tests equally rare. We retain both probabilities and derive the required variable-time routine from ordinary amplitude estimation and amplitude amplification. This analysis also delineates the scope of the result: a general density functional can still have a genuine multiplicative lower bound, as the explicit example in Proposition~\ref{prop:product-cost} shows.

\subsection{Main Results}
We state our results for additive entropy error $0<\varepsilon<1/10$, fixed order $\alpha>0$ with $\alpha\ne1$, and a known rank upper bound $R\ge2$. The ambient dimension is denoted by $d_\rho$, and we take $R\le d_\rho$ without loss of generality. We count forward, inverse, and controlled queries to the purification oracle. The notation $\widetilde O$ suppresses logarithmic factors in the parameters specified in Section~\ref{sec:setup}, with constants that may depend on the fixed order. We present the bounds at constant success probability; reducing the failure probability incurs logarithmic overhead. Rank-one states have zero entropy and require no queries.

\emph{Tsallis entropy.}
For $1<\alpha<2$, Theorem~\ref{thm:tsallis} improves the earlier rank-independent bound $O(\varepsilon^{-1-1/(\alpha-1)})$~\cite[Theorem~III.2]{LiuWang26} to $\widetilde O(\varepsilon^{-1/(\alpha-1)})$. The new precision exponent matches the sample-to-query lower bound when the dimension and rank are large enough to contain the hard family in \eqref{eq:hard-dimension}. Rank information gives the additional bound $\widetilde O(R^{2-\alpha}/\varepsilon)$, so we use the smaller of the rank-sensitive and rank-independent bounds. For every fixed real $\alpha\ge2$, our complexity is $\widetilde O(1/\varepsilon)$, which matches the rank-two precision lower bound. Controlled permutations and amplitude estimation already give this scaling at integer orders~\cite{Ekert02,BHMT02,WangMonomials25}; our result extends it to noninteger orders. We thereby resolve the rank-independent query-precision part of the improvement question raised by Liu and Wang~\cite[Section~I-D, problem~(iii)]{LiuWang26}. For $0<\alpha<1$, we obtain $\widetilde O((R/\varepsilon)^{1/\alpha})$. Table~\ref{tab:results} includes both the direct earlier estimators and the bounds obtained by converting R\'enyi estimators to Tsallis estimators.

\emph{Von Neumann entropy.}
Theorem~\ref{thm:vn} gives a query bound of $\widetilde O(R/\varepsilon)$, improving the $\widetilde O(R/\varepsilon^2)$ rank-sensitive bound of Wang \emph{et al.}~\cite{WangEntropies24}. We also compare with the dimension-dependent approach of Gily\'en and Li~\cite[Theorem~13]{GilyenLi20}, whose rank-sensitive refinement gives $\widetilde O(\sqrt{d_\rho R}/\varepsilon^{3/2})$~\cite[Section~I-A]{WangEntropies24}. The known lower bound is $\widetilde\Omega(R/\sqrt\varepsilon+1/\varepsilon)$: it matches our rank dependence at constant precision and our precision dependence at fixed rank, while leaving a gap in the joint dependence. We derive the large-rank statement from Wang's sample lower bound~\cite[Theorem~5.3]{WangFramework26} and the controlled-access lifting argument in Appendix~\ref{app:lower}; Proposition~\ref{prop:accuracy-lower} gives the fixed-rank precision bound.

\emph{R\'enyi entropy.}
Theorem~\ref{thm:renyi} gives $\widetilde O(R/\varepsilon)$ queries for $1<\alpha\le3$. At noninteger orders $1<\alpha<3$, this saves a factor $\varepsilon^{-1/\alpha}$ over the rank-based bound $\widetilde O(R/\varepsilon^{1+1/\alpha})$ of Wang, Zhang, and Li~\cite[Corollary~5]{WZL24}, and it matches the known linear rank dependence at constant precision. At orders two and three, it recovers the bounds obtained from direct moment estimators; we give the third-order construction in Proposition~\ref{prop:odd-moment}.

For $\alpha>3$, our new bound is $\widetilde O(R^{3/2-3/(2\alpha)}/\varepsilon)$. Comparing the displayed powers of rank and precision, with logarithmic factors suppressed, the new term is no larger than the preceding rank-based term when $\varepsilon\le R^{-(\alpha-3)/2}$; we select the smaller applicable bound. Under the additional spectral promise $\lambda_{\max}\le C/R$ for fixed $C>1$, Proposition~\ref{prop:upper-spectrum} gives $\widetilde O(R/\varepsilon)$ for every $\alpha>1$. Below order one, our bound is $\widetilde O((R/\varepsilon)^{1/\alpha})$, which improves the earlier rank-based bound by a factor $1/\varepsilon$; the dimension-based alternative of~\cite[Corollary~3]{WZL24} can still be smaller on full-rank states. Table~\ref{tab:promises} records these alternatives and separates the integer- and noninteger-order lower bounds.

\begin{table}[tbp]\centering\caption{Purified-query bounds for von Neumann and Tsallis entropies at additive error $\varepsilon$ and fixed order. Tildes suppress logarithmic factors; implicit constants may depend on the fixed order. Bold entries identify improved algorithmic bounds.}\label{tab:results}

\footnotesize\setlength{\tabcolsep}{3pt}\renewcommand{\arraystretch}{1.15}
\begin{tabular}{@{}L{59pt}C{151pt}C{99pt}C{141pt}@{}}
\hline\noalign{\vskip7pt}
Problem / conditions & Previous upper bounds & Our upper bounds & Lower bound\\\noalign{\vskip7pt}\hline\noalign{\vskip7pt}
von Neumann\newline Rank cap $R$
&\cell{$\displaystyle \Ot\!\left(R/\varepsilon^2\right)$~\cite{WangEntropies24}\tablegap$\displaystyle \Ot\!\left(\sqrt{d_\rho R}/\varepsilon^{3/2}\right)$~\cite{GilyenLi20,WangEntropies24}} &\cell{$\displaystyle \boldsymbol{\Ot\!\left(R/\varepsilon\right)}$} &\cell{$\displaystyle \OtOmega\!\left(R/\sqrt\varepsilon+1/\varepsilon\right)$~\cite{WangFramework26,Tang26}}\\\noalign{\vskip7pt}\hline\noalign{\vskip7pt}
Tsallis\newline $1<\alpha<2$\newline Rank cap $R$
&\cell{$\displaystyle O\!\left(\varepsilon^{-1-1/(\alpha-1)}\right)$~\cite{LiuWang26}\tablegap$\displaystyle \Ot\!\left(R/\varepsilon^{1+1/\alpha}\right)$~\cite{WZL24}} &\cell{$\displaystyle \boldsymbol{\Ot\!\left(R^{2-\alpha}/\varepsilon\right)}$\tablegap$\displaystyle \boldsymbol{\Ot\!\left(\varepsilon^{-1/(\alpha-1)}\right)}$} &\cell{$\displaystyle \Omega\!\left(1/\varepsilon\right)$ (Prop.~\ref{prop:accuracy-lower})}\\\noalign{\vskip7pt}\hline\noalign{\vskip7pt}
Tsallis\newline $1<\alpha<2$\newline Unrestricted rank
&\cell{$\displaystyle O\!\left(\varepsilon^{-1-1/(\alpha-1)}\right)$~\cite{LiuWang26}} &\cell{$\displaystyle \boldsymbol{\Ot\!\left(\varepsilon^{-1/(\alpha-1)}\right)}$} &\cell{$\displaystyle \OtOmega\!\left(\varepsilon^{-1/(\alpha-1)}\right)$~\cite{WangFramework26,Tang26}}\\\noalign{\vskip7pt}\hline\noalign{\vskip7pt}
Tsallis\newline Noninteger $\alpha>2$
&\cell{$\displaystyle O\!\left(\varepsilon^{-1-1/(\alpha-1)}\right)$~\cite{LiuWang26}\tablegap$\displaystyle \Ot\!\left(R/\varepsilon^{1+1/\alpha}\right)$~\cite{WZL24}} &\cell{$\displaystyle \boldsymbol{\Ot\!\left(1/\varepsilon\right)}$} &\cell{$\displaystyle \Omega\!\left(1/\varepsilon\right)$ (Prop.~\ref{prop:accuracy-lower})}\\\noalign{\vskip7pt}\hline\noalign{\vskip7pt}
Tsallis\newline Integer $\alpha\ge2$
&\cell{$\displaystyle O\!\left(1/\varepsilon\right)$~\cite{Ekert02,BHMT02,WangMonomials25}} &\cell{$\displaystyle \Ot\!\left(1/\varepsilon\right)$} &\cell{$\displaystyle \Omega\!\left(1/\varepsilon\right)$ (Prop.~\ref{prop:accuracy-lower})}\\\noalign{\vskip7pt}\hline\noalign{\vskip7pt}
\multicolumn{4}{@{}l}{\textit{Extension to orders below one}}\\[2pt]
Tsallis\newline $0<\alpha<1$\newline Rank cap $R$
&\cell{$\displaystyle \Ot\!\left(R^{(3-\alpha^2)/(2\alpha)}/\varepsilon^{(3+\alpha)/(2\alpha)}\right)$~\cite{WangEntropies24}\tablegap$\displaystyle \Ot\!\left(R^{2/\alpha-\alpha}/\varepsilon^{1+1/\alpha}\right)$~\cite{WZL24}\tablegap$\displaystyle\begin{gathered}
 \Ot\!\Bigl(d_\rho^{(1+\alpha)/(2\alpha)}\\[-1pt]
 \cdot R^{(1-\alpha)(1+1/(2\alpha))}/\varepsilon^{1+1/(2\alpha)}\Bigr)
 \end{gathered}$~\cite{WZL24}} &\cell{$\displaystyle \Ot\!\left((R/\varepsilon)^{1/\alpha}\right)$} &\cell{$\displaystyle\begin{gathered}
 \OtOmega\!\Bigl(R^{(1+1/\alpha)/2}/\varepsilon^{1/(2\alpha)}\\[-1pt]
 {}+R^{1-\alpha}/\varepsilon\Bigr)
 \end{gathered}$~\cite{WangFramework26,Tang26}}\\\noalign{\vskip7pt}\hline\noalign{\vskip7pt}
\end{tabular}\par\smallskip\begin{minipage}{.98\textwidth}\footnotesize Distinct upper-bound lines are alternatives, so their minimum is available. The first previous bound is rank-based for von Neumann and rank-free for Tsallis above one. Below one, the three previous bounds are direct, converted rank-based, and converted dimension-based, respectively. Bounds in the third column are Theorem~\ref{thm:vn} (von Neumann) and Theorem~\ref{thm:tsallis} (Tsallis).  Matching below-two rank-free lower bounds require dimension and allowed rank sufficient for \eqref{eq:hard-dimension}. Rank-dependent lower bounds use the large-rank, small-error ranges in Appendix~\ref{app:table-certificates}. That appendix also proves the R\'enyi-to-Tsallis conversions.\par\smallskip For noninteger Tsallis orders above one, the rank-free precision exponent is optimal (with the stated dimension condition below two). Integer-order $1/\varepsilon$ scaling was already known.\end{minipage}\end{table}

\begin{table}[tbp]\centering\caption{Purified-query bounds for R\'enyi entropy with rank at most $R$. Conventions are as in Table~\ref{tab:results}. The third column displays the multi-level algorithm alone.}\label{tab:promises}

\footnotesize\setlength{\tabcolsep}{3pt}\renewcommand{\arraystretch}{1.15}
\begin{tabular}{@{}L{59pt}C{151pt}C{99pt}C{141pt}@{}}
\hline\noalign{\vskip7pt}
Order / conditions & Comparison upper bounds & Our upper bound & Lower bound\\\noalign{\vskip7pt}\hline\noalign{\vskip7pt}
\multicolumn{4}{@{}l}{\textit{General rank promise}}\\[2pt]
Noninteger\newline $1<\alpha<3$
&\cell{$\displaystyle \Ot\!\left(R/\varepsilon^{1+1/\alpha}\right)$~\cite{WZL24}} &\cell{$\displaystyle \boldsymbol{\Ot\!\left(R/\varepsilon\right)}$} &\cell{$\displaystyle\begin{gathered}
 \OtOmega\!\Bigl(R/\varepsilon^{1/(2\alpha)}\\[-1pt]
 {}+R^{(1-1/\alpha)/2}/\varepsilon\Bigr)\end{gathered}$~\cite{WangFramework26,Tang26}}\\\noalign{\vskip7pt}\hline\noalign{\vskip7pt}
$\alpha=3$
&\cell{$\displaystyle O\!\left(R/\varepsilon\right)$ (Prop.~\ref{prop:odd-moment})} &\cell{$\displaystyle \Ot\!\left(R/\varepsilon\right)$} &\cell{$\displaystyle \Omega\!\left(R^{2/3}/\varepsilon^{1/3}+R^{1/3}/\varepsilon\right)$~\cite{Acharya20,BoundsList25}}\\\noalign{\vskip7pt}\hline\noalign{\vskip7pt}
$\alpha=2$
&\cell{$\displaystyle O\!\left(R/\varepsilon\right)$~\cite{Ekert02,BHMT02}} &\cell{$\displaystyle \Ot\!\left(R/\varepsilon\right)$} &\cell{$\displaystyle \Omega\!\left(R^{1/2}/\varepsilon^{1/2}+R^{1/4}/\varepsilon\right)$~\cite{Acharya20,BoundsList25}}\\\noalign{\vskip7pt}\hline\noalign{\vskip7pt}
Noninteger\newline $\alpha>3$
&\cell{$\displaystyle \Ot\!\left(R/\varepsilon^{1+1/\alpha}\right)$~\cite{WZL24}} &\cell{$\displaystyle \boldsymbol{\Ot\!\left(R^{3/2-3/(2\alpha)}/\varepsilon\right)}$} &\cell{$\displaystyle\begin{gathered}
 \OtOmega\!\Bigl(R/\varepsilon^{1/(2\alpha)}\\[-1pt]
 {}+R^{(1-1/\alpha)/2}/\varepsilon\Bigr)\end{gathered}$~\cite{WangFramework26,Tang26}}\\\noalign{\vskip7pt}\hline\noalign{\vskip7pt}
Integer\newline $\alpha\ge4$
&\cell{$\displaystyle \Ot\!\left(R/\varepsilon^{1+1/\alpha}\right)$~\cite{WZL24}\tablegap$\displaystyle O\!\left(R^{\alpha-1}/\varepsilon\right)$~\cite{Ekert02,BHMT02}} &\cell{$\displaystyle \boldsymbol{\Ot\!\left(R^{3/2-3/(2\alpha)}/\varepsilon\right)}$} &\cell{$\displaystyle\begin{gathered}
 \Omega\!\Bigl(R^{1-1/\alpha}/\varepsilon^{1/\alpha}\\[-1pt]
 {}+R^{(1-1/\alpha)/2}/\varepsilon\Bigr)
 \end{gathered}$~\cite{Acharya20,BoundsList25}}\\\noalign{\vskip7pt}\hline\noalign{\vskip7pt}
$0<\alpha<1$
&\cell{$\displaystyle \Ot\!\left(R^{1/\alpha}/\varepsilon^{1+1/\alpha}\right)$~\cite{WZL24}\tablegap$\displaystyle \Ot\!\left(d_\rho^{(1+\alpha)/(2\alpha)}/\varepsilon^{1+1/(2\alpha)}\right)$~\cite{WZL24}} &\cell{$\displaystyle \Ot\!\left((R/\varepsilon)^{1/\alpha}\right)$} &\cell{$\displaystyle\begin{gathered}
 \OtOmega\!\Bigl(R^{(1+1/\alpha)/2}/\varepsilon^{1/(2\alpha)}\\[-1pt]
 {}+R^{(1/\alpha-1)/2}/\varepsilon\Bigr)\end{gathered}$~\cite{WangFramework26,Tang26}}\\\noalign{\vskip7pt}\hline\noalign{\vskip7pt}
\multicolumn{4}{@{}l}{\textit{Additional promise: $\lambda_{\max}\le C/R$, fixed $C>1$}}\\[2pt]
$\alpha>1$
&\cell{$\displaystyle \Ot\!\left(R/\varepsilon^{1+1/\alpha}\right)$~\cite{WZL24}\tablegap Integer $\alpha$: $\displaystyle O\!\left(R^{\alpha-1}/\varepsilon\right)$~\cite{Ekert02,BHMT02}} &\cell{$\displaystyle\Ot(R/\varepsilon)$} &\cell{$\displaystyle\Omega(1/\varepsilon)$ (App.~\ref{app:spectral-promise-lower})}\\\noalign{\vskip7pt}\hline\noalign{\vskip7pt}
\end{tabular}\par\smallskip\begin{minipage}{.98\textwidth}\footnotesize Odd integers $k\ge5$ also admit $O(R^{(k-1)/2}/\varepsilon)$ by Proposition~\ref{prop:odd-moment}; this is dominated by the displayed multi-level term. The third-order comparison follows from the direct block-product construction in Proposition~\ref{prop:odd-moment}. Selecting the smaller applicable bound across the second and third columns gives the combined upper bound.  For $\alpha>3$, comparison of the displayed powers, with logarithmic factors suppressed, makes the new term no larger than the previous $\Ot(R/\varepsilon^{1+1/\alpha})$ term when $\varepsilon\le R^{-(\alpha-3)/2}$.  At order two this bound was already available without the promise. For $C=1$, the spectrum and all three entropies are fixed, so zero queries suffice. General-rank upper bounds are Theorem~\ref{thm:renyi}, and the promised bound is Proposition~\ref{prop:upper-spectrum}. Lower-bound ranges are in Appendix~\ref{app:table-certificates}; the promise-preserving family is in Appendix~\ref{app:spectral-promise-lower}.\par\smallskip Below one, the two previous bounds are rank-based and dimension-based, respectively. In the additional-promise row, the $1/\varepsilon$ lower bound uses the promise-preserving family with even $R$.\end{minipage}\end{table}

The common framework estimates effectively specified spectral functionals whose local rescalings admit bounded polynomial compilation. It includes logarithmic and fixed-power profiles. Theorem~\ref{thm:master} combines coherent decomposition, local compilation, joint readout, and error allocation, with all approximation and implementation errors included. This framework explains the entropy results above; their different tail bounds and final readouts are developed in Section~\ref{sec:applications}. The comparison tables retain earlier algorithms whenever their parameter dependence is better. In particular, below-one improvements over a rank-based baseline need not improve a dimension-based estimator on full-rank states.

\subsection{Technical Overview}
We organize the construction around the weighted spectral functional
\begin{equation}\label{eq:functional-intro}
 F_g(\rho)=\Tr[\rho g(\rho)]=\sum_{i:\lambda_i>0}\lambda_i g(\lambda_i).
\end{equation}
Here the eigenvalues $\lambda_i$ describe the mathematical target; our algorithm neither learns the eigenbasis nor receives the individual eigenvalues as classical data. We instead apply coherent spectral tests at dyadic scales $\delta_\ell$. Each test records whether an eigenvalue is above or below a gapped threshold, and only the paths that have not previously stopped undergo the next test. The first-stop records are orthogonal, so their actual response functions decompose \eqref{eq:functional-intro} exactly, even where neighboring tests have overlapping transition intervals. Approximation affects the localization of a branch, rather than the completeness of the decomposition.

For a contribution $A_\ell$ at scale $\delta_\ell$, we choose a known envelope $G_\ell$ for the local magnitude of $g$ and approximate the normalized profile $g/G_\ell$. If the contribution is allowed error $e_\ell$, its normalized readout needs error only $e_\ell/G_\ell$. We implement this profile as a globally bounded polynomial with local processing cost $C_\ell=\widetilde O(1/\delta_\ell)$. A Hadamard test then reads the encoded block on the unnormalized first-stop branch. Its joint probabilities include both the event of stopping at this level and the chosen Hadamard outcome, so the signal retains the branch mass. The fixed block-encoding normalization, $\beta=2$, enters the recovery formula explicitly.

The statistical analysis uses two probability bounds because a selected band and the computation leading to it need not have the same mass. We write $S_\ell$ for a known upper bound on the mass reaching level $\ell$, and $B_\ell\le S_\ell$ for a bound on the mass selected there. Ordinary amplitude estimation can exploit $B_\ell$, but each call still executes the complete level preparation. Variable-time estimation additionally incorporates the probabilities of reaching the preceding stages. For the basic rank profile, we may use the exact envelopes $S_\ell=B_\ell=\min\{1,8R\delta_\ell\}$, which give the following three query upper bounds:
\begin{equation}\label{eq:intro-costs}
\begin{aligned}
 Q_\ell^{\mathrm{basic}}&=\widetilde O(C_\ell G_\ell/e_\ell),\\
 Q_\ell^{\mathrm{fixed}}&=\widetilde O(C_\ell\sqrt{B_\ell}G_\ell/e_\ell),\\
 Q_\ell^{\mathrm{rank}}&=\widetilde O(C_\ell B_\ell G_\ell/e_\ell).
\end{aligned}
\end{equation}
The basic fixed-time estimate in the first line ignores probability information, the second uses the selected-band bound, and the third also uses the rank-based survival profile. For the logarithmic profile, $G_\ell$ grows only logarithmically as the scale decreases, whereas $C_\ell B_\ell=\widetilde O(R)$ on the fine levels. This is the cancellation responsible for our von Neumann bound.

A moment bound for R\'enyi entropy can give a selected-band cap $B_\ell$ substantially smaller than $S_\ell$. Under the same dyadic survival profile, the corresponding coefficient improves from $C_\ell G_\ell S_\ell$ to $C_\ell G_\ell\sqrt{S_\ell B_\ell}$. The remaining factor $\sqrt{S_\ell}$ accounts for earlier tests, which must act on every path that reaches them, including paths that eventually leave the target band. We prove the general stopping-time bound before deriving this dyadic simplification in Section~\ref{sec:statistics}.

To implement the variable-time readout when the physical signal can vanish, we add an orthogonal branch with a known success probability. We estimate the augmented probability, subtract the known contribution, and recover each amplification gain through ordinary amplitude estimation. The analysis retains the exact sine factors from amplitude amplification and gives a finite execution limit for every classical transcript. The algorithm uses known probability bounds and estimates of intermediate gains; their true values enter the proof of correctness and cost.

Once the local costs are known, the remaining choice is how to distribute the total statistical error. If level $\ell$ has query upper bound $\widetilde O(a_\ell/e_\ell)$, we assign $e_\ell$ proportional to $\sqrt{a_\ell}$; Cauchy--Schwarz shows that this minimizes the sum of these bounds under a fixed total tolerance. We then choose the deterministic approximation errors to fit the same final budget. Powers require a cutoff-tail bound, the logarithm requires leakage weighted by its singular profile, and R\'enyi entropy requires relative moment accuracy before taking a logarithm. Figure~\ref{fig:framework} summarizes the resulting replacement of a global product by scale-dependent processing and readout.

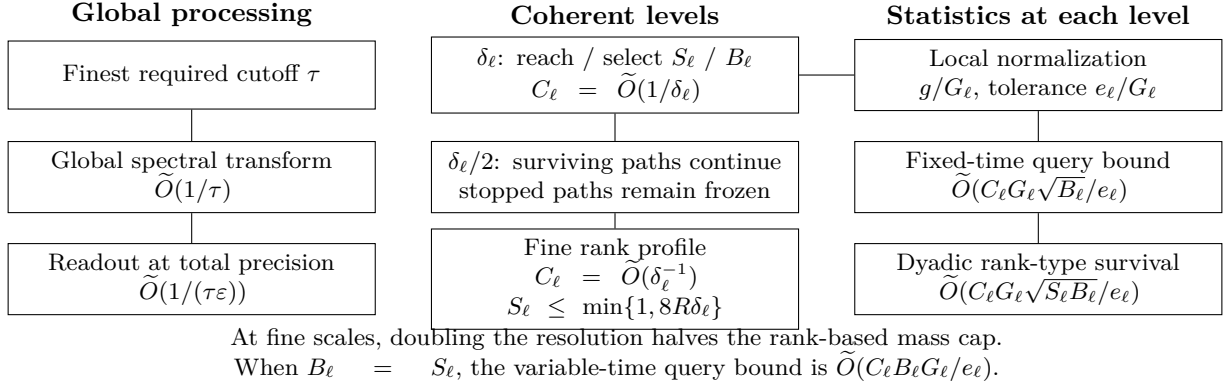
\begin{figure}[tbp]
\centering
\begin{adjustbox}{max width=\linewidth}
\begin{tikzpicture}[x=1cm,y=1cm,>=stealth,font=\footnotesize,
box/.style={draw,align=center,text width=4.6cm,minimum height=.9cm},
head/.style={align=center,text width=4.8cm,font=\small\bfseries}]
\node[head] at (2.4,.8) {Global processing};
\node[head] at (8,.8) {Coherent levels};
\node[head] at (13.6,.8) {Statistics at each level};
\node[box] (g1) at (2.4,0) {Finest required cutoff $\tau$};
\node[box] (g2) at (2.4,-1.35) {Global spectral transform\\ $\widetilde O(1/\tau)$};
\node[box] (g3) at (2.4,-2.7) {Readout at total precision\\ $\widetilde O(1/(\tau\varepsilon))$};
\draw (g1)--(g2); \draw (g2)--(g3);
\node[box] (m1) at (8,0) {$\delta_\ell$: reach / select $S_\ell$ / $B_\ell$\\ $C_\ell=\widetilde O(1/\delta_\ell)$};
\node[box] (m2) at (8,-1.35) {$\delta_\ell/2$: surviving paths continue\\ stopped paths remain frozen};
\node[box] (m3) at (8,-2.7) {Fine rank profile\\ $C_\ell=\widetilde O(\delta_\ell^{-1})$\\ $S_\ell\le\min\{1,8R\delta_\ell\}$};
\draw (m1)--(m2); \draw (m2)--(m3);
\node[box] (s1) at (13.6,0) {Local normalization\\ $g/G_\ell$, tolerance $e_\ell/G_\ell$};
\node[box] (s2) at (13.6,-1.35) {Fixed-time query bound\\ $\widetilde O(C_\ell G_\ell\sqrt{B_\ell}/e_\ell)$};
\node[box] (s3) at (13.6,-2.7) {Dyadic rank-type survival\\ $\widetilde O(C_\ell G_\ell\sqrt{S_\ell B_\ell}/e_\ell)$};
\draw (s1)--(s2); \draw (s2)--(s3);
\draw (m1.east)--(s1.west);
\node[align=center,text width=16cm] at (8,-3.7) {At fine scales, doubling the resolution halves the rank-based mass cap.\\ When $B_\ell=S_\ell$, the variable-time query bound is $\widetilde O(C_\ell B_\ell G_\ell/e_\ell)$.};
\end{tikzpicture}
\end{adjustbox}
\caption{Multi-level allocation of processing and statistical work. The middle panel shows processing costs and certified mass envelopes. On fine scales with $8R\delta_\ell\le1$, the rank-based envelope decreases linearly with $\delta_\ell$, while processing has inverse-scale cost up to logarithmic factors. The refined variable-time query bound assumes the dyadic rank-type survival bound proved in Section~\ref{sec:statistics}.}\label{fig:framework}
\end{figure}

\subsection{Related Work}
\paragraph{Classical entropy and distribution functionals.}
The difficulty of estimating entropy near zero is already present in the classical problem introduced by Shannon~\cite{Shannon48}. Consistency and bias analyses show why empirical frequencies alone can give poor estimates of nonlinear functionals~\cite{Antos01,Paninski03,Paninski04}. Work on entropy under different forms of distribution access~\cite{BatuEntropy05}, unseen distributions~\cite{Valiant17}, and polynomial approximation~\cite{Jiao15,Jiao17,WuYang16} developed ways to treat rare events differently from well-observed ones. The classical R\'enyi problem also distinguishes integer and noninteger orders~\cite{Acharya17}. We share the principle of adapting the approximation to a probability scale, but our resource is coherent access to a quantum preparation, and the spectral basis is not supplied.

\paragraph{Quantum queries to classical distributions.}
Quantum distribution testing~\cite{Bravyi11,Belovs19,MontanaroWolf16} and entropy estimation~\cite{LiWu19,ShinJeong25} exploit coherent probability access. Wang, Zhang, and Li~\cite{WZL24} combine QSVT, relative moment estimation, and variable-time methods; their Section~VI explicitly treats quantum-state reductions, including dimension- and rank-dependent bounds. We use those quantum-state results in our comparison. Chen \emph{et al.}~\cite{Multilevel26} subsequently formulate multi-level estimation using nondestructive discrimination and level-dependent approximation. Their Theorem~13 is the closest general functional-estimation result to our construction.

The access distinction matters quantitatively. The labeled distribution encoding in~\cite{Multilevel26} has singular values proportional to $\sqrt{p_i}$, whereas the swap construction for an unknown quantum state gives the block $\rho$ and therefore the eigenvalues $\lambda_i$. This changes the cost of resolving a probability scale. Their theorem also uses bounded, parity-definite polynomials through squared amplitudes. We instead retain the complete first-stop history and estimate a signed, unnormalized matrix element of a local density-operator transform. Our statistical bound separately retains the mass reaching a spectral test and the mass selected by it. The distinction is essential when a moment promise reduces the latter without reducing the earlier tests. Table~\ref{tab:nearest} compares these assumptions and the resulting readouts; Appendix~\ref{app:adjacent-comparison} gives a same-model two-classifier comparison.

\paragraph{Spectra, moments, and coherent entropy estimation.}
Quantum spectrum estimation and tomography provide general procedures for recovering spectral information~\cite{KeylWerner01,ODonnellWright16,HaahTomography17,ODonnellSpectrum21}. Direct estimation can be substantially more economical when only one functional is needed. Controlled permutations yield integer moments~\cite{Ekert02}; subsequent constructions reduce the number of simultaneously available copies or estimate several moments together~\cite{ShinTrace25,Simultaneous25,ShiMoments25}. Coherent entropy algorithms use purified access~\cite{GilyenLi20,Subramanian21}, spectral or entropy promises~\cite{Gur21,WangZhao23}, and low-rank structure~\cite{WangEntropies24,WZL24}. These assumptions lead to different dependence on dimension, rank, precision, and the smallest eigenvalue, so they cannot be interchanged in a complexity comparison.

Liu and Wang~\cite{LiuWang26} give effectively computable uniform approximations for positive powers and obtain dimension-independent precision bounds. Our rank-independent improvement builds on this progress by reducing the remaining precision cost for noninteger powers. Fixed integer moments already admit optimal $1/\varepsilon$ precision scaling, while growing-order dependence is studied in~\cite{WangMonomials25,WangHighOrder26}. Broader spectral-sum algorithms combine polynomial transformations with trace estimation under their respective matrix-access assumptions~\cite{SpectralSums20}. We develop the density-spectrum construction and local cost analysis needed for the entropy bounds considered here.

\begin{table}[tbp]\centering\footnotesize
\caption{Assumptions and guarantees of the closest frameworks and entropy costs of our two readout routes. The extra-qubit statement excludes the input encoding.}\label{tab:nearest}
\setlength{\tabcolsep}{4pt}\renewcommand{\arraystretch}{1.3}
\begin{tabular}{@{}L{49pt}L{132pt}L{118pt}L{145pt}@{}}\hline
Feature & Chen et al., Theorem 13 & WZL, Section VI & This work\\\hline
Encoding & $\sqrt{p_i}/2$, labeled distribution & State reductions already included & $\rho$ block; arbitrary purification\\[3pt]
Transform & Bounded parity-definite square polynomials & QSVT and relative moments & Bounded local profile; signed readout\\[3pt]
Statistics & Levelwise amplitude estimation & Annealing and variable time & Separate $S_j,B_\ell$; calibrated gains\\[3pt]
Resources & Local counts and four extra qubits & Rank/dimension alternatives & Query bounds; retained-history space\\[3pt]\hline\end{tabular}\par\smallskip
\begin{tabular}{@{}L{164pt}C{132pt}C{156pt}@{}}\hline
Entropy task, same purified access & Probability-sensitive fixed time & Variable time\\\hline
von Neumann & $\Ot(R/\varepsilon^{3/2})$ & $\Ot(R/\varepsilon)$\\[3pt]
R\'enyi, $1<\alpha<3/2$ & $\Ot(R/\varepsilon^{3/(2\alpha)})$ & $\Ot(R/\varepsilon)$\\[3pt]
R\'enyi, $3/2\le\alpha\le3$ & $\Ot(R/\varepsilon)$ & $\Ot(R/\varepsilon)$\\[3pt]
R\'enyi, $\alpha>3$ & $\Ot(R^{3/2-3/(2\alpha)}/\varepsilon)$ & $\Ot(R^{3/2-3/(2\alpha)}/\varepsilon)$\\[3pt]\hline
\end{tabular}
\par\smallskip\begin{minipage}{.97\textwidth}\footnotesize The lower panel gives upper bounds derived here with controlled purified access, rank at most $R\ge2$, fixed order, additive entropy error $\varepsilon$, and constant confidence; polylogarithms are suppressed. Proposition~\ref{prop:fixed-entropy} proves the fixed-time column for the first-stop prefixes of Section~\ref{sec:master}. Proposition~\ref{prop:adjacent-bands} separately gives a two-classifier construction with the same soft query powers; the comparison suppresses logarithmic factors. Rank-free Tsallis powers also follow from fixed time; the remaining Tsallis and below-one R\'enyi comparisons are in Proposition~\ref{prop:fixed-entropy}. Boundary orders may add level logarithms, and constants are not uniform in the order.\end{minipage}
\end{table}

\paragraph{Copies and restricted measurements.}
Sample-based entropy estimation asks how many independent copies of an unknown state are necessary~\cite{Acharya20}. The samplizer of~\cite[Theorem~III.1]{Samplizer25} converts block-encoded query procedures into copy-based algorithms while accounting for their simulation error. Further power estimators and lower bounds~\cite{ChenWang25,ChenLiuWang26}, recent von Neumann and R\'enyi/Tsallis estimators~\cite{GaoWang26,ChenWang26}, and estimators for nonlinear state functions~\cite{Yao26} improve different parts of this sample-complexity landscape. Restrictions on quantum memory, communication, or the allowed measurement further change the problem~\cite{Purity24,Anshu22,Song26}. These results describe complementary statistical models. Two-state functionals, such as quantum Tsallis relative entropy~\cite{BaoRelative25}, further involve the joint dependence on two input operators; our theorem treats a function of a single density operator.

\paragraph{Approximation, statistics, and lower bounds.}
QSP and QSVT provide the bounded transformations used in our local compiler~\cite{LowChuang17,QSVT19,LowChuang19}, with effective phase synthesis supplied by~\cite{Haah19}. Amplitude estimation~\cite{BHMT02}, quantum Monte Carlo, and variance-sensitive mean estimation~\cite{Montanaro15,HamoudiMagniez19,Hamoudi21} show how the statistical cost can depend on more than a worst-case outcome range. Variable-time amplification and estimation~\cite{Ambainis12,CGJ19}, including developments in variable-time search~\cite{AKV23}, similarly use the probabilities of executing expensive stages. Their role in linear-system algorithms illustrates why state preparation, spectral processing, and amplification should be accounted for separately~\cite{CKS17,LowSu26}. We use these established principles, but derive the specialized calibration from ordinary amplitude estimation so that its gains, reversals, and finite execution limits are explicit.

Our lower-bound comparisons draw on hybrid, polynomial, and adversary arguments~\cite{BBBV97,Beals01,HoyerLeeSpalek07,Reichardt11,Bun20}. Sample-to-query lifting~\cite{Lifting25,Tang26,BoundsList25}, together with sample-based simulation~\cite{LMR14,Kimmel17,Go25}, transfers suitable sample lower bounds to the purified model. The recent framework of~\cite{WangFramework26} supplies the entropy sample bounds used here. Appendix~\ref{app:table-certificates} records each transferred statement, its order and precision range, and the required purification and control conventions. We use these lower bounds to identify the parameter ranges in which our precision exponent is optimal, while leaving unmatched joint rank--precision dependence as an open question.

\subsection{Organization}
Section~\ref{sec:setup} fixes the estimation target, query model, and standard primitives. Section~\ref{sec:framework} constructs the multi-level estimator and proves its cost formula. Section~\ref{sec:applications} derives the entropy bounds and their logarithmic-moment consequences, and Section~\ref{sec:discussion} discusses their scope and the remaining questions. Appendices~\ref{app:instrument}--\ref{app:budgets} supply the detailed spectral, statistical, and implementation proofs; Appendix~\ref{app:entropy} proves the entropy scale sums, and Appendix~\ref{app:lower} gives the accuracy lower bounds, lifting argument, and comparison certificates.

\section{Preliminaries and Query Model}\label{sec:setup}
\subsection{Spectral Functionals and Error Criteria}
We consider a density operator $\rho$ on a $d_\rho$-dimensional system $S$. Its positive eigenvalues are $\lambda_1,\ldots,\lambda_r$, with corresponding orthonormal eigenvectors $\ket{u_i}$, so that $\rho\ket{u_i}=\lambda_i\ket{u_i}$ and $\sum_{i=1}^r\lambda_i=1$. We use natural logarithms throughout. The three entropy functionals, listed in the order of our applications, are
\begin{subequations}\label{eq:entropy-defs}
\begin{align}
 T_\alpha(\rho)&:=\frac{1-\Tr(\rho^\alpha)}{\alpha-1},\label{eq:def-tsallis}\\
 S(\rho)&:=-\Tr(\rho\log\rho),\label{eq:def-vn}\\
 S_\alpha^{\mathrm{R}}(\rho)&:=\frac{\log\Tr(\rho^\alpha)}{1-\alpha}.\label{eq:def-renyi}
\end{align}
\end{subequations}
Tsallis and R\'enyi orders are fixed $\alpha>0$, $\alpha\ne1$; von Neumann entropy has no order parameter. The convention $0\log0=0$ defines the trace in \eqref{eq:def-vn} continuously at zero. The common auxiliary quantity for the power-based applications is the moment
\begin{equation}\label{eq:moment-def}
 M_\alpha(\rho):=\Tr(\rho^\alpha)=\sum_{i=1}^r\lambda_i^\alpha.
\end{equation}
The functional framework uses two scalar profiles:
\begin{align}\label{eq:entropy-profiles}
 g_\alpha(x)&:=x^{\alpha-1},&F_{g_\alpha}(\rho)&=M_\alpha(\rho),\\
 g_{\log}(x)&:=-\log x,&F_{g_{\log}}(\rho)&=S(\rho).\nonumber
\end{align}
Tsallis estimation applies the affine map $(1-M_\alpha)/(\alpha-1)$; R\'enyi estimation applies $\log M_\alpha/(1-\alpha)$; von Neumann estimation directly returns the logarithmic functional. Scalar variables $x,y$ are used for approximation, and operator functions are defined by spectral calculus on the positive support. Zero eigenvalues contribute zero to every weighted functional. When the input $\rho$ is fixed, we abbreviate $M_\alpha(\rho)$ and $F_g(\rho)$ as $M_\alpha$ and $F_g$.

For a functional tolerance $\varepsilon_F\in(0,1)$ and failure probability $\nu\in(0,1/3)$, the goal is to output $\widehat F$ satisfying
\begin{equation}\label{eq:target}
 \Pr[\abs{\widehat F-F_g(\rho)}\le\varepsilon_F]\ge1-\nu.
\end{equation}
We distinguish this tolerance from the final entropy error $\eps$, since the conversion depends on the entropy. A known $R\ge r$ is used only in rank-sensitive statements. No eigenvalue lower bound is assumed.

The moment ranges determine the precision needed before the final R\'enyi logarithm. Convexity and concavity give
\begin{equation}\label{eq:moment-ranges}
 \begin{cases}1\le M_\alpha\le R^{1-\alpha},&0<\alpha<1,\\
 R^{1-\alpha}\le M_\alpha\le1,&\alpha>1.
 \end{cases}
\end{equation}
For $R=1$ all three entropies vanish, so the algorithm returns zero. Subsequent rank-dependent constructions take $R\ge2$.

We use $O$, $\Omega$, and $\Theta$ with constants that may depend on the function family and on parameters declared fixed in the relevant statement, including the orders $\alpha$ or $k$ and the spectral-promise constant $C$. Dependence on varying parameters is displayed explicitly, except for logarithmic factors suppressed by a tilde. These factors may involve inverse error, inverse cutoff, inverse failure probability, and the rank when present; their exponents may also depend on the fixed parameters. Local statements specify the behavior near order boundaries.

\subsection{Purified Query Access and Block Encodings}
To apply spectral transformations, we first represent the density operator as a block of a unitary. A unitary $U$ is a $(\beta,a,\eta)$ block encoding of $A$ when
\begin{equation}
 \norm{A-\beta(\bra{0^a}\otimes I)U(\ket{0^a}\otimes I)}\le\eta.
\end{equation}
Our readout keeps the normalization $\beta$ explicit; Appendix~\ref{app:compiler} fixes $\beta=2$ for the local compiler.

The input oracle $U_\rho$ prepares $\ket{\psi_\rho}_{SE}$ from all-zero registers, with $\Tr_E\ket{\psi_\rho}\bra{\psi_\rho}=\rho$. Each call to $U_\rho$, $U_\rho^\dagger$, or a controlled version counts as one query. Access to controlled versions is part of the model. We measure complexity by the number of these oracle calls. On fresh preparation registers $A_1E_1$ and a target register $S$, the swap construction
\begin{equation}\label{eq:block}
 W_\rho=(U_\rho^\dagger\otimes I_S)
       (I_{E_1}\otimes\mathrm{SWAP}_{A_1,S})(U_\rho\otimes I_S)
\end{equation}
has top-left block $\rho_S$ and costs two queries~\cite{GilyenLi20}. Tensor factors in \eqref{eq:block} are ordered $E_1,A_1,S$; $U_\rho$ acts on $E_1A_1$. Fresh preparation registers are distinct from the purification being processed. Contracting the fresh purification $\sum_i\sqrt{\lambda_i}\ket{a_i}_{E_1}\ket{u_i}_{A_1}$ on either side of SWAP gives $\sum_i\lambda_i\ket{u_i}\bra{u_i}=\rho$ on $S$, proving the block identity.

Correctness is required for every purification and unitary completion consistent with $\rho$; the construction uses no information about that completion. We allow a bounded number of controls and compute a nested predicate reversibly into one flag. Throughout the computation, the original environment $E$ remains untouched.

The clean-input convention is also important. Zero auxiliary inputs specify the block of a complete module, not a condition that may be reimposed before each internal oracle call. A QSVT sequence reuses its signal registers through $W_\rho,W_\rho^\dagger$ and phase projectors without intermediate measurement or reset. We therefore give complete simulation and local-transform modules their own zero workspaces and retain their outputs together with the classifier histories until the complete inverse. Ideal Hamiltonian simulation returns its auxiliaries to zero on the clean-input subspace, whereas a general contraction block need not. In either case, the inverse reverses the full prefix, including its flags and module workspaces. Appendix~\ref{app:budgets} follows this convention when counting slots and assigning synthesis precision.

Table~\ref{tab:registers} in Appendix~\ref{app:instrument} records the workspace and its reversal.

\subsection{Algorithmic Primitives}
We use three standard algorithmic ingredients. First, polynomial eigenvalue transformation implements a real degree-$d$ polynomial bounded on $[-1,1]$ with $O(d)$ block-encoding queries and a known constant normalization~\cite{QSVT19}. Polynomial approximation error and finite synthesis error are tracked separately. Second, additive amplitude estimation estimates a specified success probability to error $h$ using $O(h^{-1}\log(1/\nu_{\mathrm{call}}))$ calls to its preparation circuit and inverse, with failure at most $\nu_{\mathrm{call}}$~\cite{BHMT02}. Coherent flags define success even for a subnormalized branch.

Third, variable-time estimation organizes reversible stages that freeze stopped paths~\cite{Ambainis12,CGJ19}. Appendix~\ref{app:vt} gives the complete specialized calibration used here. It uses ordinary amplitude estimation to choose conservative amplification counts, tracks the exact gains, and recovers the original probability. Here $p$ is the flagged success probability, $T_{\max}$ the longest execution time, and $T_2$ the root-mean-square stopping time. A known lower bound $p_*>0$ supplies its probability scale.

With a known scale $p_*\le p\le2p_*$, this construction gives query bound $\Ot(\gamma^{-1}[T_{\max}+(T_2+1)/\sqrt p])$ for relative precision $\gamma$. The appendix derives this bound from the actual stage costs and places a deterministic cap on every execution, so the guarantee does not rely on postselecting a favorable running time.

\section{Multi-Level Estimation of Quantum Spectral Functionals}\label{sec:framework}
\subsection{Framework Overview and Admissible Profiles}\label{sec:local-condition}
We now turn the preceding cost comparison into an estimator for $F_g$. The input consists of a function profile, a desired accuracy, and computable bounds on the contributions that may be omitted or misclassified. These bounds determine which spectral scales must be resolved. At each retained scale we then keep four quantities separate: the processing cost $C_\ell$, the local function magnitude $G_\ell$, the mass $S_\ell$ that can reach the test, and the mass $B_\ell$ selected by it. We begin with the approximation condition that makes the local processing cost an inverse spectral scale. The coherent decomposition and readout developed next will use the other three quantities to prove Theorem~\ref{thm:master}.

We choose a cutoff $\tau\in(0,1/8]$ to control the omitted small eigenvalues and a classifier parameter $\xi\in(0,1/2)$ to control misclassification. Our dyadic scales are
\begin{equation}\label{eq:scales}
 \delta_\ell=2^{-\ell-1},\quad 0\le\ell\le L,\qquad
 \tau/2<\delta_L\le\tau.
\end{equation}
We associate level $\ell$ with the interval $I_\ell=[\delta_\ell/2,2\delta_\ell]$. Neighboring intervals overlap so that an eigenvalue near a threshold may contribute to either band; the construction does not require a discontinuous spectral test.

\begin{definition}[Uniformly locally analytic profiles]\label{def:analytic}
A real profile $g$ and positive envelopes $G_\ell$ are admissible if
\begin{equation}\label{eq:profile}
 f_\ell(y)=g(4\delta_\ell y)/G_\ell
\end{equation}
extend analytically to a common complex neighborhood of $J=[1/32,3/4]$ and are uniformly bounded there. If necessary, $g$ is continued beyond one solely to define these auxiliary profiles; the target functional uses only its values on $(0,1]$.
\end{definition}
We require an effective description of the profile: its scale parameters, envelopes, and uniform analytic bounds are known, and the coefficients of the bounded approximants can be computed to any requested finite precision with classical work polynomial in their degree and precision bit length. This preprocessing produces the circuit descriptions whose oracle use we count. The input also includes computable bounds on the omitted tail and off-band leakage; the applications below derive these from rank, normalization, or a known moment scale.

Enlarging all envelopes by a fixed constant ensures that the profiles have magnitude at most $1/4$ on this neighborhood. This slack will enforce global polynomial boundedness. Powers and logarithms satisfy the definition with
\begin{equation}\label{eq:envelopes}
 G_\ell=\begin{cases}
 C_\alpha\delta_\ell^{\alpha-1},&g(x)=x^{\alpha-1},\\
 C_H[1+\log(1/\delta_\ell)],&g(x)=\log(1/x).
 \end{cases}
\end{equation}
The rescaled power is independent of $\ell$; the rescaled logarithm is a bounded constant plus a bounded multiple of $\log y$. In contrast, rescaling $\sin(1/x)$ leaves increasingly rapid oscillations and does not meet this uniform condition. Uniform regularity after rescaling permits approximation with polylogarithmic degree.

For $0<\alpha<1$, the power profile has exponent in $(-1,0)$. Choosing the principal analytic branch on the same neighborhood of $J$, which stays away from zero and the negative real axis, gives the same uniform local analyticity and effective representation. Its envelope grows at fine scales; this growth will be balanced against the rank tail in Section~\ref{sec:applications}.

\subsection{Coherent Spectral Decomposition}\label{sec:first-stop}
We separate the spectral scales without measuring the system eigenstate. At scale $\delta_j$, a coherent classifier records a stop or continue decision in a fresh register and retains its work registers. We write $s_j(x)$ for its stop probability on an eigenstate of eigenvalue $x$. Gapped phase estimation~\cite{Cleve98}, implemented with block-Hamiltonian simulation~\cite{LowChuang19}, gives
\begin{equation}\label{eq:classifier}
 \begin{aligned}
 s_j(x)&\le\xi^2,&&\text{when }x\le\delta_j/2,\\
 1-s_j(x)&\le\xi^2,&&\text{when }x\ge\delta_j,
 \end{aligned}
\end{equation}
at cost $\Ot(1/\delta_j)$, with logarithmic dependence on $1/\xi$. The ideal spectral classifier preserves $\ket{u_i}$; finite synthesis error is accounted for separately in Appendix~\ref{app:budgets}.

We use sans-serif symbols for Boolean flags and initialize $\mathsf r_{-1}=1$. At each level $j=0,\ldots,L$, the classifier writes its high-eigenvalue bit $h_j$. We then compute the first-stop and continuation predicates into fresh zero-initialized targets:
\begin{equation}\label{eq:reversible-flags}
 \mathsf{f}_j=\mathsf{r}_{j-1}h_j,\qquad
 \mathsf{r}_j=\mathsf{r}_{j-1}(1-h_j).
\end{equation}
Because the flags have values in $\{0,1\}$, the products in \eqref{eq:reversible-flags} specify their Boolean values. We compute them with Toffoli gates, using a negative control for $1-h_j$, into fresh targets while retaining both inputs and all phase-estimation records. The flag $\mathsf r_{j-1}$ controls whether classifier $j$ acts, whereas $\mathsf f_j$ identifies the branch that stops there for the first time. Retaining these records makes the branches orthogonal and allows the complete computation to be reversed, as proved in Appendix~\ref{app:instrument}.

Because each later classifier is controlled on the continuation flag, it acts as the identity on a branch that has already stopped. We denote by $K_\ell$ the map from the input system into the first-stop branch at level $\ell$, including its work and flag registers, and by $K_\perp$ the map into the branch that continues through all levels. Unitarity of the complete circuit implies
\begin{equation}\label{eq:complete}
 \sum_{\ell=0}^L K_\ell^\dagger K_\ell+K_\perp^\dagger K_\perp=I.
\end{equation}
These maps take values in a larger space including flags and work registers; their output norms encode the branch probabilities. Their scalar responses are
\begin{equation}\label{eq:responses}
 \begin{aligned}
 m_\ell(x)&=s_\ell(x)\prod_{j<\ell}[1-s_j(x)],\\
 m_\perp(x)&=\prod_{j=0}^L[1-s_j(x)].
 \end{aligned}
\end{equation}
For every eigenvalue $x$, these responses sum to one. They are the actual probabilities of the ideal spectral circuit, including its transition regions, and therefore define the weights in an exact decomposition rather than an approximate partition of the spectrum.

The unnormalized state of the selected branch is $\ket{\zeta_\ell}=(K_\ell\otimes I_E)\ket{\psi_\rho}$. Its probability mass and contribution to the functional are, respectively,
\begin{equation}\label{eq:layer-exact}
 w_\ell=\sum_i\lambda_i m_\ell(\lambda_i),\qquad
 A_\ell=\sum_i\lambda_i m_\ell(\lambda_i)g(\lambda_i).
\end{equation}
Completeness gives $F_g=\sum_\ell A_\ell+A_\perp$, where $A_\perp$ is defined by the same expression with $m_\perp$ in place of $m_\ell$. This identity is the starting point for the deterministic error analysis.

\begin{lemma}[Localization, survival, and band mass]\label{lem:mass}
Outside $I_\ell$, $m_\ell(x)\le\xi^2$. For $r\le R$, the probability of reaching stage $\ell$ and its first-stop mass $w_\ell$ are bounded by
\begin{equation}\label{eq:B}
 S_\ell=\min\{1,8R\delta_\ell\},\qquad B_\ell=S_\ell,
\end{equation}
provided $\xi^2\le R\delta_L$. Also $m_\perp(x)\le\xi^2$ for $x\ge\delta_L$. A sharper known bound $w_\ell\le B_\ell\le S_\ell$ may be substituted for the band cap while retaining $S_j$ for all prefix costs.
\end{lemma}
The localization bound follows because an off-band eigenvalue must either trigger a test below its high-eigenvalue range or survive the preceding high-eigenvalue test. For the rank bound, we split the spectrum at $2\delta_\ell$: the eigenvalues below this threshold carry mass at most $2R\delta_\ell$, while those above it reach the level with probability at most $\xi^2$. The same argument bounds the selected mass. Appendix~\ref{app:instrument} supplies the endpoint cases and constants. When no rank information is used, we set $S_\ell=B_\ell=1$ instead.

\subsection{Local Function Approximation and Readout}\label{sec:compiler}
Our local transform must approximate $g/G_\ell$ on the selected interval while remaining bounded on the full QSVT domain $[-1,1]$. We meet these requirements in two stages: a bounded polynomial rescales the selected interval to a fixed interval away from zero, and a second bounded polynomial approximates the normalized profile there.

\begin{lemma}[Bounded local compiler]\label{lem:compiler}
Under Definition~\ref{def:analytic}, for every $0<\eta_\ell<1/10$ there is a constant-normalized block encoding $U_{\ell,g}$ of a Hermitian contraction $q_\ell(\rho)$ such that
\begin{equation}\label{eq:compiler}
 \sup_{x\in I_\ell}\abs{q_\ell(x)-g(x)/G_\ell}\le\eta_\ell,
 \qquad C_\ell^{\mathrm{loc}}=\Ot(1/\delta_\ell).
\end{equation}
Here one may take the known normalization $\beta=2$, and precision dependence in $C_\ell^{\mathrm{loc}}$ is polylogarithmic.
\end{lemma}
The construction in Appendix~\ref{app:compiler} maps $x$ to approximately $x/(4\delta_\ell)$ and compiles \eqref{eq:profile} on a fixed interval away from zero. Its two stages multiply their degrees, but only the rescaling degree contains an inverse power of $\delta_\ell$.

We read the local transform through a Hadamard test controlled on the first-stop flag. Its joint probabilities $p_{\ell,+}$ and $p_{\ell,-}$ include both selection of that branch and the Hadamard outcome. We do not postselect the block-encoding work registers, since the desired signal is a matrix element rather than its squared magnitude. On the retained branch, we define the unnormalized expectation
\begin{equation}\label{eq:readout-expectation}
 \mathcal R_\ell:=\bra{\zeta_\ell}
 [q_\ell(\rho_S)\otimes I_{\mathrm{rest}}]\ket{\zeta_\ell}.
\end{equation}
Here $I_{\mathrm{rest}}$ acts on the environment and classification records. The joint probabilities then satisfy
\begin{align}
 p_{\ell,\pm}&=\frac12\left(w_\ell\pm
 \frac{\mathcal R_\ell}{\beta}\right),\label{eq:hadamard}\\
 \widetilde A_\ell&:=G_\ell\mathcal R_\ell
 =\beta G_\ell(p_{\ell,+}-p_{\ell,-}).\label{eq:readout}
\end{align}
In subsequent operator expressions, identities on untouched registers are implicit. In particular, $p_{\ell,+}+p_{\ell,-}=w_\ell$ and each probability is at most $B_\ell$. The joint probabilities retain the branch mass within the observed signal.

The statistical estimator recovers the contribution of $q_\ell$, so we must still compare that quantity with the contribution of the target function $g$. Within $I_\ell$, the local approximation guarantee controls this difference. Outside the interval, we bound the magnitudes of both terms and define the resulting functional leakage by
\begin{equation}\label{eq:leakage}
 \mathcal L_{g,\ell}=\sum_{\lambda_i\notin I_\ell}
 \lambda_i m_\ell(\lambda_i)\bigl(\abs{g(\lambda_i)}+G_\ell\bigr).
\end{equation}
The local deterministic error then obeys
\begin{equation}\label{eq:local-error}
 \abs{A_\ell-\widetilde A_\ell}
 \le G_\ell\eta_\ell B_\ell+\mathcal L_{g,\ell}.
\end{equation}
For a singular function, this weighted leakage incorporates both the off-band probability mass and the function magnitude.

For a concrete illustration, consider spectrum $(1/2,1/4,1/4)$ and $g(x)=\sqrt{x}$. An idealized branch collecting the two eigenvalues $1/4$ has mass $w=1/2$ and contribution $A=1/4$. Ignoring approximation error, its readout obeys $p_++p_-=1/2$ and $\beta G(p_+-p_-)=1/4$. Thus the band mass and its functional contribution are different quantities, both retained by the joint readout. The actual overlapping classifier replaces this illustrative partition by its exact response weights.

\subsection{Fixed-Time and Variable-Time Estimation}\label{sec:statistics}
We write $T_\ell^{\mathrm{cls}}$ for the cumulative classifier query cost through level $\ell$ and $C_\ell^{\mathrm{loc}}$ for the cost of its local transform. Including the initial preparation, the complete unamplified real branch has cost $C_\ell=1+T_\ell^{\mathrm{cls}}+C_\ell^{\mathrm{loc}}$. Both processing terms are $\Ot(\delta_\ell^{-1})$, because the classifier costs form a geometric sequence.

\paragraph{One complete band execution.}
For a fixed level $\ell$ and Hadamard sign, we prepare $\ket{\psi_\rho}$ and execute classifiers $0,\ldots,\ell$ with the flag updates in \eqref{eq:reversible-flags}. A branch that stops before $\ell$ is permanently marked as failed. The flag $\mathsf f_\ell$ controls the local transform and Hadamard test; a branch that is still running after this classifier also fails. The success event is the conjunction of $\mathsf f_\ell=1$ and the selected Hadamard outcome. We exclude the transform auxiliaries from this predicate, since the Hadamard interference already extracts the matrix element with normalization $\beta=2$. This complete preparation is the circuit used by fixed-time estimation.

Variable-time estimation needs a known positive probability scale even when the real signal vanishes. A fair coherent selector $d$ therefore chooses between the real preparation and an auxiliary, or reference, branch. The latter prepares a coin $b$ with success probability $B_\ell$ and stops immediately. The augmented success event is reference success or real success, on orthogonal selector values; its probabilities consequently add. After estimating that probability, we subtract the known reference contribution. Each inverse reverses the full preparation, including phase records and flag gates. Reflections thus use known zero inputs or explicitly computed success flags.

\begin{figure}[tbp]\centering
\begin{adjustbox}{max width=\linewidth}
\begin{tikzpicture}[x=1cm,y=1cm,>=stealth,font=\footnotesize,
gate/.style={draw,fill=white,align=center,minimum height=.65cm},
ann/.style={align=center,text width=7.6cm}]
\node[anchor=east] at (1.5,0.0) {$d,b$};\draw (1.65,0.0)--(16,0.0);
\node[anchor=east] at (1.5,-0.95) {$P_j,h_j,\mathsf r_j,\mathsf f_j$};\draw (1.65,-0.95)--(16,-0.95);
\node[anchor=east] at (1.5,-1.9) {$S,E$};\draw (1.65,-1.9)--(16,-1.9);
\node[anchor=east] at (1.5,-2.85) {$Z_\ell,H$};\draw (1.65,-2.85)--(16,-2.85);
\node[gate,text width=2.1cm,minimum height=3.6cm] (prep) at (3.05,-1.425) {Coherent setup\\ $d=0$: $\ket{\psi_\rho}$\\ $d=1$: $B_\ell$ coin};
\node[gate,text width=3cm,minimum height=1.65cm] (cls) at (6.65,-1.425) {Classifier prefix $0,\ldots,\ell$\\ fresh history and flags};
\draw (6.65,0)--(cls.north); \draw[fill=white] (6.65,0) circle (1.8pt);
\node[above] at (6.65,.05) {$d=0$};
\node[gate,text width=2.5cm,minimum height=1.6cm] (local) at (10.25,-2.375) {$q_\ell/\beta$, $\beta=2$\\ Hadamard readout};
\fill (10.25,-.95) circle (1.8pt); \draw (10.25,-.95)--(local.north);
\node[above] at (10.25,-.85) {$\mathsf f_\ell=1$};
\node[gate,text width=2.3cm,minimum height=3.6cm] (event) at (14,-1.425) {Reversible success flag\\[5pt] Reference branch $d=1$:\\ success when $b=1$\\[5pt] Physical branch $d=0$:\\ success when $\mathsf f_\ell=1$\\ and $H=h_\pm$};
\node[ann] at (4.4,-4.65) {Earlier stops and paths running after the target are failures. Reference paths stop after setup.};
\node[ann] at (12.5,-4.65) {$Z_\ell$ is not postselected; all histories are retained.\\ $\bar p=(B_\ell+p_{\ell,\pm})/2$};
\node[draw,align=center,text width=15.5cm] (cal) at (8.45,-6.0) {Classical calibration calls the complete coherent preparation and its inverse;\\ recover $p_{\ell,\pm}=2\bar p-B_\ell$ and combine the two signs.};
\end{tikzpicture}
\end{adjustbox}
\caption{Variable-time module circuit for one selected level and Hadamard sign. Fixed time omits the known-success branch and calibration; reference-branch failure is excluded at the first stage-end projection. The wires show coherent controlled operations, not intermediate branch measurements. The prefix executes a classifier only on a still-running real path; every inverse reverses its records, Boolean updates, and clean-input preparation. Here $h_+=0$ and $h_-=1$. The environment $E$ is untouched. The final success flag is the projection used by the statistical estimator.}\label{fig:band-flow}
\end{figure}
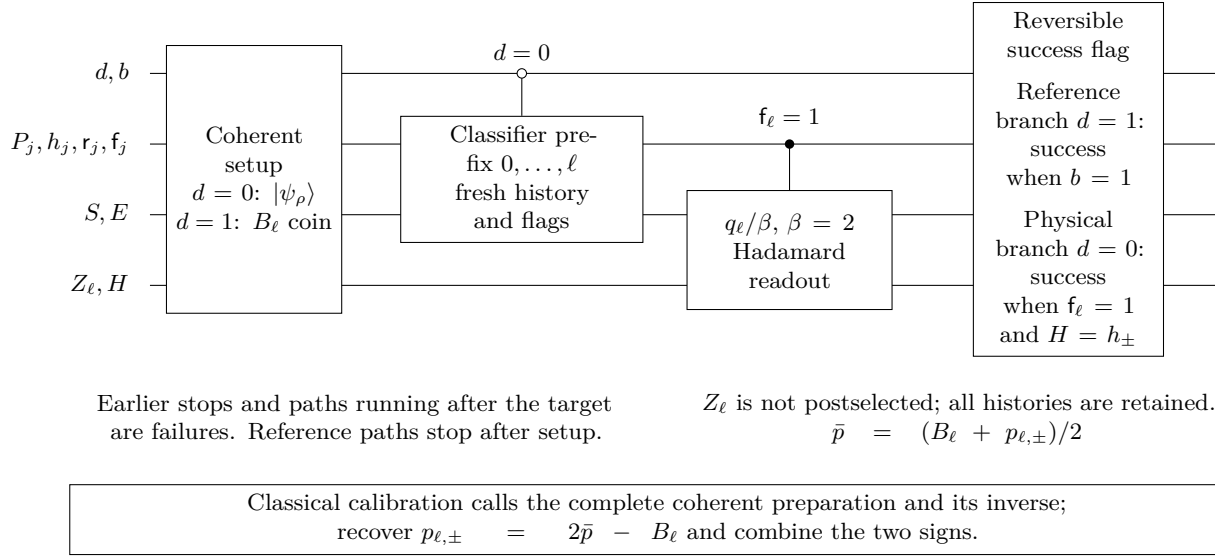

\subsubsection{Probability-sensitive fixed-time estimation}
\begin{proposition}[Fixed-time estimation of one band]\label{prop:fixed-band}
Suppose complete preparation and inverse circuits are available, with known bounds $p_{\ell,\pm}\le B_\ell\le1$ and $0<e_\ell<B_\ell G_\ell$. For each sign, the following integer parameters define a fixed-time estimator. Here $d_0\ge1$ is the input-preparation cost and $0<\zeta_\ell<1/4$ is the confidence parameter.
\begin{align}\label{eq:fixed-cap}
 a_\ell^{\mathrm{AE}}&=e_\ell/(4\beta G_\ell),\quad
 C_\ell=d_0+\sum_{j=0}^\ell c_j+C_\ell^{\mathrm{loc}},\nonumber\\
 M_\ell&=\left\lceil\frac{4\pi\sqrt{B_\ell}}{a_\ell^{\mathrm{AE}}}
              +\frac{2\pi}{\sqrt{a_\ell^{\mathrm{AE}}}}\right\rceil,\nonumber\\
 r_\ell^{\mathrm{FT}}&=2\lceil4\log(4/\zeta_\ell)\rceil+1,\nonumber\\
 K_{\ell,\pm}^{\mathrm{FT}}&=8r_\ell^{\mathrm{FT}}(M_\ell+1)C_\ell.
\end{align}
Ordinary amplitude estimation~\cite{BHMT02}, followed by a median and projection onto $[0,B_\ell]$, estimates the signed contribution within $e_\ell$ with probability at least $1-2\zeta_\ell$, including the finite-synthesis and arithmetic allowances in Appendix~\ref{app:budgets}. The deterministic query bound is $K_{\ell,+}^{\mathrm{FT}}+K_{\ell,-}^{\mathrm{FT}}$. When $B_\ell G_\ell\le e_\ell$, the estimator returns zero without calling either circuit.
\end{proposition}
\begin{proof}
For success probability $p$, amplitude estimation has error at most $2\pi\sqrt p/M+\pi^2/M^2$ with probability at least $8/\pi^2$~\cite[Theorem~12]{BHMT02}. Our choice of $M_\ell$ makes this error at most $3a_\ell^{\mathrm{AE}}/4$. Applying Hoeffding's inequality~\cite{Hoeffding63} to the independent repetitions bounds the median failure probability by $\zeta_\ell/4$. After multiplication by $\beta G_\ell$, the two signed errors sum to at most $3e_\ell/8$, leaving the stated arithmetic allowance. Each repetition uses at most $8(M_\ell+1)$ preparations or inverses, including the initial and controlled calls, and the synthesis transfer adds failure probability at most $\zeta_\ell/4$ per sign. This proof also covers $p=0$, because the additive amplitude-estimation guarantee does not require a positive success probability.
\end{proof}
Before suppressing logarithms, the cost is $O(C_\ell[G_\ell\sqrt{B_\ell}/e_\ell+\sqrt{G_\ell/e_\ell}]\log(1/\zeta_\ell))$. On active layers $B_\ell G_\ell>e_\ell$, the first term dominates the second. Thus
\begin{equation}\label{eq:fixed-cost}
 Q_\ell^{\mathrm{fixed}}=\Ot\!\left(\frac{G_\ell\sqrt{B_\ell}}{e_\ell\delta_\ell}\right).
\end{equation}
Setting $B_\ell=1$ recovers the probability-independent bound. A supplied rank or moment cap can already help fixed-time estimation; variable time provides the additional survival-mass saving below.

\subsubsection{Variable-time estimation with separate mass bounds}
\paragraph{Survival and selection have different masses.}
The mass reaching a test need not be comparable to the mass selected by it. To see the distinction, we consider an ideal sharp classifier and a dyadic scale $0<\delta<1/4$. A spectrum consisting of $1/2$, $\delta$, and $m$ copies of $(1/2-\delta)/m<\delta/2$ has remaining mass $1/2$ when the scale-$\delta$ test is reached, because the eigenvalue $1/2$ has already stopped. Only the eigenvalue $\delta$ stops at this test, so the selected mass is $\delta$. Our implemented classifiers have transition regions and a small error rather than sharp thresholds, but they retain this distinction. Consequently, an improved upper bound $B_\ell$ for the selected mass does not by itself improve the earlier survival bounds $S_j$ in \eqref{eq:T2-proof}.

For a selected sign, we freeze branches stopping before $\ell$ and branches continuing beyond it as failures. Only the first-stop-$\ell$ branch executes the local transform and Hadamard test. We denote the maximum and root-mean-square query times of this real preparation by $T_{\max,\ell}$ and $T_{2,\ell}$. Because the cumulative classifier cost is $\Ot(\delta_j^{-1})$ and the final local cost is $C_\ell^{\mathrm{loc}}$, we can use the known bounds
\begin{equation}\label{eq:times}
 \begin{split}
 T_{\max,\ell}&=\Ot(\delta_\ell^{-1})+C_\ell^{\mathrm{loc}},\\
 T_{2,\ell}^2&\le\Ot\!\left(1+\sum_{j\le\ell}
 S_j\delta_j^{-2}\right)+O(B_\ell (C_\ell^{\mathrm{loc}})^2).
 \end{split}
\end{equation}
We use these computable upper bounds throughout. Appendix~\ref{app:vt} proves them by charging each segment to the probability of reaching it.

The physical success probability $p=p_{\ell,\pm}$ may vanish, whereas relative estimation requires a known positive scale. We supply that scale by preparing a fair control coin: one branch executes the real circuit, and the other has the known success probability $B_\ell$ and stops immediately. The augmented success probability is therefore
\begin{equation}\label{eq:reference}
 \bar p=(p+B_\ell)/2,\qquad B_\ell/2\le\bar p\le B_\ell.
\end{equation}
Its constant setup cost is included in \eqref{eq:times}. The calibrated estimator of Theorem~\ref{thm:calibrated} estimates $\bar p$ to relative error $\gamma$ using
\begin{equation}\label{eq:vt-primitive}
 Q_\ell^{\mathrm{rel}}(\gamma)=
 \Ot\!\left[\frac1\gamma\left(T_{\max,\ell}
 +\frac{T_{2,\ell}+1}{\sqrt{\bar p}}\right)\right]
\end{equation}
queries, including prefix construction and gain recovery. Recovering $p=2\bar p-B_\ell$ gives additive error at most $2\gamma B_\ell$. Taking $\gamma=e_\ell/(8\beta B_\ell G_\ell)$ for each sign therefore yields
\begin{equation}\label{eq:general-vt-cost}
 Q_\ell^{\mathrm{VT}}=\Ot\!\left[\frac{G_\ell}{e_\ell}
 \left(B_\ell T_{\max,\ell}
 +\sqrt{B_\ell}(T_{2,\ell}+1)\right)\right].
\end{equation}
This expression separates the probability of the final signal from the cost of reaching it.

For the rank profile \eqref{eq:B}, dyadic summation gives $T_{2,\ell}=\Ot(1+\sqrt{S_\ell}/\delta_\ell)$ when $C_\ell^{\mathrm{loc}}=\Ot(1/\delta_\ell)$ and $B_\ell\le S_\ell$. Hence
\begin{equation}\label{eq:rank-cost}
 Q_\ell^{\mathrm{VT}}=\Ot\!\left(
 \frac{G_\ell\sqrt{S_\ell B_\ell}}{e_\ell\delta_\ell}\right).
\end{equation}
The basic rank bound sets $B_\ell=S_\ell$ and recovers $\Ot(B_\ell G_\ell/(e_\ell\delta_\ell))$. A smaller terminal cap improves only the factor $\sqrt{B_\ell}$; the surviving mass still determines the preceding work.

To instantiate this cost, let $d_0\ge1$, $c_j$, and $C_\ell^{\mathrm{loc}}$ be integer query bounds for setup, classifier stages, and local readout, respectively, with controlled calls counted and the same counts for their inverses. Controlled input purification permits $d_0=1$; known coins use no input queries. The following weighted cost bound charges the first classifier in full and retains reference-success mass in all later stages:
\begin{equation}\label{eq:cal-band-H}
 \overline{\mathcal H}_\ell
 =d_0+c_0+\sum_{j=1}^\ell c_j\sqrt{(S_j+B_\ell)/2}
   +C_\ell^{\mathrm{loc}}\sqrt{B_\ell}.
\end{equation}
\begin{proposition}[Calibrated estimation of one band]\label{prop:band-interface}
Suppose the complete band circuits, their integer segment costs, and valid bounds $S_j,B_\ell$ are given, with $0<e_\ell<B_\ell G_\ell$ and $0<\zeta_\ell<1/4$. One call per sign to the calibration procedure of Appendix~\ref{app:calibration-call}, with the following parameters, gives a band estimator:
\begin{equation}\label{eq:band-interface-cap}
 \begin{gathered}
 m_\ell=\ell+2,\quad p_{*,\ell}=B_\ell/2,\quad
 \gamma_\ell=\frac{e_\ell}{8\beta B_\ell G_\ell},\\
 Q_{\mathrm{band}}\le K_{\ell,+}+K_{\ell,-},\\
 K_{\ell,\pm}=K_{\mathrm{call}}(m_\ell,p_{*,\ell},
 \overline{\mathcal H}_\ell,\gamma_\ell,\zeta_\ell).
 \end{gathered}
\end{equation}
The $\ell+1$ classifiers and final readout form the stages, and \eqref{eq:cal-band-H} includes the setup cost. All calibration parameters are computed from \eqref{eq:cal-target} and \eqref{eq:cal-counter-data}, without access to unknown branch probabilities. We project $2\widehat{\bar p}_{\ell,\pm}-B_\ell$ onto $[0,B_\ell]$ before forming the signed combination. With probability at least $1-2\zeta_\ell$, the resulting estimate differs from $\widetilde A_\ell$ by at most $e_\ell$, including the arithmetic allowance. The query bound holds for every transcript. When $B_\ell G_\ell\le e_\ell$, including $B_\ell=0$, the estimator returns zero without calibration.
\end{proposition}
Appendix~\ref{app:vt} proves this guarantee, and \eqref{eq:cal-counter-data} gives the explicit integer expression for $K_{\mathrm{call}}$. We invoke the estimator only on active bands; the local approximation and leakage of skipped bands remain in the global error budget.

\subsection{Complete Algorithm and Complexity}\label{sec:master}
The inputs are controlled purified access, an effective admissible profile $g$, error and confidence targets, and bounds on the tail, leakage, and relevant probability masses. A known spectral upper bound can also remove bands with zero sharp-cutoff contribution. Their actual classifier leakage is included in the deterministic error budget.
We begin by choosing the cutoff and classifier accuracy from the tail and leakage bounds. The selected fixed-time or variable-time estimator determines the coefficients $a_\ell$ in \eqref{eq:coefficients}; these coefficients in turn determine the statistical tolerances in \eqref{eq:allocation}. We choose the local approximation accuracies so that their errors fit the remaining budget in \eqref{eq:budget}. All of these choices depend on supplied bounds and the function description, not on an observed spectrum.

For a retained level with $B_\ell G_\ell\le e_\ell$, the zero estimate already meets its statistical tolerance, so we do not prepare its circuit. At every other level, we construct the coherent prefix and bounded local transform and estimate the two joint probabilities using Proposition~\ref{prop:fixed-band} or~\ref{prop:band-interface}. The classical estimate of that level is $\widehat A_\ell=\beta G_\ell(\widehat p_{\ell,+}-\widehat p_{\ell,-})$. We return their sum, with rounding charged to the reserved arithmetic tolerance. Each statistical call uses a fresh preparation; the combination across levels is entirely classical.
Each complete functional call uses one of the two statistical estimators. Figure~\ref{fig:band-flow} gives the variable-time circuit; fixed-time estimation uses the same real branch without the known-success branch or gain calibration.

Theorem~\ref{thm:master} combines the localization and mass bounds of Lemma~\ref{lem:mass} with the approximation guarantee of Lemma~\ref{lem:compiler}. The band estimators then provide the required statistical accuracy and a finite query bound. Appendix~\ref{app:budgets} combines these errors and transfers the analysis to finite-precision circuits. In particular, the tail and leakage quantities below are bounded before execution; their spectral definitions are used to prove those bounds, not as additional oracle inputs.

\begin{theorem}[Multi-level functional estimation]\label{thm:master}
Suppose controlled purified access and an effectively specified admissible profile are given. We retain a set $\mathcal J$ of levels and use $\mathcal T_g$ as an upper bound for the sum of the absolute terminal contribution and the absolute contributions of all omitted levels. The approximation, leakage, and statistical tolerances are required to satisfy
\begin{equation}\label{eq:budget}
 \mathcal T_g+\sum_{\ell\in\mathcal J}
 (\mathcal L_{g,\ell}+G_\ell\eta_\ell B_\ell+e_\ell)
 \le3\varepsilon_F/4.
\end{equation}
For either statistical estimator, the input includes valid mass bounds $w_\ell\le B_\ell\le1$; the choice $B_\ell=1$ is always available. The local-transform cost is $C_\ell^{\mathrm{loc}}=\Ot(\delta_\ell^{-1})$. The variable-time estimator additionally requires the stopping-time bounds in \eqref{eq:times} and uses positive mass bounds on the levels that it actually prepares. We define the cost coefficients by
\begin{equation}\label{eq:coefficients}
 a_\ell=\begin{cases}
 G_\ell\sqrt{B_\ell}/\delta_\ell,&\text{fixed-time},\\
 G_\ell[B_\ell T_{\max,\ell}
 +\sqrt{B_\ell}(T_{2,\ell}+1)],&\text{variable-time}.
 \end{cases}
\end{equation}
Here the times may be replaced by known upper bounds. With the rank profile, $\xi^2\le R\delta_L$, and $B_\ell\le S_\ell$, one may instead use $a_\ell=G_\ell\sqrt{S_\ell B_\ell}/\delta_\ell$, up to logarithmic factors.

Suppose a total statistical tolerance $E=\Theta(\varepsilon_F)$ is feasible in \eqref{eq:budget}, and compilation and bound evaluation have polylogarithmic precision overhead in the displayed parameters. We allocate this tolerance over $\mathcal J_+=\{\ell\in\mathcal J:a_\ell>0\}$. A level with zero coefficient has zero selected mass and is assigned $e_\ell=\widehat A_\ell=0$ without preparation. If $\mathcal J_+$ is empty, the estimator returns zero and only the tail and leakage terms in \eqref{eq:budget} remain. Otherwise, our allocation is
\begin{equation}\label{eq:allocation}
 e_\ell=E\frac{\sqrt{a_\ell}}{\sum_{j\in\mathcal J_+}\sqrt{a_j}}.
\end{equation}
The bounded implementation of Appendix~\ref{app:budgets} satisfies \eqref{eq:target} with
\begin{equation}\label{eq:master-cost}
 Q=\Ot\!\left(1+
 \frac{(\sum_{\ell\in\mathcal J}\sqrt{a_\ell})^2}{\varepsilon_F}\right).
\end{equation}
For an empty set or when all retained layers are skipped, the output is zero before rounding and no oracle calls are needed.
\end{theorem}
\begin{proof}
We separate the error argument from the cost calculation. On simultaneous statistical success, the exact decomposition of $F_g$ and \eqref{eq:local-error} bound the total ideal error by the left-hand side of \eqref{eq:budget}. This argument also covers omitted preparations: a zero coefficient gives $B_\ell=0$ and hence $A_\ell=\widetilde A_\ell=0$, whereas a skipped nonzero level satisfies $|\widetilde A_\ell|\le B_\ell G_\ell\le e_\ell$.

For an active level, \eqref{eq:fixed-cost} or \eqref{eq:general-vt-cost} gives the query upper bound $\Ot(a_\ell/e_\ell)$. For nonempty $\mathcal J_+$, under the constraint $\sum_{\ell\in\mathcal J_+} e_\ell=E$, Cauchy--Schwarz yields $(\sum_{\ell\in\mathcal J_+}\sqrt{a_\ell})^2\le(\sum_{\ell\in\mathcal J_+} a_\ell/e_\ell)E$. Equality holds for \eqref{eq:allocation}, and substituting that allocation into the sum of the level bounds proves \eqref{eq:master-cost}. Finally, Appendix~\ref{app:budgets} allocates the failure probabilities and bounds every execution. Classical rounding uses the remaining value-error allowance, while finite synthesis is charged to the failure budget by comparing the complete output distributions.
\end{proof}

The significance of \eqref{eq:master-cost} is the order in which approximation and estimation are optimized. A global construction first chooses an operator approximation accurate enough everywhere and then repeats it to obtain a scalar estimate. Here $G_\ell$, $B_\ell$, and the prefix probabilities enter the cost before we select the error allocation. Thus the finest spectral scale need not be processed at the precision demanded by the entire functional. The local magnitudes and mass bounds determine when this sum improves the global cost. Rapidly oscillating families can retain a genuine degree--precision product, as Proposition~\ref{prop:product-cost} illustrates. The theorem provides an implementable criterion for calculating these local costs.

\section{Applications to Quantum Entropy Estimation}\label{sec:applications}
We now choose envelopes, cutoffs, and statistical tolerances for each entropy. All statements use $0<\varepsilon<1/10$, $0<\nu<1/3$, and fixed order; confidence costs are logarithmic. For each application we first choose the functional tolerance, then control the deterministic bias, and finally sum the layer costs. The number of levels is logarithmic in the problem parameters, so uniform layer tolerances suffice for the stated soft bounds; the square-root allocation can improve the unsuppressed total. The application bounds are proved in Appendix~\ref{app:entropy}.

Table~\ref{tab:entropy-parameters} summarizes the parameters used in the three applications. We write $u_T=\min\{1/20,|\alpha-1|\varepsilon/4\}$, $u_<= (1-\alpha)\varepsilon/4$, and $h_f=\min\{1,\alpha-1\}\varepsilon/256$ for the required moment tolerances. The rank bounds are $B_\ell^R=S_\ell^R=\min\{1,8R\delta_\ell\}$, while $B_\ell^s$ denotes the moment-dependent bound in \eqref{eq:moment-caps}. Our cutoff functions are $t_0(u)=\min\{1/8,(u/128)^{1/(\alpha-1)}\}$ and $t_R(u)=\min\{1/8,(u/(128R))^{1/\alpha}\}$. Each row uses the envelopes in \eqref{eq:envelopes} and the error allocation in \eqref{eq:application-recipe}; none of these choices requires the unknown spectrum. We assign separate failure budgets to the moment-scale pilot, its final readout, and the individual band calls as detailed below and in Appendix~\ref{app:budgets}.

\begin{table}[tbp]\centering\footnotesize
\caption{Entropy instance parameters; projection precedes entropy conversion.}\label{tab:entropy-parameters}
\setlength{\tabcolsep}{3pt}\renewcommand{\arraystretch}{1.35}
\begin{tabular}{@{}L{80pt}C{55pt}L{103pt}C{60pt}L{146pt}@{}}\hline
Instance / route & Profile & Tolerance / cutoff & Mass bounds & Output / conditions\\\hline
Tsallis, rank-free & $x^{\alpha-1}$ & $u_T$, $t_0(u_T)$ & $S_j=B_\ell=1$ & $(1-\widehat M)/(\alpha-1)$; $\alpha>1$\\[3pt]
Tsallis, rank-based & $x^{\alpha-1}$ & $u_T$, $t_R(u_T)$ & $S_j^R,B_\ell^R$ & $(1-\widehat M)/(\alpha-1)$; $\alpha>0$, $\alpha\ne1$\\[3pt]
von Neumann & $\log(1/x)$ & $u=\varepsilon$, \eqref{eq:log-cutoff} & $S_j^R,B_\ell^R$ & $\widehat F\in[0,\log R]$\\[3pt]
R\'enyi, $0<\alpha<1$ & $x^{\alpha-1}$ & $u_<$, $t_R(u_<)$ & $S_j^R,B_\ell^R$ & $\log\widehat M/(1-\alpha)$;\newline $\widehat M\in[1,R^{1-\alpha}]$\\[3pt]
R\'enyi, pilot / final & $x^{\alpha-1}$ & $u=hs$, $t_R(hs)$ & $S_j^R,B_\ell^s$ & $h=1/16$ / $h_f$;\newline $\alpha>1$, $M_\alpha\le s\le1$\\[3pt]
$\lambda_{\max}\le C/R$ & $x^{\alpha-1}$ & $u=h_fR^{1-\alpha}$, $t_R(u)$ & $S_j^R,B_\ell^R$ & $\delta_\ell\le2C/R$; $\alpha>1$, $C>1$\\[3pt]
\hline\end{tabular}
\par\smallskip\begin{minipage}{.97\textwidth}\footnotesize The defaults use fixed time without rank and Proposition~\ref{prop:band-interface} otherwise. All rows also admit Proposition~\ref{prop:fixed-band}; its costs are in Proposition~\ref{prop:fixed-entropy}. For the pilot/final row, we retain $\delta_\ell\le2s^{1/\alpha}$ and charge every earlier classifier. In the last row, we project to $[R^{1-\alpha},1]$ before taking the logarithm; $C=1$ and $R=1$ are the zero-query cases. Constants may depend on fixed $\alpha,C$; the bounds are not uniform as $\alpha$ approaches one.\end{minipage}\end{table}

\subsection{Tsallis Entropy Estimation}
\begin{theorem}[Tsallis entropy]\label{thm:tsallis}
For a state of rank at most $R$, controlled purified queries suffice to estimate $T_\alpha$ to additive error $\varepsilon$ with failure at most $\nu$, at cost
\begin{equation}\label{eq:tsallis-bound}
 Q_{T_\alpha}=
 \begin{cases}
 \Ot((R/\varepsilon)^{1/\alpha}),&0<\alpha<1,\\
 \Ot(\min\{R^{2-\alpha}/\varepsilon,
 \varepsilon^{-1/(\alpha-1)}\}),&1<\alpha<2,\\
 \Ot(1/\varepsilon),&\alpha\ge2.
 \end{cases}
\end{equation}
The rank-free alternative for $1<\alpha<2$ and the bound for $\alpha\ge2$ require no rank input, including in their hidden logarithms.
\end{theorem}
\begin{proof}
We estimate $M_\alpha$ with the profile $g(x)=x^{\alpha-1}$ and envelope $G_\ell=C_\alpha\delta_\ell^{\alpha-1}$. The moment tolerance $u=\min\{1/20,|\alpha-1|\varepsilon/4\}$ leaves error at most $\varepsilon/4$ after the affine Tsallis conversion. For $\alpha>1$, we may choose between two bounds on the omitted tail:
\begin{equation}\label{eq:power-tails}
 \sum_{0<\lambda_i<t}\lambda_i^\alpha
 \le\min\{t^{\alpha-1},Rt^\alpha\}.
\end{equation}
The first tail bound leads to the rank-independent cutoff $t_0=\min\{1/8,(u/128)^{1/(\alpha-1)}\}$, while the second gives the rank-dependent cutoff $t_R=\min\{1/8,(u/(128R))^{1/\alpha}\}$. In each case, we use the first dyadic scale at or below the selected cutoff. With $S_\ell=B_\ell=\min\{1,8R\delta_\ell\}$, the resulting cost coefficients satisfy
\begin{equation}\label{eq:power-coeff}
 a_\ell=\Theta\!\left(
 \min\{\delta_\ell^{\alpha-2},R\delta_\ell^{\alpha-1}\}\right).
\end{equation}
Theorem~\ref{thm:master} bounds the cost of layer $\ell$ by $\widetilde O(a_\ell/e_\ell)$. We choose $n=L+1$ and $e_\ell=u/(8n)$, so the total cost is $\widetilde O(n\sum_\ell a_\ell/u)$. The remaining task is therefore to identify the spectral scales that dominate this sum.

For $1<\alpha<2$, their sum is $O(R^{2-\alpha})$: terms grow down to scale $1/R$ and then decrease. With rank-free fixed-time estimation the coefficients are $\delta_\ell^{\alpha-2}$, whose sum is $O(t_0^{\alpha-2})$ below two, logarithmic at two, and bounded above two. Uniform layer error $u/[8(L+1)]$ proves the displayed power laws; the allocation in \eqref{eq:allocation} can also be used.

For $0<\alpha<1$, the rank tail $Rt^\alpha$ remains valid. The finest term dominates the rank-sensitive sum, giving cost $\Ot(Rt_R^{\alpha-1}/u)=\Ot((R/u)^{1/\alpha})$. Appendix~\ref{app:power} supplies the bias parameters and all boundary cases. Projection onto the physical moment interval and the affine entropy conversion complete the algorithm.
\end{proof}
The cutoffs balance the discarded moment against its allowed error; the remaining dyadic sum then determines the query exponent. The transition at $\alpha=2$ is a convergence threshold of the scale sum. For example, at $\alpha=5/2$ the rank-free coefficients are $\sqrt{\delta_\ell}$, so the argument applies directly to a fractional power.

\subsection{Von Neumann Entropy Estimation}
\begin{theorem}[Von Neumann entropy]\label{thm:vn}
For rank at most $R$, von Neumann entropy admits additive-$\varepsilon$ estimation with failure at most $\nu$ using $\Ot(R/\varepsilon)$ controlled purified queries.
\end{theorem}
\begin{proof}
For von Neumann entropy, we use $g(x)=\log(1/x)$ and $G_\ell=C_H(1+\log(1/\delta_\ell))$. The cutoff is
\begin{equation}\label{eq:log-cutoff}
 t=\frac{\varepsilon}{256R\log(256R/\varepsilon)}.
\end{equation}
The low-spectrum contribution is at most $Rt\log(1/t)\le\varepsilon/128$. With rank caps, the coefficient of $1/e_\ell$ is $O(R(1+\ell))$ at fine scales and no larger at coarser scales. There are $O(\log(R/\varepsilon))$ layers, giving the stated cost. Appendix~\ref{app:log} bounds the off-band error by entropy-weighted leakage and verifies the complete error budget.
\end{proof}
The cutoff controls the entropy tail, while the same rank promise controls the cost of resolving the retained spectrum. Below $1/R$, halving the spectral scale doubles its processing cost and halves its rank mass bound. Their product stays of order $R$. This cancellation holds for the entire prefix by the survival analysis in Appendix~\ref{app:vt}.

The von Neumann estimator takes controlled purified access and $R,\varepsilon,\nu$. These inputs and \eqref{eq:application-recipe} determine the cutoff, number of levels, local envelopes, mass bounds, and approximation tolerances. We omit inactive bands, apply Proposition~\ref{prop:band-interface} to every active band, and project the final estimate onto $[0,\log R]$. Appendix~\ref{app:explicit-resources} specifies finite execution limits and synthesis precision, including an explicit logarithmic bound. With the same mass bounds, fixed-time readout gives $\Ot(R/\varepsilon^{3/2})$; the two guarantees in Table~\ref{tab:nearest} therefore differ by a factor $\varepsilon^{-1/2}$.

\subsection{R\'enyi Entropy Estimation}
\begin{theorem}[R\'enyi entropy]\label{thm:renyi}
For rank at most $R$, additive-$\varepsilon$ estimation of $S_\alpha^{\mathrm{R}}$ with failure at most $\nu$ has query cost
\begin{equation}\label{eq:renyi-bound}
 Q_{S_\alpha^{\mathrm{R}}}=
 \begin{cases}
 \Ot((R/\varepsilon)^{1/\alpha}),&0<\alpha<1,\\
 \Ot(R/\varepsilon),&1<\alpha\le3,\\
 \Ot\!\left(\min\left\{\frac{R^{3/2-3/(2\alpha)}}{\varepsilon},\frac{R}{\varepsilon^{1+1/\alpha}}\right\}\right),&\alpha>3.
 \end{cases}
\end{equation}
\end{theorem}
The high-order minimum combines the multi-level construction proved below with the rank-based estimator of~\cite[Corollary~5]{WZL24}; the smaller certified bound is selected from $R,\alpha,\varepsilon$. The construction below proves the new term independently.

For $0<\alpha<1$, $M_\alpha\ge1$, so the logarithm is 1-Lipschitz on the physical moment interval. An additive moment estimate with error $(1-\alpha)\varepsilon/4$ proves the first row. For $\alpha>1$, the logarithm requires moment error proportional to $\varepsilon M_\alpha$. Using the physical lower bound $R^{1-\alpha}$ would demand the precision needed for the smallest possible moment even when the actual moment is much larger. We instead first prove an estimator with error proportional to a supplied upper scale $s$. A coarse pilot will then find $M_\alpha\le s\le8M_\alpha$, allowing a final call at that scale to provide relative moment accuracy. The next proposition supplies the known-scale subroutine used in both steps.

\begin{proposition}[Moment estimation at an upper scale]\label{prop:moment-scale}
Suppose $\alpha>1$ and a known $s$ satisfies $R^{1-\alpha}\le s\le1$ and $M_\alpha\le s$. For $0<h<1/10$, one can estimate $M_\alpha$ to additive error $hs$, with failure at most $\nu$, using $\Ot(R/h)$ queries for $1<\alpha\le3$ and $\Ot(R^{3/2-3/(2\alpha)}/h)$ for $\alpha>3$.
\end{proposition}
For $x\ge\delta_\ell/2$, $x\le x^\alpha(\delta_\ell/2)^{1-\alpha}$. Summing this inequality over the eigenvalues in the selected band uses $\sum_i\lambda_i^\alpha=M_\alpha\le s$ to bound their total mass. Eigenvalues below the band contribute at most the classifier leakage $\xi^2$. Combining this moment bound with the rank and unit bounds gives
\begin{equation}\label{eq:moment-caps}
 \begin{split}
 S_\ell&=\min\{1,8R\delta_\ell\},\\
 b_\ell&=\min\{1,R\delta_\ell,s\delta_\ell^{1-\alpha}\},\\
 B_\ell&=\min\{S_\ell,c_\alpha b_\ell\}.
 \end{split}
\end{equation}
Here $c_\alpha$ is a sufficiently large fixed constant, and $\xi^2\le\min_\ell b_\ell$ controls the leakage slack. The survival cap continues to be $S_\ell$. Also $\lambda_{\max}\le s^{1/\alpha}$, so layers with $\delta_\ell/2>s^{1/\alpha}$ have only leakage contributions and may be omitted. Appendix~\ref{app:renyi} proves the resulting three-regime sum.

\begin{table}[tbp]\centering\footnotesize
\caption{R\'enyi reference profiles and dominant endpoints of the dyadic sums. Here $\delta_\star=s^{1/(\alpha-1)}$, $\Lambda=s^{1/\alpha}$, and $k(\delta)=\delta^{\alpha-2}\sqrt{S_{\mathrm{ref}}B_{\mathrm{ref}}}/s$. Actual circuit bounds differ by fixed-order constants; see \eqref{eq:cap-equivalence}.}\label{tab:renyi-scales}
\setlength{\tabcolsep}{4pt}\renewcommand{\arraystretch}{1.35}\begin{tabular}{p{.18\textwidth}p{.08\textwidth}p{.13\textwidth}p{.22\textwidth}p{.30\textwidth}}\hline
Active scale & $S_{\mathrm{ref}}$ & $B_{\mathrm{ref}}$ & $k(\delta)$ & Dominant endpoint\\\hline
$\delta\le1/R$ & $R\delta$ & $R\delta$ & $R\delta^{\alpha-1}/s$ & Upper endpoint, since $\alpha>1$.\\[3pt]
$1/R\le\delta\le \delta_\star$ & $1$ & $1$ & $\delta^{\alpha-2}/s$ & Lower endpoint for $1<\alpha<2$; logarithmic level count at $\alpha=2$; upper endpoint for $\alpha>2$.\\[3pt]
$\delta_\star\le\delta\le\Lambda$ & $1$ & $s\delta^{1-\alpha}$ & $\delta^{(\alpha-3)/2}/\sqrt s$ & Lower endpoint for $1<\alpha<3$; logarithmic level count at $\alpha=3$; upper endpoint for $\alpha>3$.\\[3pt]\hline\end{tabular}
\par\smallskip\begin{minipage}{.96\textwidth}\footnotesize Only retained levels contribute in each range. An empty intersection contributes zero, and we assign a shared boundary to one range only. At a critical order, the factor is the actual number of retained levels, at most logarithmic.\end{minipage}\end{table}

Table~\ref{tab:renyi-scales} shows which bound determines the cost. In the middle region the exponent changes sign at $\alpha=2$; in the upper region it changes sign at $\alpha=3$. The latter transition determines whether the upper endpoint $\Lambda$ dominates the rank dependence. The few retained scales between $\Lambda$ and $2\Lambda$ change only fixed-order constants.

\begin{lemma}[Pilot scale estimate]\label{lem:pilot}
For fixed $\alpha>1$ and $R\ge2$, the pilot described below uses at most $N_{\mathrm{pilot}}=\lceil(\alpha-1)\log_2R\rceil+3$ calls to Proposition~\ref{prop:moment-scale}. With probability at least $1-\nu/8$, every call satisfies the upper-scale promise and the retained scale obeys $s/8\le M_\alpha\le s$. Every other transcript also terminates within the same call limit.
\end{lemma}
\paragraph{Finding the moment scale and taking the logarithm.}
We initialize the pilot at $s=1$ and allow at most $N_{\mathrm{pilot}}$ calls. Before each call, we apply the certified lower-scale check in Appendix~\ref{app:renyi}; the controller terminates with a failure flag if this check rules out $s\ge R^{1-\alpha}$. At a permitted scale, Proposition~\ref{prop:moment-scale} estimates $M_\alpha$ to error $s/16$, with failure allowance $\nu/(8N_{\mathrm{pilot}})$. We retain $s$ when the rounded estimate is at least $s/4$. Otherwise, we halve $s$ and continue, returning a failure flag if the call limit is exhausted. On simultaneous success, Lemma~\ref{lem:pilot} shows that each new scale remains an upper bound and the retained scale is within a factor of eight of the moment.

At the retained scale, we make one final moment-estimation call with $h=\min\{1,\alpha-1\}\varepsilon/256$ and failure budget $\nu/8$. We project its output onto the physical interval $[R^{1-\alpha},1]$ and then evaluate $\log\widehat M/(1-\alpha)$ to the prescribed rounding tolerance. The pilot determines the scale; the final call supplies the requested accuracy.
Projection cannot increase moment error. On successful calls, $|\widehat M/M_\alpha-1|\le8h$, and the logarithmic conversion in Appendix~\ref{app:renyi} proves Theorem~\ref{thm:renyi}. For $\alpha\ge3/2$, the same moment bounds with fixed-time readout already give the displayed query powers, as summarized in Table~\ref{tab:nearest} and proved in Proposition~\ref{prop:fixed-entropy}.

\begin{proposition}[A one-sided spectral upper bound]\label{prop:upper-spectrum}
For fixed $\alpha>1$ and $C\ge1$, the additional promise $\lambda_{\max}\le C/R$ gives additive-$\varepsilon$ R\'enyi estimation with $\Ot(R/\varepsilon)$ queries.
 When $C=1$, zero queries suffice.
\end{proposition}
The promise gives $R^{1-\alpha}\le M_\alpha\le(C/R)^{\alpha-1}$. Thus a known scale is already within a constant factor of the moment. Retaining only layers with $\delta_\ell\le2C/R$ makes the coefficient sum $O(R^{2-\alpha})$. Division by the required moment tolerance $\Theta(\varepsilon R^{1-\alpha})$ proves the bound; Appendix~\ref{app:renyi} includes the omitted-layer budget.

\subsection{Logarithmic Moments and Entropy Variance}\label{sec:log-moments}
We next apply the functional theorem to moments and fluctuations of the information spectrum. For a fixed positive integer $k$, we define
\begin{equation}\label{eq:log-moment-def}
 \mu_k(\rho)=\Tr[\rho(-\log\rho)^k],\qquad
 V(\rho)=\mu_2(\rho)-S(\rho)^2.
\end{equation}
The traces are restricted to the support of $\rho$; equivalently, the zero-eigenvalue contribution is its continuous limit, zero. Thus $\mu_1=S$, and $V$ is the variance of $-\log\lambda_i$ under probabilities $\lambda_i$. This entropy variance appears in second-order source coding and pure-state entanglement conversion~\cite{DattaLeditzky15,AbdelhadiRenes20}. Its dependence on a single spectrum allows us to estimate it with the same local-profile construction.

\begin{proposition}[Logarithmic moments and entropy variance]\label{prop:log-moments}
For fixed $k\ge1$, controlled purified access and a rank upper bound $R\ge2$ suffice to estimate $\mu_k(\rho)$ to additive error $\varepsilon$ using $\widetilde O(R/\varepsilon)$ queries. Under the same access, $V(\rho)$ can be estimated to additive error $\varepsilon$ using $\widetilde O(R/\varepsilon)$ queries. Both statements hold with constant success probability and require no lower bound on the nonzero eigenvalues.
\end{proposition}
\begin{proof}
We first obtain a computable bound on the singular profile. If $Y$ takes the value $-\log\lambda_i$ with probability $\lambda_i$, then at most $R$ eigenvalues contribute to its upper tail, so
\begin{equation}\label{eq:log-moment-tail-law}
 \Pr(Y\ge t)\le\min\{1,R e^{-t}\},\qquad t\ge0.
\end{equation}
Integrating $k t^{k-1}\Pr(Y\ge t)$ and splitting at $t=\log R$ therefore gives
\begin{equation}\label{eq:log-moment-envelope}
 \begin{aligned}
 \mu_k&\le\mathcal M_k(R),\\
 \mathcal M_k(R)&=(\log R)^k+k!\sum_{j=0}^{k-1}\frac{(\log R)^j}{j!}.
 \end{aligned}
\end{equation}
For fixed $k$, $\mathcal M_k(R)=O((1+\log R)^k)$. We will use this bound to control leakage from the entire spectrum, including eigenvalues outside the retained intervals.

We next choose the cutoff, beginning at $\tau=2^{-2k-3}\le\min\{1/8,e^{-k}\}$. At each trial we compute an upper enclosure of $R\tau[\log(1/\tau)]^k$ with excess at most $\varepsilon/512$. We accept the trial if this upper enclosure is at most $\varepsilon/128$ and otherwise halve $\tau$. The accepted value therefore satisfies
\begin{equation}\label{eq:log-moment-cutoff}
 R\tau[\log(1/\tau)]^k\le\varepsilon/128.
\end{equation}
The function $x[\log(1/x)]^k$ is increasing on $(0,e^{-k}]$, so \eqref{eq:log-moment-cutoff} bounds the contribution of all eigenvalues below $\tau$. A value of order $\varepsilon/[R(1+\log(R/\varepsilon))^k]$, with a sufficiently small constant depending only on $k$, makes the true expression at most $\varepsilon/256$ and hence passes the enclosure test. We choose the least $L$ with $\delta_L\le\tau$, so $\tau/2<\delta_L\le\tau$. The procedure therefore terminates with $L=O(\log(R/\varepsilon))$ levels, using finite certified precision at every trial.

For $g_k(x)=(-\log x)^k$, we use envelopes
\begin{equation}\label{eq:log-moment-local-envelope}
 G_\ell=C_k[1+\log(1/\delta_\ell)]^k.
\end{equation}
Indeed, $g_k(4\delta_\ell y)$ is the $k$th power of $\log(1/(4\delta_\ell))-\log y$. On a fixed complex neighborhood of $J$ that avoids zero and the negative real axis, normalization by \eqref{eq:log-moment-local-envelope} makes these profiles uniformly bounded and analytic. Their coefficients are effectively computable, so Definition~\ref{def:analytic} and Lemma~\ref{lem:compiler} apply.

To check the remaining deterministic errors, let $G_{\max}=\max_{\ell\le L}G_\ell$. The off-band bound of Lemma~\ref{lem:mass} gives $\mathcal L_{g_k,\ell}\le\xi^2[\mathcal M_k(R)+G_\ell]$, while the terminal contribution is at most $\varepsilon/128+\xi^2\mathcal M_k(R)$. We therefore choose
\begin{equation}\label{eq:log-moment-leakage-choice}
 \xi^2\le\min\left\{R\delta_L,\frac14,
 \frac{\varepsilon}{128(L+2)[1+\mathcal M_k(R)+G_{\max}]}\right\}.
\end{equation}
We set $\varepsilon_F=\varepsilon$, use \eqref{eq:eta-budget}, and allocate $E=\varepsilon/4$ to the layer estimates. The total off-band leakage is at most $\varepsilon/128$, and the terminal contribution is at most $\varepsilon/64$. Together with compiler error at most $\varepsilon/8$, these bounds leave room for the rounding allowance in Appendix~\ref{app:budgets}. We assign a total failure budget of $1/12$, including finite synthesis. All additional precision factors are logarithmic in $R/\varepsilon$ for fixed $k$.

It remains to sum the query costs. With $S_\ell=B_\ell=\min\{1,8R\delta_\ell\}$, the rank-profile coefficient in Theorem~\ref{thm:master} satisfies
\begin{equation}\label{eq:log-moment-coefficient}
 a_\ell=G_\ell B_\ell/\delta_\ell
 \le8RG_\ell=O(R(1+\ell)^k).
\end{equation}
Consequently, $(\sum_{\ell=0}^L\sqrt{a_\ell})^2=O(R(L+1)^{k+2})$. Substitution into \eqref{eq:master-cost} proves the stated $\widetilde O(R/\varepsilon)$ bound. This argument also shows why a stronger logarithmic singularity does not change the query powers: its effect is absorbed by the number of levels and the local logarithmic envelope.

Finally, we estimate $\mu_2$ to error $\varepsilon/4$ and $S$ to error $\varepsilon/[8(1+\log R)]$, allocating failure probability at most $1/12$ to each call, including its implementation error. We project the entropy estimate onto $[0,\log R]$. Since both the projected estimate $\widehat S$ and $S$ lie in that interval, $|\widehat S^2-S^2|\le2\log R\,|\widehat S-S|\le\varepsilon/4$. Thus $\widehat\mu_2-\widehat S^2$ estimates $V$ within $\varepsilon/2$, leaving slack for arithmetic. Projection onto $[0,\infty)$ cannot increase the error because $V\ge0$. Both calls use $\widetilde O(R/\varepsilon)$ queries, and the union bound gives success probability at least $5/6$. Rank-one states have $\mu_k=V=0$ and can be handled directly.
\end{proof}

The proposition extends the entropy construction to spectral fluctuation statistics using the same purified access and rank information. In the cited second-order coding results, entropy variance determines the size of the fluctuation term. Our estimate thus supplies this parameter with the same soft rank--precision dependence as entropy itself.

\section{Conclusion and Open Problems}\label{sec:discussion}
The central conclusion of this work is that the multiplicative overhead of a worst-case spectral transform followed by a worst-case readout is not fundamental to quantum entropy estimation. The two costs are usually introduced by different algorithmic steps, but the quantities that determine them are not independent: the same eigenvalues fix the resolution scale, the local function magnitude, and the spectral mass. By retaining this dependence throughout the construction, we replace the global product by a sum of local costs. Local normalization accounts for the rank-independent Tsallis improvement, whereas survival-sensitive estimation supplies the additional saving for von Neumann entropy. For R\'enyi entropy, a moment-scale estimate further bounds the selected mass, while the survival probabilities account for the cost of reaching it.

The resulting rank-independent Tsallis bounds are near-optimal for every fixed order above one: they are $\widetilde\Theta(\varepsilon^{-1/(\alpha-1)})$ for $1<\alpha<2$ when the allowed rank accommodates the hard family, and $\widetilde\Theta(1/\varepsilon)$ for $\alpha\ge2$. For von Neumann and R\'enyi entropies, the new bounds improve the available algorithms without settling all joint rank--precision tradeoffs. In particular, the gap between $\widetilde O(R/\varepsilon)$ and $\widetilde\Omega(R/\sqrt\varepsilon+1/\varepsilon)$ for von Neumann entropy still calls for either a sharper algorithm or a stronger joint lower bound. The corresponding low-rank Tsallis and noninteger-order R\'enyi regimes remain related open problems.

Beyond these particular bounds, Theorem~\ref{thm:master} changes the resource calculation for a broad class of scalar spectral statistics. Once the normalized local profiles are effectively approximable and the discarded contributions are controlled, their estimation cost can be expressed in terms of local quantities rather than an independently optimized global transform. The logarithmic-moment and entropy-variance bounds in Proposition~\ref{prop:log-moments} demonstrate its application to logarithmic profiles as well as powers. Entropy variance describes fluctuations that enter finite-blocklength compression and entanglement conversion~\cite{DattaLeditzky15,AbdelhadiRenes20}; our result estimates that descriptor under coherent purified access, without presupposing a lower bound on the nonzero eigenvalues. Extending this approach to two-state relative-entropy variances would require an estimator that handles both input operators.

This distinction between estimating an expectation and implementing an operator also clarifies what a lower-bound argument must establish. A product lower bound for an explicitly composed input problem does not automatically constrain every algorithm for the corresponding scalar statistic. Conversely, local decomposition cannot remove a product that the statistic itself requires: Proposition~\ref{prop:product-cost} exhibits a bounded polynomial functional on rank-two states with query complexity $\Theta(D/\varepsilon)$. The local approximation and mass conditions in the general theorem identify the structure supporting the entropy improvements.

Several extensions remain worth investigating. Tighter spectral-mass information may sharpen particular rank--precision regimes, while a more economical calibration may reduce logarithmic overhead. Dependence on a varying entropy order is a separate issue: our constants are for fixed $\alpha$, and neither the limit at one nor large orders are uniform, as existing order-sensitive moment bounds also suggest~\cite{WangMonomials25,WangHighOrder26}. A concrete application of the local viewpoint is centered estimation near the maximally mixed spectrum. For known exact rank $R$, subtracting the tangent of $x\log x$ at $1/R$ gives a nonnegative residual whose sum is $\log R-S(\rho)$. Under a two-sided spectral promise, its local magnitude is smaller than that of the uncentered expression. Determining whether this yields a deficit-dependent improvement requires controlling the cost of resolving that narrower interval as well; both costs enter the resulting query bound.

\clearpage\appendix
\section{Coherent Spectral Decomposition and Mass Bounds}\label{app:instrument}

\begin{table}[tbp]\centering\footnotesize
\caption{Registers of one band preparation and their reversal.}\label{tab:registers}
\setlength{\tabcolsep}{3pt}\renewcommand{\arraystretch}{1.2}\begin{tabular}{p{.19\columnwidth}p{.28\columnwidth}p{.43\columnwidth}}\hline
Registers & Role & Retention and reversal\\\hline
$S,E$ & Input system and purification & Only $S$ is processed; the input preparation is reversed last.\\
$A_1,E_1$ & Density-block auxiliaries & Each module starts with clean auxiliaries and reuses them within QSVT.\\
$P_j,h_j$ & Classifier history & All phase and majority records remain until the complete inverse.\\
$\mathsf r_j,\mathsf f_j$ & Running and stop flags & Fresh targets retain the reversible Boolean history.\\
$Z_\ell,H$ & Local transform and readout & The work register $Z_\ell$ is retained and excluded from success.\\
$d,b$ & Selector and reference coin & Both branches remain coherent and are reversed with the prefix.\\\hline\end{tabular}\end{table}

\subsection{Spectral classification and finite implementation}
We implement the spectral classifier by phase estimation on $e^{i\rho}$, whose eigenphases identify the eigenvalues without changing their eigenvectors~\cite{Cleve98}. A signed phase representative in $[-\pi,\pi)$ separates the physical spectrum $[0,1]$ from the branch cut. Phase resolution of a sufficiently small constant times $\delta_j$, followed by a threshold between $\delta_j/2$ and $\delta_j$, distinguishes the two promised regions with constant success probability. We reduce its error to $\xi^2$ by coherent repetition and reversible majority, retaining all records. The total simulated evolution time is $\Ot(1/\delta_j)$, and block-Hamiltonian simulation~\cite{LowChuang19} implements it with the same inverse-scale dependence. Appendix~\ref{app:budgets} accounts for finite simulation and gate precision.

At scale $\delta$, we choose a phase grid of size $N=2^{\lceil\log_2(64\pi/\delta)\rceil}$. Starting from its uniform superposition, we apply the controlled powers of $e^{i\rho}$ and the inverse Fourier transform. On an eigenstate of eigenvalue $x$, outcome $y$ has the geometric-sum amplitude below. We interpret $2\pi y/N$ by its signed representative and compute the high bit by comparison with $3\delta/4$.

\begin{equation}\label{eq:classifier-amplitude}
 v_y(x)=\frac1N\sum_{k=0}^{N-1}e^{ik(x-2\pi y/N)}.
\end{equation}

The standard two-nearest-grid-point event has probability at least $8/\pi^2$ and circular error at most $2\pi/N\le\delta/32$. On that event the threshold separates $x\le\delta/2$ from $x\ge\delta$. The signed representative causes no physical wrap-around ambiguity: the physical spectrum is contained in $[0,1]$, whereas the cut is at $\pm\pi$. A phase estimate just below zero is interpreted as negative and therefore remains low. This is why a signed, rather than unsigned, comparison is used.

We repeat this computation into $k_\xi$ fresh registers, where $k_\xi$ is odd and at least $16\log(2/\xi)$. Conditional on an eigenvalue, the fresh records have a product distribution, so Hoeffding's inequality~\cite{Hoeffding63} bounds the majority error by $\exp(-k_\xi/8)\le\xi^2$. A reversible majority circuit writes the bit used in the flag update. This conditional calculation determines $s_j(x)$ for every eigenvalue; it does not assume independence after averaging over a superposition of eigenvectors.

The sum of evolution times in one phase-estimation register is $N-1$. Its inverse uses the same times with opposite signs. The repetitions therefore require total evolution time $O(\delta^{-1}\log(1/\xi))$. Calling the controlled block-Hamiltonian simulation for each power, to the prescribed finite accuracy, gives the classifier cost stated in Section~\ref{sec:first-stop}. A second control for the running flag is implemented by first computing the conjunction of the controls in a known ancilla. No additional kind of oracle access is required.

The simulation interface is a clean implementation of the target evolution on a zero auxiliary input. If its projected block approximates $e^{it\rho}$ within $\varepsilon_{\mathrm{sim}}$, unitarity bounds the leaked norm by $O(\sqrt{\varepsilon_{\mathrm{sim}}})$, because the target block is unitary. Taking internal accuracy $O(\zeta^2)$ therefore approximates the clean isometry within $O(\zeta)$. Close isometries admit unitary completions at operator distance $O(\zeta)$: rotate their image subspaces by their principal angles, then align the two orthonormal frames. This also controls inverses and controlled versions. Each ideal completion acts as $e^{it\rho}$ on the clean subspace. The actual auxiliary registers are retained, and the full-circuit hybrid argument accounts for their leakage. The quadratic internal accuracy changes only logarithmic simulation factors.

\subsection{Reversible histories and their inverse}
Each flag update in \eqref{eq:reversible-flags} is an XOR into a fresh bit, implemented by a Toffoli gate with a positive or negative control on $h_j$. When $\mathsf r_{j-1}=0$, both new flags remain zero; when $\mathsf r_{j-1}=1$, exactly one of $\mathsf f_j$ and $\mathsf r_j$ becomes one. Induction therefore gives at most one first stop on each reachable history. Different first-stop labels occupy orthogonal computational-basis subspaces, even when their phase records are entangled with the system. We reverse a prefix by undoing its flag updates and controlled classifier in reverse stage order. Because the enabling flags were retained, these inverses act on exactly the branches required to undo the forward circuit, without measurement or reset.

\subsection{Completeness of the retained histories}
We denote by $J_0$ the isometry that initializes all classifier registers and by $V$ the sequential controlled classifier. For each first-stop label $\ell$, $\Pi_\ell$ projects onto that label, while $\Pi_\perp$ projects onto the all-continue branch. Thus $K_b=\Pi_bVJ_0$ for every branch $b$. The projectors are orthogonal and sum to the identity on the reachable subspace, which gives
\begin{equation}
 \sum_b K_b^\dagger K_b
 =J_0^\dagger V^\dagger\left(\sum_b\Pi_b\right)VJ_0=I.
\end{equation}
This remains true for any implemented unitary classifier. Spectral diagonality, used for the analytic decomposition, is exact for the ideal circuit and approximate for its compiled realization.

Classifier $j$ acts only when $\mathsf r_{j-1}=1$, equivalently when all preceding first-stop flags are zero. Conditioned on an eigenvalue, this yields \eqref{eq:responses}. Since the classifier preserves that eigenvector, $K_\ell\ket{u_i}=\ket{u_i}\ket{\kappa_\ell(\lambda_i)}$ for a possibly subnormalized workspace vector. Consequently $K_\ell^\dagger K_\ell=m_\ell(\rho)$, and summing the responses proves $F_g=\sum_\ell A_\ell+A_\perp$ on the support of $\rho$.

\subsection{Conditional states, degeneracies, and retained records}

To analyze the conditional states, we express the input purification in a Schmidt basis; the algorithm does not compute this basis. For an eigenvalue $x$, we denote the classifier workspace vectors after projection onto continue and stop by $v_j^0(x)$ and $v_j^1(x)$. Their squared norms are $1-s_j(x)$ and $s_j(x)$, respectively. Because each classifier uses fresh records, the first-stop-$\ell$ workspace is the tensor product of the earlier continue vectors and the final stop vector, together with the deterministic flags.

\begin{align}\label{eq:history-state}
 \ket{\psi_\rho}&=\sum_i\sqrt{\lambda_i}\ket{u_i}_S\ket{e_i}_E,\nonumber\\
 K_\ell\ket{u_i}&=\ket{u_i}\ket{\kappa_\ell(\lambda_i)},\nonumber\\
 \|\kappa_\ell(x)\|^2&=\|v_\ell^1(x)\|^2
                  \prod_{j<\ell}\|v_j^0(x)\|^2.
\end{align}

If an eigenvalue is degenerate, every vector in its eigenspace has the same scalar response, so the argument is independent of the chosen eigenbasis. Orthogonality of the environment vectors eliminates cross terms when computing branch mass. The local polynomial also acts as a scalar within every eigenspace; hence the same calculation gives the weighted functional contribution. Orthogonality between stop labels, in contrast, comes from the flags and holds independently of the environment. These are separate uses of orthogonality.

The inverse of a prefix must include majority computation, the individual threshold bits, the Fourier transform, controlled evolutions, and initial Hadamards in their reverse order. Uncomputing only the visible high bit would leave which-eigenvalue information in the phase records and would not implement the reflection used by amplitude amplification. Retaining and reversing the complete record is sufficient even when the forward computation has a nonzero classification error: the error changes amplitudes, not reversibility.

\subsection{Proof of Lemma~\ref{lem:mass}}
If $x<\delta_\ell/2$, stopping at $\ell$ requires a false-high outcome there, so $m_\ell(x)\le\xi^2$. If $x>2\delta_\ell$ and $\ell\ge1$, reaching $\ell$ requires a false-low outcome at $\ell-1$, whose threshold scale is $2\delta_\ell$. The case $\ell=0$ has no eigenvalues above $2\delta_0=1$. Thus one decisive incorrect test suffices to bound each off-band response; no sum over all previous classification errors is needed.

Separating eigenvalues below and above $2\delta_\ell$ gives
\begin{equation}
 w_\ell\le\sum_{\lambda_i\le2\delta_\ell}\lambda_i
       +\xi^2\sum_{\lambda_i>2\delta_\ell}\lambda_i
 \le2R\delta_\ell+\xi^2.
\end{equation}
The same argument bounds the probability of being active immediately before stage $\ell\ge1$. At stage zero the reach probability is one, and $S_0=\min\{1,8R\delta_0\}=1$ covers it. Surviving stage $L$ when $x\ge\delta_L$ requires a false-low outcome there. In particular, the tail mass is at most $R\delta_L+\xi^2$. These estimates prove the lemma and identify the deterministic bounds needed for the variable-time circuit.

\subsection{Adjacent-classifier comparison}\label{app:adjacent-comparison}
The input of Chen et al.'s Theorem~13~\cite{Multilevel26} has the following interpretation. Denote its envelope by $A_j$ (its $B_j$, unrelated to our probability cap). Up to its displayed level and approximation factors, its two cost terms are $A_j\sqrt{n_j}/u$ and $\sqrt{A_j/u}/\phi_{j+1}$. Its square-polynomial condition is on $g(x^2)$, using singular values $\sqrt{p_i}/2$. Substituting $p_i=\lambda_i$ does not provide that labeled oracle from an arbitrary purification. The density block instead supplies $\lambda_i$, so such a substitution also changes the scale cost. We therefore compare implementations after converting the spectral access, as described below. 
\begin{proposition}[Two-classifier fixed-time construction]\label{prop:adjacent-bands}
For the nonnegative entropy profiles and bounds of Appendix~\ref{app:entropy}, two adjacent classifiers per independently prepared band, bounded local QSVT, and Proposition~\ref{prop:fixed-band} give coefficient $G_\ell\sqrt{B_\ell}/\delta_\ell$. They require no calibrated variable-time estimation. The construction below gives this density-state adaptation explicitly.
\end{proposition}
\begin{proof}
For the two-classifier comparison, we use the high-response probabilities $s_j(x)$ from the same classifiers, but place the two tests in distinct work registers. Both preserve each system eigenvector. Selecting a low response at $\ell-1$ and a high response at $\ell$ therefore gives the weights
\begin{align}\label{eq:adjacent-weights}
 \omega_0(x)&=s_0(x),\quad \omega_\perp(x)=1-s_L(x),\nonumber\\
 \omega_\ell(x)&=(1-s_{\ell-1}(x))s_\ell(x)\quad(\ell\ge1),\nonumber\\
 \sum_{\ell=0}^L\omega_\ell(x)+\omega_\perp(x)
 &=1+\sum_{\ell=1}^Ls_{\ell-1}(x)(1-s_\ell(x)),\nonumber\\
 0&\le\sum_{\ell=1}^Ls_{\ell-1}(x)(1-s_\ell(x))\le L\xi^2.
\end{align}
Unlike the first-stop construction, these weights form an approximate partition. If $x\le\delta_\ell=\delta_{\ell-1}/2$, then $s_{\ell-1}(x)\le\xi^2$; if $x\ge\delta_\ell$, then $1-s_\ell(x)\le\xi^2$. The resulting excess is nonnegative and biases the functional by at most $L\xi^2H_g$. Each off-band weight is at most $\xi^2$, so the rank and moment mass bounds remain valid, and the terminal branch has the same tail-plus-leakage bound. We use one quarter of the classifier-error allowance in \eqref{eq:application-recipe} to cover this excess together with the original leakage. The local approximation, signed readout with $\beta=2$, and zero-output rule are unchanged. Since the two classifiers together cost $\Ot(\delta_\ell^{-1})$, Proposition~\ref{prop:fixed-band} proves the same soft query bound as for a complete first-stop prefix. Each inverse reverses both work registers and the input preparation.
\end{proof}
We can also express the access conversion for the closest distribution framework. Theorem~12 of Chen \emph{et al.}~\cite[pp.~58:10--58:11]{Multilevel26} applies to projected-unitary encodings. A known rotation scales our density block to $\rho/2$. Its Hermitian dilation has eigenvectors $\ket+\ket{u_i}$ with eigenvalues $\lambda_i/2$, and the state $\ket+\ket{\psi_\rho}$ places weight $\lambda_i$ on these vectors without revealing their basis. With threshold $\delta_\ell/2$ and gap ratio two, the discrimination cost is $O(\log(1/\varepsilon_1)+\delta_\ell^{-1}\log(1/\varepsilon_2))$ purified queries. Thus the conversion has inverse-eigenvalue resolution cost. Our proposition uses the phase classifier above for its error analysis and gives the resulting soft query powers; the full logarithmic and workspace costs depend on the chosen implementation.

\section{Bounded Local Polynomial Approximation and Readout}\label{app:compiler}
\subsection{A globally bounded fixed-domain approximation}
We first approximate the normalized profile on a fixed interval and then control the polynomial outside that interval. We use $K=[1/16,5/8]\subset\operatorname{int}J$. Definition~\ref{def:analytic} supplies a fixed Bernstein ellipse around $J$ inside the common analytic neighborhood. The Bernstein-ellipse bound~\cite[ch.~8]{Trefethen13} gives Chebyshev coefficients satisfying $|c_k|\le2M\varrho^{-k}$ on the affine image of $J$, with $M$ and $\varrho>1$ independent of the level. Consequently, the degree-$d$ truncation $P_{\ell,d}$ satisfies
\begin{equation}
 \sup_{y\in J}|P_{\ell,d}(y)-f_\ell(y)|\le C\varrho^{-d}.
\end{equation}
Taking $d=O(\log(1/\eta))$ gives error $O(\eta)$ and keeps $P_{\ell,d}$ bounded by $1/3$ on $J$, after choosing the envelope slack and constants appropriately.

For a fixed level, write $P_d=P_{\ell,d}$ and set $M_d=\max_{y\in[-1,1]}|P_d(y)|$. The Chebyshev recurrence gives $M_d\le\exp(O(d))$. The explicit window below suppresses this exterior growth while preserving the local approximation. Coefficients are converted to the global Chebyshev basis before perturbation, so their $\ell_1$ error controls the norm on the actual QSVT domain. The fixed norm margin permits finite phase factorization~\cite{LowChuang17,QSVT19,Haah19}.

\subsection{An explicit fixed-gap window}
To bound the same polynomial on the full domain, we use a polynomial window between $K=[1/16,5/8]$ and $J=[1/32,3/4]$. We choose transition centers $a=3/64$ and $b=11/16$, and a sign approximant $S_\omega$ with error $\omega$ and gap $1/64$, obtained by rescaling Lemma~14 of~\cite{QSVT19}. We take its bounded domain to be $[-2,2]$, which contains both shifted arguments below for every $y\in[-1,1]$.

\begin{equation}\label{eq:explicit-window}
 W_\omega(y)=\frac{1+S_\omega(y-a)}2\,
             \frac{1-S_\omega(y-b)}2.
\end{equation}

Each factor is in $[0,1]$ globally. On $K$, both desired signs are separated from zero by at least the gap, so $|1-W_\omega|\le\omega$. Outside $J$, one factor is at most $\omega/2$, so $|W_\omega|\le\omega/2$. The degree is $O(\log(1/\omega))$, because all gaps here are fixed. This verifies the window assertions used above on the entire outer domain, including its negative part.

The windowed polynomial has norm at most $\|P_d\|_J$ on $J$ and at most $\omega M_d/2$ outside $J$. On $K$, its approximation error is bounded by that of $P_d$ plus $\omega\|P_d\|_K$. Choosing $\omega$ proportional to $\eta/\max\{1,M_d\}$ therefore gives the required accuracy and a constant global norm margin. Since $M_d\le\exp(O(d))$, the added degree is $O(d+\log(1/\eta))$. The extension outside $J$ thus costs another logarithmic degree, not another inverse spectral scale.

To preserve the global norm bound with finite coefficients, we first convert the exact windowed polynomial to the Chebyshev basis on $[-1,1]$. We then approximate its coefficients with total absolute error at most $\eta/10$, refining their bit precision during the conversion until a certified enclosure meets this tolerance. Every basis function has magnitude at most one on the QSVT domain, so the coefficient error gives the same uniform bound on the polynomial perturbation. The norm margin retained above consequently permits phase factorization after rounding.

We set $P_\ell=W_\omega P_{\ell,d}$ with the accuracy just specified. The truncation has norm at most $1/3$ on $J$, so after tightening fixed constants the product has norm at most $1/2$ on the full interval. We retain this margin when rounding the coefficients.

\subsection{Spectral rescaling and composition}
The polynomial construction for uniform spectral amplification~\cite[Theorem~17]{QSVT19} supplies an odd bounded polynomial $r_\ell$ with
\begin{equation}\label{eq:amplifier}
 \abs{r_\ell(x)-x/(4\delta_\ell)}\le\eta_{\mathrm{sc}},\quad
 0\le x\le2\delta_\ell,
\end{equation}
and degree $O(\delta_\ell^{-1}\log(1/(\delta_\ell\eta_{\mathrm{sc}})))$. At the two coarsest scales $\delta_\ell\ge1/4$, the linear polynomial $x/(4\delta_\ell)$ suffices. We compose $r_\ell$ with the outer polynomial before compiling the quantum circuit.

For $x\in I_\ell$, the ideal argument $x/(4\delta_\ell)$ lies in $[1/8,1/2]$. For sufficiently small $\eta_{\mathrm{sc}}$, the actual argument lies in $K$. The analytic profiles have a uniform derivative bound there. Comparing the outer approximation to the profile before comparing the two arguments gives
\begin{align}
 &\abs{P_\ell(r_\ell(x))-f_\ell(x/(4\delta_\ell))}\nonumber\\
 &\quad\le\abs{P_\ell(r_\ell(x))-f_\ell(r_\ell(x))}
        +O(\eta_{\mathrm{sc}})\le O(\eta).
\end{align}
We choose $\eta_{\mathrm{sc}}=\Theta(\eta)$ with a sufficiently small constant and set $q_\ell=P_\ell\circ r_\ell$. Since $r_\ell$ maps $[-1,1]$ into itself and $P_\ell$ is bounded there, $q_\ell$ is a global contraction. To implement a polynomial of arbitrary parity, we separate its even and odd parts, $q_{\ell,\mathrm e}(x)=[q_\ell(x)+q_\ell(-x)]/2$ and $q_{\ell,\mathrm o}(x)=[q_\ell(x)-q_\ell(-x)]/2$. Corollary~11 of~\cite{QSVT19} implements each bounded real transform. Because $\rho$ is positive semidefinite, these singular-value transforms are the required eigenvalue transforms. A balanced control ancilla combines the two blocks into $q_\ell(\rho)/2$, equivalently the arbitrary-parity construction of~\cite[Theorem~31]{QSVT19} applied to $q_\ell/2$. We therefore use the fixed normalization $\beta=2$. The degree is at most $\deg(P_\ell)\deg(r_\ell)$, giving $\Ot(1/\delta_\ell)$ queries and proving Lemma~\ref{lem:compiler}.

\subsection{Registers in the arbitrary-parity readout}

The local compiler and the Hadamard test use distinct controls. We denote the parity selector by $z$ and the Hadamard ancilla by $h$. After padding the parity circuits to a common workspace, the compiler is $(H_z\otimes I)(|0\rangle\langle0|_z\otimes U_{\mathrm{e}}+|1\rangle\langle1|_z\otimes U_{\mathrm{o}})(H_z\otimes I)$, with zero-input block $q_\ell(\rho)/2$. We control this complete unitary on $h$ and enable the test only on the selected first-stop flag. The selected-branch output is

\begin{align}\label{eq:hadamard-vector}
 \ket{\Omega_\ell}={}&\tfrac12\ket0_h(I+U_{\ell,g})\ket{\zeta_\ell,0^a}\nonumber\\
 &+\tfrac12\ket1_h(I-U_{\ell,g})\ket{\zeta_\ell,0^a}.
\end{align}

Taking the squared norms of the two components gives $(w_\ell\pm\operatorname{Re}\langle\zeta_\ell,0^a|U_{\ell,g}|\zeta_\ell,0^a\rangle)/2$. The zero auxiliaries occur in the input bra and ket because they were prepared in zero, not because they were measured and accepted at the output. All compiler output auxiliaries are summed over by the Hadamard measurement.

Since $\|q_\ell\|\le1$ and $\beta=2$, each exact sign probability lies between $w_\ell/4$ and $3w_\ell/4$. Thus a nonzero band always gives two nonzero sign probabilities, although both vanish when $w_\ell=0$. We nevertheless calibrate each sign with the known cap $B_\ell$, since the actual mass may be unknown and arbitrarily small. The known-success construction supplies a uniform positive scale for both signs without having to determine whether the physical band is empty.

\begin{equation}\label{eq:sign-range}
 \frac{w_\ell}{4}\le p_{\ell,\pm}\le\frac{3w_\ell}{4},
 \qquad \bar p_{\ell,\pm}\in[B_\ell/2,B_\ell].
\end{equation}

The same calculation is valid for a mixed conditional system state, since the purification and all records were kept in the vector above. It also allows a real profile that changes sign: its sign is represented by the difference of the two joint probabilities, while the positive envelope determines accuracy. Entropy profiles are nonnegative on the physical spectrum, but the framework and the parity construction do not rely on this additional property.

\subsection{Approximation error}
Equation~\eqref{eq:hadamard-vector} already proves \eqref{eq:hadamard}: the two probabilities sum to $w_\ell$ and their difference is the real zero-input matrix element. It remains to charge the local and off-band approximation errors.

On $I_\ell$, the contribution to $|A_\ell-\widetilde A_\ell|$ is at most $G_\ell\eta_\ell w_\ell$. Off-band, global boundedness of $q_\ell$ gives the summand $\lambda_i m_\ell(\lambda_i)(|g(\lambda_i)|+G_\ell)$. Summing yields \eqref{eq:local-error}.

\subsection{A constructive smooth-cutoff compiler for entropy profiles}\label{app:smooth-compiler}

The composite construction above applies to the admissible analytic class. For the entropy profiles there is also a direct even-polynomial construction, useful for making effective generation explicit. We give it on the same local interval $I_\ell=[\delta_\ell/2,2\delta_\ell]\cap[0,1]$; it produces the polynomial of Lemma~\ref{lem:compiler} with the same access and readout normalization.

We construct the cutoff from the standard smooth step
\begin{equation}\label{eq:smooth-bump}
 b(t)=\begin{cases}e^{-1/t},&t>0,\\0,&t\le0,\end{cases}
 \qquad u(t)=\frac{b(t)}{b(t)+b(1-t)}.
\end{equation}

The denominator is positive everywhere, so $u$ is smooth, equals zero on $(-\infty,0]$, and equals one on $[1,\infty)$. We form $\chi_\delta$ as a product of two rescaled copies, equal to one on $[\delta/2,2\delta]$ and zero outside $[\delta/4,4\delta]$. After enlarging the envelope by a fixed-order constant, the even function

\begin{equation}\label{eq:smooth-profile}
 h_\delta(x)=\chi_\delta(|x|)g(|x|)/G_\delta,
 \qquad h_\delta(0)=0
\end{equation}

has norm at most $1/4$ on the real line. We take $G_\delta=c_\alpha\delta^{\alpha-1}$ for powers and $G_\delta=c(1+\log(1/\delta))$ for the logarithm. For this classical construction, the logarithm is evaluated at positive arguments up to $4\delta$, although the density spectrum remains in $[0,1]$. Since the cutoff support stays away from zero, the same construction also covers $0<\alpha<1$.

To control the degree needed for a uniform polynomial approximation, we establish the derivative estimate
\begin{equation}\label{eq:smooth-derivatives}
 \|h_\delta^{(k)}\|_\infty
 \le K_g(C_g/\delta)^k(k!)^2,\qquad k\ge0.
\end{equation}

The estimate follows by tracking derivatives through the cutoff construction. Cauchy's formula for $e^{-1/z}$ on a circle of radius $t/2$ around $t>0$ gives $|b^{(k)}(t)|\le k!(2/t)^k e^{-2/(3t)}$. Maximizing $t^{-k}e^{-2/(3t)}$ bounds this by $C^k(k!)^2$. On $[0,1]$, the denominator of $u$ is bounded below by a positive constant. Its reciprocal remains in the same derivative class, as we see by differentiating the reciprocal identity and using $\binom{k}{j}(j!)^2((k-j)!)^2\le(k!)^2$ in the product-rule convolution. A larger geometric constant absorbs the sum. Rescaling the cutoff contributes $\delta^{-k}$, and the normalized logarithm and fixed powers on $[\delta/4,4\delta]$ satisfy the same estimate. A final product-rule convolution proves \eqref{eq:smooth-derivatives}; the even extension introduces no singularity because the cutoff vanishes near zero.

We expand $h_\delta$ in Chebyshev polynomials, or equivalently expand $h_\delta(\cos\theta)$ in a cosine series. The chain-rule expansion preserves the derivative bound after changing $C_g$: derivatives of cosine are bounded by one, and the partition coefficients are absorbed by the factorial-squared factor. Integrating the $n$th Fourier coefficient by parts $k$ times gives the bound $K_g(C_g/(n\delta))^k(k!)^2$. Choosing $k$ proportional to $\sqrt{n\delta}$ yields

\begin{align}\label{eq:smooth-tail}
 |c_n|&\le C_g e^{-c_g\sqrt{n\delta}},\nonumber\\
 \sum_{n>D}|c_n|&\le\frac{C_g}{\delta}(1+\sqrt{D\delta})
                  e^{-c_g\sqrt{D\delta}}.
\end{align}

The second inequality follows by comparison with an integral under the substitution $v=\sqrt{x\delta}$. A degree $D=O(\delta^{-1}\log^2(1/(\delta\eta)))$ makes the tail at most $\eta/2$. Only even coefficients occur, and we compute them with total absolute error at most $\eta/2$. Since $|T_n(x)|\le1$ on $[-1,1]$, the resulting polynomial approximates $h_\delta$ uniformly to error $\eta$ and has norm below $1/2$ for $\eta<1/8$. It therefore satisfies the local and global requirements of Lemma~\ref{lem:compiler}.

One finite coefficient algorithm integrates the smooth cosine integrand on a uniform mesh, using \eqref{eq:smooth-derivatives} to bound the quadrature error. We refine the mesh and arithmetic until each of the at most $D+1$ coefficient enclosures has width below $\eta/[2(D+1)]$. This construction has classical work polynomial in degree and inverse tolerance; the analytic composite construction above supplies the stronger precision-bit-length guarantee required by Section~\ref{sec:local-condition}. Neither preprocessing cost is an oracle query. Phase factorization uses the global norm margin~\cite{Haah19}, and the readout still has block $q_\ell/2$, so both constructions give the same joint-probability formula.

\section{Calibrated Variable-Time Estimation and Probability Recovery}\label{app:vt}

We now derive the variable-time estimator used for a selected spectral band. The construction follows the variable-time amplification principle of Ambainis~\cite{Ambainis12} and the gain-tracking approach of Chakraborty, Gily\'en, and Jeffery~\cite{CGJ19}. Our purpose is to express the required procedure in terms of ordinary amplitude estimation, with explicit amplification choices and a finite query bound. We first obtain a relative probability estimate from a known positive lower bound, then recover the product of the actual amplification gains, and finally bound the cost of every execution. This derivation proves Theorem~\ref{thm:calibrated} and the band guarantee in Proposition~\ref{prop:band-interface}.

\subsection{Ordinary relative estimation with a finite search}\label{app:ordinary-relative}

We begin with a circuit $A$ that prepares a success event of probability $b\ge b_*>0$. Ordinary amplitude estimation with integer parameter $M$ returns $\widetilde b\in[0,1]$ satisfying~\cite[Theorem~12]{BHMT02}

\begin{equation}\label{eq:cal-bhmt}
 |\widetilde b-b|\le 2\pi\sqrt{b(1-b)}/M+\pi^2/M^2
\end{equation}

with probability at least $8/\pi^2$. We use an implementation with at most $8(M+1)$ calls to $A$ and $A^\dagger$ per repetition, including the calls inside each Grover iterate. For an odd number $r\ge8\log(1/\zeta)$ of independent repetitions, Hoeffding's inequality~\cite{Hoeffding63} bounds the median failure probability by $\zeta$.

To find the probability scale, we try $M=1,2,4,\ldots,2^J$, where $J=\lceil\log_2(32\pi/\sqrt{b_*})\rceil$, and accept a median as soon as it is at least $64\pi^2/M^2$. Exhausting this finite search returns a failure flag. On an accurate estimate, the quantity $x=M\sqrt b$ must satisfy $x\ge7\pi$ at acceptance; otherwise $M^2\widetilde b\le(x+\pi)^2<64\pi^2$. Hence an accepted median lies in $[b/2,2b]$. Conversely, $x\ge16\pi$ forces acceptance, so on simultaneous success the search ends before $M$ reaches $32\pi/\sqrt b$. We then obtain logarithmic relative error $|\log(\widehat b/b)|\le\eta$ by taking a final median with

\begin{equation}\label{eq:cal-final-M}
 M_{\mathrm{fin}}=\left\lceil\frac{64\pi}{\eta\sqrt{\widetilde b}}\right\rceil,
 \qquad 0<\eta<1/10.
\end{equation}

Substitution in \eqref{eq:cal-bhmt} proves the logarithmic guarantee with slack. The geometric sum of the coarse lengths and the final length uses at most $10^4r/(\eta\sqrt b)$ preparation calls. The value $b_*$ bounds the number of medians and their prescribed search lengths on every transcript; an invalid accepted value is rejected before taking a square root or dividing.

\subsection{Monotone stages and the small-target schedule}

We consider a preparation of cost $d_0$, followed by $m\ge1$ reversible stages $V_j$ of costs $d_j$. Each stage appends fresh records and separates histories that have not failed from permanently failed histories. Later stages preserve failure, while a stopped success remains in the not-failed subspace and is otherwise left unchanged. We write $a_j$ for the not-failed probability of the original, unamplified prefix and $\mathcal H_j$ for its probability-weighted cumulative cost:

\begin{equation}\label{eq:cal-stage-masses}
 \begin{gathered}1=a_0\ge a_1\ge\cdots\ge a_m=p\ge p_*>0,\\
 \mathcal H_j=d_0+\sum_{i=1}^j d_i\sqrt{a_{i-1}}.\end{gathered}
\end{equation}

The not-failed projectors are determined by explicit flags, so their reflections require no knowledge of the conditional system state. For a selected band, an early stop outside the target band and rejection by the final Hadamard test are failures; a success in the reference branch is already stopped. We choose a small amplification target and per-stage estimation tolerance as follows:

\begin{equation}\label{eq:cal-target}
 \kappa=\frac1{10^4m},\qquad \eta=\frac{\gamma}{100m},
 \qquad 0<\gamma<1/4.
\end{equation}

We first attenuate the original input with a known coin of active probability $\kappa$, marking its inactive outcome as failed. The resulting preparation $\mathcal A_0$ has not-failed probability $\sigma_0=\kappa$. Given $\mathcal A_{j-1}$, we apply $V_j$ only on its not-failed subspace to obtain the raw preparation $\mathcal B_j$. The preceding relative-estimation procedure supplies an estimate $\widehat b_j$ of its not-failed probability $b_j$. For $j<m$, we choose the largest positive odd integer below the following threshold, with one as the default:

\begin{equation}\label{eq:cal-odd}
 k_j=\max\left\{1,\ \max\left\{k\in2\mathbb N+1:
 k\le\sqrt{\kappa/(4\widehat b_j)}\right\}\right\}.
\end{equation}

When the inner set is empty, we omit its maximum. We apply $(k_j-1)/2$ Grover iterations to $\mathcal B_j$ and call the resulting preparation $\mathcal A_j$. At the final stage $j=m$, we estimate $b_m$ without amplifying it. The two-dimensional amplitude-amplification formula~\cite{BHMT02} gives the exact relations

\begin{equation}\label{eq:cal-exact-stage}
 \begin{gathered}b_j=\sigma_{j-1}\frac{a_j}{a_{j-1}},\qquad \sigma_j=f_{k_j}(b_j),\\
 f_k(b)=\sin^2(k\arcsin\sqrt b).\end{gathered}
\end{equation}

Amplification multiplies the complete not-failed component by one scalar. Its normalized state is therefore the original conditional prefix state, up to a common phase. Induction proves \eqref{eq:cal-exact-stage}, including entanglement with all retained histories. Multiplying the ratios $a_j/a_{j-1}$ then yields

\begin{equation}\label{eq:cal-product}
 p=\frac{b_m}{\kappa}\prod_{j<m}\frac{b_j}{f_{k_j}(b_j)},
 \qquad
 \widehat p=\frac{\widehat b_m}{\kappa}
 \prod_{j<m}\frac{\widehat b_j}{f_{k_j}(\widehat b_j)}.
\end{equation}

The integers $k_j$ in the reported expression are the integers used by the actual circuit. No small-angle approximation replaces the sine in this identity. All auxiliary registers are retained, so each Grover reflection uses a complete prefix and its inverse.

\subsection{Stage invariants and stable probability recovery}

For the following analysis, we condition on all probability estimates so far having logarithmic relative error at most $\eta$. We will prove that every amplified stage then remains within the interval

\begin{equation}\label{eq:cal-invariant}
 \kappa/100\le \sigma_j\le\kappa\quad(0\le j<m),
 \qquad k_j\arcsin\sqrt{b_j}\le2\sqrt\kappa.
\end{equation}

The upper bound begins with $b_j\le\sigma_{j-1}\le\kappa$. If $k_j=1$, the threshold in \eqref{eq:cal-odd} is smaller than three, so $\widehat b_j>\kappa/36$ and $b_j\ge\kappa/40$. If $k_j\ge3$, the selected odd integer lies between one third and one times the threshold, giving $\kappa/40\le k_j^2b_j\le\kappa/3$. The inequalities $\sqrt b\le\arcsin\sqrt b\le\sqrt b(1+2\kappa)$ and $\sin x\ge x(1-x^2/6)$ now imply \eqref{eq:cal-invariant}. With a factor-two allowance at the endpoints, the same bounds hold for every argument between $b_j$ and $\widehat b_j$.

\begin{lemma}[Stable correction]\label{lem:cal-stability}
For fixed odd $k$ and $0<k\theta<\pi/2$, where $\theta=\arcsin\sqrt b$,
\begin{equation}\label{eq:cal-derivative}
 \frac{d\log f_k(b)}{d\log b}=k\tan\theta\cot(k\theta)\in[0,1].
\end{equation}
\end{lemma}

\begin{proof}Differentiation gives the displayed expression. It is nonnegative, and $\tan(k\theta)\ge k\tan\theta$ follows from the monotonicity of $\tan x/x$ on $(0,\pi/2)$. Therefore $\log(b/f_k(b))$ is 1-Lipschitz as a function of $\log b$.\end{proof}

\begin{equation}\label{eq:cal-recovery-error}
 \left|\log\frac{\widehat p}{p}\right|
 \le\sum_{j=1}^m|\log(\widehat b_j/b_j)|\le m\eta\le\gamma/100.
\end{equation}

Here Lemma~\ref{lem:cal-stability} is applied with each selected $k_j$ fixed. Exponentiation gives $|\widehat p-p|\le\gamma p/20$, leaving slack for finite arithmetic. The correction factor retains the dependence of its numerator and amplified denominator on the same estimate.

\subsection{Weighted prefix costs and a computable execution cap}

To account for all nested preparations and inverses, we denote the query cost of $\mathcal A_j$ by $q_j$ and its amplitude gain by $g_j=\sqrt{\sigma_j/b_j}$. A complete Grover iteration uses both the preparation and its inverse, so the cost recurrence and its weighted form are

\begin{align}\label{eq:cal-weighted-recurrence}
 q_j&=k_j(q_{j-1}+d_j),\quad q_0=d_0,\nonumber\\
 q_j\sqrt{a_j/\sigma_j}
 &=\frac{k_j}{g_j}(q_{j-1}+d_j)\sqrt{a_{j-1}/\sigma_{j-1}}.
\end{align}

We write $z_j=q_j\sqrt{a_j/\sigma_j}$ and unroll the weighted recurrence. This expresses the cost in terms of the original stage probabilities, rather than the artificially amplified ones:
\begin{align}\label{eq:cal-unrolled}
 z_j={}&\left(\prod_{t=1}^j\frac{k_t}{g_t}\right)
          \frac{d_0}{\sqrt\kappa}\nonumber\\
 &+\sum_{i=1}^j\left(\prod_{t=i}^j\frac{k_t}{g_t}\right)
 d_i\sqrt{a_{i-1}/\sigma_{i-1}}.
\end{align}
The product factors are at most $e^{2m\kappa}<2$. The exact invariant $\sigma_j\ge\kappa/100$ gives the first bound below; certified threshold rounding weakens its constant from $20$ to $40$ using $\sigma_j\ge\kappa/200$. Direct substitution of $b_j=\sigma_{j-1}a_j/a_{j-1}$ gives
\begin{align}\label{eq:cal-raw-expanded}
 \frac{q_{j-1}+d_j}{\sqrt{b_j}}
 &=\frac{q_{j-1}\sqrt{a_{j-1}}+
              d_j\sqrt{a_{j-1}}}{\sqrt{\sigma_{j-1}a_j}}\nonumber\\
 &\le\frac{41\sqrt{200}\,\overline{\mathcal H}}{\sqrt{\kappa p_*}}
 <\frac{1000\overline{\mathcal H}}{\sqrt{\kappa p_*}}.
\end{align}
Thus, in exact arithmetic,
\begin{align}\label{eq:cal-weighted-bound}
 q_j\sqrt{a_j}&\le20\mathcal H_j,\nonumber\\
 \frac{q_{j-1}+d_j}{\sqrt{b_j}}
 &\le\frac{1000\overline{\mathcal H}}{\sqrt\kappa\sqrt{p_*}},
 \qquad \overline{\mathcal H}\ge\max\{1,\mathcal H_m\}.
\end{align}

The second inequality uses $a_j\ge p_*$ and the lower amplified-mass bound. On a successful history, $b_j\ge\kappa p_*/100$ in exact arithmetic; the rounded construction retains a constant-factor version of this bound. We therefore use the smaller common lower bound $b_*$ below in every coarse search. For a complete call with failure allowance $\zeta$, our finite search and query limits are

\begin{align}\label{eq:cal-counter-data}
 b_*&=\kappa p_*/1000,\quad
 J=\left\lceil\log_2(32\pi/\sqrt{b_*})\right\rceil,\nonumber\\
 r&=2\left\lceil4\log\frac{4m(J+2)}{\zeta}\right\rceil+1,\nonumber\\
 K_{\mathrm{call}}&=\left\lceil
 \frac{2^{30}mr\overline{\mathcal H}}{\eta\sqrt\kappa\sqrt{p_*}}
 \right\rceil.
\end{align}

There are at most $m(J+2)$ medians. Conditional on previous good estimates, each median fails with probability at most $\zeta/[4m(J+2)]$. Their conditional union bound is at most $\zeta/4$. The preparation-call bound following \eqref{eq:cal-final-M}, multiplied by \eqref{eq:cal-weighted-bound} and summed over $m$ stages, is strictly below $K_{\mathrm{call}}$. The resulting coefficient is $10^7<2^{30}$, with slack for interval arithmetic. Initial preparations, controlled inverses, coarse trials, and confidence repetitions are all included. These counts give the displayed executable cap.

We check the query count of each requested block before executing it, evaluating the recurrence $q_j=k_j(q_{j-1}+d_j)$ with saturation at $K_{\mathrm{call}}+1$. A block whose cost exceeds the remaining budget is rejected without execution. We likewise reject invalid probabilities and searches that exhaust $J$. The preceding analysis shows that these rejections do not occur on the successful-estimation event, while the same counter bounds every other transcript. This is a bound on oracle use, distinct from the classical arithmetic cost. The resulting relative estimator uses

\begin{equation}\label{eq:cal-relative-cost}
 O\!\left(\frac{m^{5/2}r\overline{\mathcal H}}{\gamma\sqrt{p_*}}\right)
\end{equation}

queries. The factors $m,J,r$ are logarithmic or polylogarithmic in the entropy parameters for the circuits used here. Classical outcomes select later complete preparations; no measurement interrupts a coherent prefix.

\subsection{A complete calibration call}\label{app:calibration-call}

We now collect the preceding constructions into one calibration call. Its classical inputs are the stage descriptions, exact query counts $d_j$, a lower success bound $p_*$, an upper bound $\overline{\mathcal H}$ on the weighted cost, a relative tolerance $\gamma$, and a failure allowance $\zeta$. These inputs determine the coarse-search schedule and execution limits before the first quantum call; the measured probabilities determine the final estimation lengths and the permitted amplification counts used.

From these inputs we set $\kappa,\eta,J,r$, and $K_{\mathrm{call}}$ as above. We initialize the attenuated preparation $\mathcal A_0$ and record each integer amplification choice when it is selected. At stage $j$, this record determines the complete raw preparation $\mathcal B_j$ and its inverse. The finite relative-estimation procedure of Appendix~\ref{app:ordinary-relative} then supplies $\widehat b_j$ with tolerance $\eta$ and the allocated failure probability.

For $j<m$, we choose the odd integer $k_j$ by \eqref{eq:cal-odd}, construct the corresponding amplified preparation, and retain its full description for later calls and inverses. The final stage is not amplified. We then evaluate the exact product estimator \eqref{eq:cal-product} from the measured probabilities and the selected integers, using the reserved arithmetic precision. The query counter is checked before each charged block; a failed validity check or an attempted call beyond $K_{\mathrm{call}}$ returns a failure flag. Thus even an unsuccessful calibration has a predetermined finite cost.

The controller stores circuit descriptions and classical estimates; it does not store a quantum copy of the unknown conditional state between statistical calls. Every invocation of $\mathcal B_j$ starts with all inputs in their known zero states and executes the recorded sequence. Its inverse reverses that sequence. A median therefore samples repeated executions of one fixed circuit, conditional on the preceding classical transcript. This conditional viewpoint is what permits adaptive amplification choices without an independence assumption between stages.

A stopped success is included in every later not-failed projector. This convention is essential for probability recovery. The reference-success amplitude is physically idle but is amplified together with every other not-failed amplitude. It is therefore charged by the $\sqrt{B_\ell}$ term in every prefix weight. A stopped failure is excluded from every later not-failed projector and can never be restored by a physical stage. Grover amplification can rotate between total good and bad components, but its good component remains the same conditional prefix state; it does not redefine which histories a subsequent stage accepts.

\subsection{Specialized estimation theorem}
\begin{theorem}[Calibrated monotone stages]\label{thm:calibrated}
Consider the monotone reversible stages of \eqref{eq:cal-stage-masses}, with known $p_*>0$ and $\overline{\mathcal H}\ge\max\{1,\mathcal H_m\}$. For $0<\gamma<1/4$ and $0<\zeta<1/4$, the calibration procedure of Appendix~\ref{app:calibration-call}, with the certified arithmetic of Appendix~\ref{app:budgets}, returns a relative-$\gamma$ estimate of $p$ with probability at least $1-\zeta$. It uses at most $K_{\mathrm{call}}$ oracle queries on every transcript, with the integer parameters in \eqref{eq:cal-counter-data}. Its query scaling is \eqref{eq:cal-relative-cost}.
\end{theorem}
\begin{proof}
At the first stage, the attenuated preparation has not-failed mass $\kappa$. Suppose all preceding coarse and final medians satisfy their amplitude-estimation guarantees. The next raw mass is at least $\kappa p_*/200$, including the slack for an interval-rounded odd choice. In particular it exceeds the supplied common lower bound $b_*$. The coarse search terminates within its preset number of trials and returns a constant-factor scale. The final median has logarithmic relative error at most $\eta$, with margin for representing its output. This establishes the induction hypothesis for the next stage.

Conditional on any such preceding transcript, the ordinary median guarantee holds with its allocated failure probability. A union bound over at most $m(J+2)$ medians shows that all these guarantees hold with probability at least $1-\zeta/4$. The chosen odd integers preserve the small-angle interval on this event. The exact state relation \eqref{eq:cal-exact-stage} and the fixed-integer Lipschitz estimate therefore apply even though the integers were chosen adaptively. Equation~\eqref{eq:cal-recovery-error} bounds the logarithmic recovery error; the finite termwise evaluations and final exponential bring it to less than the requested $\gamma$ relative error.

The unrolled recurrence bounds every requested preparation length on the same event. Multiplying its raw cost by the ordinary estimation count and summing over stages yields less than $K_{\mathrm{call}}$. Therefore the counter cannot abort on the event just analyzed. On its complement, checking each circuit block before its execution ensures the same deterministic query cap, irrespective of the values returned by failed medians. A reported failure symbol is included in the unsuccessful probability.

Finally, finite circuit synthesis is compared with the complete ideal controlled execution, using the cap and the slot precision in Appendix~\ref{app:budgets}. Reserving at most $\zeta/4$ for this distributional transfer still leaves total unsuccessful probability below $\zeta$. Thus the stated guarantee covers a finite circuit and finite classical arithmetic, rather than an ideal real-arithmetic procedure alone.
\end{proof}

The theorem separates the data needed to run the algorithm from the probabilities used in its analysis. Its inputs include stage descriptions, a positive lower success bound, and an upper bound on the weighted cost. Ordinary amplitude estimation supplies the numerical choices made during execution. The true conditional masses $a_j$, amplified masses $\sigma_j$, and raw probabilities $b_j$ enter the proof. We use the same distinction in the functional estimator.

\subsection{From staged calibration to the band cost}

For either sign, the success projector of the augmented preparation is

\begin{equation}\label{eq:success-projector}
 \Pi_{\ell,\pm}=\ket1\bra1_d\otimes\ket1\bra1_b
 +\ket0\bra0_d\otimes\ket1\bra1_{\mathsf f_\ell}\otimes\Pi_{H,\pm}.
\end{equation}

Orthogonality gives $\bar p=(B_\ell+p_{\ell,\pm})/2$ and $p_*=B_\ell/2$. The reference success stops physically at setup and remains in every not-failed projector. The abstract initial mass is nevertheless $a_0=1$: reference-branch failure is excluded only at the end of the first classifier stage. Thus the first classifier must be charged in full. For $j\ge1$ the entering mass is at most $(S_j+B_\ell)/2$, and before local readout it is at most $B_\ell$. For the integer segment counts defined in the main text, these entering masses give the bound in \eqref{eq:cal-band-H}.

Since $S_0=1$, we have $c_0\le\sqrt2\,c_0\sqrt{(S_0+B_\ell)/2}$, so the full first-stage charge is compatible with the weighted scale bound. There are $m_\ell=\ell+2$ stages. Substituting $\gamma=e_\ell/(8\beta B_\ell G_\ell)$ in \eqref{eq:cal-relative-cost} and recovering both signed probabilities gives

\begin{equation}\label{eq:cal-band-cost}
 Q_\ell=\Ot\!\left(\frac{G_\ell\sqrt{B_\ell}}{e_\ell}
                 \overline{\mathcal H}_\ell\right).
\end{equation}

\begin{proof}[Proof of Proposition~\ref{prop:band-interface}]
The known-success branch gives \eqref{eq:reference}, allowing Theorem~\ref{thm:calibrated} to apply even when the real success probability is zero. Recovering each sign incurs error at most $2\gamma_\ell B_\ell$, so after multiplication by $\beta G_\ell$ the two errors sum to $e_\ell/2$. The remaining allowance covers arithmetic. The reference success remains in every not-failed projection and is therefore included in each weight of \eqref{eq:cal-band-H}. Finally, when the level is skipped, $|\widetilde A_\ell|\le B_\ell G_\ell\le e_\ell$, which proves the zero-output case.
\end{proof}

To relate the weighted cost to stopping times, we write $P_j$ for the actual real-branch probability of reaching stage $j$. Along any path, the sum of squared segment lengths is no larger than the square of the total length. Averaging this inequality gives $\sum_jc_j^2P_j\le T_{2,\ell}^2$, after which Cauchy--Schwarz bounds $\sum_jc_j\sqrt{P_j}$ by $\sqrt m\,T_{2,\ell}$. The reference success is different: it persists through calibration and contributes at most $\sqrt{B_\ell}T_{\max,\ell}$. Replacing the actual probabilities by certified upper bounds yields the computable stopping-time bound used in the main theorem, namely,

\begin{equation}\label{eq:T2-proof}
 T_{2,\ell}^2\le\Ot\!\left(1+\sum_{j=0}^\ell S_j\delta_j^{-2}\right)
       +O(B_\ell(C_\ell^{\mathrm{loc}})^2).
\end{equation}

Equations~\eqref{eq:cal-band-H}, \eqref{eq:cal-band-cost}, and~\eqref{eq:T2-proof} prove \eqref{eq:general-vt-cost}. For the dyadic rank profile, splitting at $\delta=1/R$ gives

\begin{equation}\label{eq:survival-sum}
 1+\sum_{j\le\ell}S_j\delta_j^{-2}=O(S_\ell\delta_\ell^{-2}).
\end{equation}

Before the split the summands grow as $\delta_j^{-2}$; afterwards they grow as $R/\delta_j$. Both sums are controlled by their last terms. Alternatively, summing $c_j\sqrt{S_j}$ directly gives $\Ot(\sqrt{S_\ell}/\delta_\ell)$. Since $B_\ell\le S_\ell$, the reference and local-readout terms obey the same bound. This proves \eqref{eq:rank-cost}, including the preparation term. With $S_j=B_\ell=1$ it gives the rank-free form. If $B_\ell G_\ell\le e_\ell$, the level is skipped before invoking relative estimation. A zero true success probability causes no exception, since its padded probability is $B_\ell/2$.

\subsection{Zero probabilities and extreme mass separation}

These bounds also describe the limiting mass configurations without modifying the circuit. If the physical success probability vanishes, the known-success branch still leaves probability $B_\ell/2$, so the relative estimator and clipping rule of Proposition~\ref{prop:band-interface} remain valid.

If $B_\ell$ is much smaller than $S_\ell$, the known-success branch gives a small final scale without reducing the real prefix workload. The bound therefore contains $\sqrt{S_\ell B_\ell}$, not $B_\ell$. At the opposite extreme $B_\ell=S_\ell$, both pieces of the mass bound coincide and the rank-sensitive cost becomes linear in that cap. At $B_\ell=1$ all calibration parameters remain valid. A zero mass bound causes the band to be omitted, with its deterministic errors still charged; relative estimation is used only for positive mass bounds.

\section{Global Error Analysis and Implementation Precision}\label{app:budgets}
\subsection{Error composition and allocation}
On the event that all estimators succeed, both estimated and skipped layers satisfy $|\widehat A_\ell-\widetilde A_\ell|\le e_\ell$. Equations~\eqref{eq:local-error} and \eqref{eq:complete} imply
\begin{equation}
 |\widehat F-F_g|\le \mathcal T_g+
 \sum_{\ell\in\mathcal J}(e_\ell+G_\ell\eta_\ell B_\ell+\mathcal L_{g,\ell}).
\end{equation}
We allocate $\varepsilon_F/8$ to each of the tail, total leakage, and total compiler error, and $E=\varepsilon_F/4$ to the layer errors. The remaining tolerance accommodates classical rounding. Synthesis accuracy is treated separately because it changes the distribution of the output rather than the value of a fixed statistical record. For every executed local transform, we choose
\begin{equation}\label{eq:eta-budget}
 \eta_\ell\le\min\left\{\frac1{20},
       \frac{\varepsilon_F}{8(L+1)G_\ell B_\ell}\right\}.
\end{equation} A band with $G_\ell B_\ell=0$ is omitted before this division, since its local transform is unnecessary.
The sum of compiler errors is then at most $\varepsilon_F/8$.

The case $\mathcal J_+=\varnothing$ is handled in Section~\ref{sec:master}. Otherwise, Cauchy--Schwarz on the positive-coefficient set gives
\begin{equation}
 \left(\sum_{\ell\in\mathcal J_+}\sqrt{a_\ell}\right)^2
 \le\left(\sum_{\ell\in\mathcal J_+} a_\ell/e_\ell\right)\left(\sum_{\ell\in\mathcal J_+} e_\ell\right).
\end{equation}
Equality holds at \eqref{eq:allocation}, proving the optimal allocation for this cost model. On every active layer $B_\ell G_\ell/e_\ell>1$, so the cost includes at least one execution of its spectral prefix and compiler. Skipped layers require no such execution. Summing the appropriate layer costs proves \eqref{eq:master-cost}.

\subsection{Errors in a complete entropy call}

The deterministic error budget and the probability budget have different roles. Tail truncation, off-band functional leakage, local polynomial approximation, and skipped contributions bound the value of the ideal estimator on its statistical success event. Synthesis error instead changes the distribution of that event. We control this change by comparing the complete output distributions, which also includes nonlinear probability recovery and adaptive decisions.

For a variable-time band, we divide its statistical allowance between the two Hadamard signs. Each sign estimates $\bar p$ within relative error $e_\ell/(8\beta B_\ell G_\ell)$. After subtraction of the reference contribution, this gives probability error at most $e_\ell/(4\beta G_\ell)$. Multiplication by $\beta G_\ell$ and addition of the two sign errors therefore use at most $e_\ell/2$. The remaining half accommodates certified arithmetic and is consistent with the zero-output rule for skipped bands. The local approximation error is charged separately, although its bound also contains $G_\ell$.

For an adaptive sequence of moment calls, we assign each call a failure allowance in advance, using its index and the deterministic upper bound on the number of calls. Conditional on all earlier calls succeeding, the next input promise holds, so its estimation theorem applies. Summing these conditional allowances bounds the probability of the first failure. If an earlier promise has failed, the query counter and outer iteration limit still ensure termination; the correctness argument relies only on histories for which the promises remain valid.

\subsection{Bounded execution and adaptive failure budgets}

We assign $\nu_{\mathrm{call}}=\nu/[8(L+1)]$ to each sign of each active level. The calibration parameters are known before execution: $p_*=B_\ell/2$, $m=\ell+2$, the integer segment lengths, and the weighted cost bound in \eqref{eq:cal-band-H}. Together with the search and median parameters in \eqref{eq:cal-counter-data}, these determine the execution limit below. The controller records the oracle count of every compiled prefix and inverse, and checks the remaining allowance before starting a requested block.

\begin{equation}\label{eq:explicit-cap}
 H_\ell=K_{\mathrm{call}}
 \quad\text{with }p_*=B_\ell/2,\quad \zeta=\nu_{\mathrm{call}}.
\end{equation}
\begin{equation}\label{eq:cap}
 H_\ell=\Ot\!\left[\gamma_\ell^{-1}
 \left(T_{\max,\ell}+(T_{2,\ell}+1)/\sqrt{B_\ell}\right)\right].
\end{equation}

Equation~\eqref{eq:explicit-cap} is the executable integer bound; \eqref{eq:cap} summarizes its scaling. The counter is checked before a requested block. Appendix~\ref{app:vt} proves that a counter abort is impossible when its relative estimates succeed. Fixed-time estimation instead has its prescribed amplitude-estimation lengths and median repetitions. Summing the caps of all signs gives a deterministic bound $\overline Q$ on every transcript, including failed transcripts. The conditional union bound over levels is at most $\nu/4$; no independence between adaptive choices is required.

\subsection{Finite arithmetic and threshold decisions}

We evaluate \eqref{eq:cal-product} in the logarithmic domain. Each summand is $\log\widehat b_j-\log f_{k_j}(\widehat b_j)$, which is exactly zero when $k_j=1$. We compute each other summand to error at most $\gamma/(100m)$ and the final exponential to relative error at most $\gamma/100$. Together with \eqref{eq:cal-recovery-error}, these choices give the requested relative tolerance. The lower bound $b_*$ and the small-angle interval keep the arguments a known distance from singular denominators. The controller rejects records outside these intervals before performing a division.

We evaluate thresholds with certified intervals rather than exact real comparisons. For the odd-integer choice, we refine the relevant interval to width at most $1/4$ and select the smaller neighboring odd integer if it straddles a boundary. This can weaken the lower bound in \eqref{eq:cal-invariant} to $\kappa/200$; the constants $1000$ and $2^{30}$ in the execution bounds include this slack. For coarse search, we compute an interval of error at most $\pi^2/M^2$ and accept only when its lower endpoint exceeds the threshold. Acceptance still certifies a constant-factor estimate, while $M\sqrt b\ge16\pi$ still guarantees acceptance. These interval rules make all decisions finite without changing the query powers.

For a pilot call, we round the moment estimate before comparing it with $s/4$ and include this rounding in the additive tolerance $s/16$. The dyadic updates of $s$ are exact. A certified lower enclosure for $R^{1-\alpha}$, refined as described in Appendix~\ref{app:renyi}, and the fixed iteration limit bound all unsuccessful histories. We allocate at most $\varepsilon_F/8$ to final summation and functional rounding, in addition to the entropy-conversion allowance in Appendix~\ref{app:renyi}. Projection of recovered probabilities onto $[0,B_\ell]$ cannot increase their error.

\subsection{Choosing precision without a circular resource bound}

We next choose the accuracy of each synthesized unitary. Since \eqref{eq:synthesis} compares the complete execution with its ideal counterpart, this accuracy depends on the total number of approximate-unitary slots. We therefore bound the slot count as a function of the precision and solve for a finite choice before execution. We first hold the semantic parameters fixed: the cutoff, local approximation tolerance, classifier error, positive mass bounds, and confidence. We choose $X\ge2$ to dominate their positive values and reciprocals, the numbers of levels and pilot calls, and the effective description lengths of the supplied profile bounds. For the entropy instances, each of these quantities is polynomially bounded in the displayed problem parameters at fixed order.

The prescribed phase-estimation powers, polynomial degrees, median lengths, and Grover recurrences give constants $c_0,c_1,c_2$ for which the number of slots at $b$ accuracy bits is at most $c_0X^{c_1}(1+b)^{c_2}$. These constants bound the circuit description before execution. We select the first positive integer $b$ satisfying the following inequality and use that precision throughout the call:

\begin{equation}\label{eq:precision-closure}
 2^{-b}c_0X^{c_1}(1+b)^{c_2}\le\nu/[64(L+1)].
\end{equation}

The left-hand side tends to zero exponentially, so this finite search terminates. Taking logarithms shows $b=O(\log X+\log(1/\nu))$ after increasing the constant to absorb $c_2\log(1+b)$. Substituting the selected precision into the simulation and compiler degrees therefore contributes only logarithmic factors. The resulting slot count is an actual computable bound for all transcripts. Oracle queries themselves are exact resources in this model; $b$ controls the known operations that surround them.

This general argument establishes finite precision. The entropy instances use the explicit $M(b)$ and integer search in Section~\ref{app:explicit-resources}. Classical preprocessing and non-query resources are accounted for at the end of this appendix.

\begin{lemma}[Counting approximate modules]\label{lem:module-count}
For the circuit templates in \eqref{eq:explicit-module-costs}, the $M(b)$ of \eqref{eq:explicit-total-queries} bounds every executed approximate module, including modules with zero input queries.
\end{lemma}
\begin{proof}
We write $v_0,v_j$ for the numbers of approximate modules in setup and stage $j$, and $Z_j$ for their number in the full amplified preparation. The inverse has the same counts, while reflections about explicit Boolean flags are exact. Starting from $q_0=d_0\ge1$, the preparation recurrences imply
\begin{align}\label{eq:module-recurrence}
 Z_j&=k_j(Z_{j-1}+v_j),\quad q_j=k_j(q_{j-1}+d_j),\nonumber\\
 Z_j&\le\left(v_0+\sum_{i=1}^jv_i\right)q_j,\nonumber\\
 v_0+\sum_i v_i&\le C\left[r_\xi\sum_{j=0}^L(h_j+1)+n+1\right]
 \le C' r_\xi(n+1)^2.
\end{align}
The first bound follows by dividing the two recurrences and using $q_{j-1}\ge1$, even when $d_j=0$. Each classifier contributes $r_\xi h_j$ evolution and $r_\xi$ Fourier modules; local QSVT and coins contribute the remaining terms. Ordinary amplitude estimation adds a constant number of Fourier modules per repetition, each of which executes a query-nonzero preparation. Summing against the query cap proves the claim for either statistical route. Rejected blocks add no executed modules. A degree-$D$ QSVT module allocates its error to $O(D)$ phases, requiring $O(b+\log D)$ phase bits without increasing its oracle count. This lemma counts complete modules, not elementary gates.
\end{proof}

\subsection{Finite synthesis precision}
For the capped algorithm, let $\overline M$ bound the number of approximate-unitary slots on every transcript, including uses inside amplification, inverses, and confidence repetitions. A slot is one simulated evolution, polynomial transform, or reference rotation; exact input-oracle queries contribute no synthesis error. The circuits and execution limits give a computable $\overline M$, polynomial in the query budget and precision bit lengths. We approximate each slot to operator-norm error at most
\begin{equation}\label{eq:synthesis}
 \delta_{\mathrm{syn}}\le\frac{\nu}{64n\max\{1,\overline M\}},\qquad n=L+1.
\end{equation}
\begin{lemma}[Adaptive distributional transfer]\label{lem:adaptive-transfer}
Suppose the ideal and synthesized executions use the same finite classical controller. If each history-selected module has operator-norm error at most $\delta_{\mathrm{syn}}$ and at most $\overline M$ modules execute on any transcript, the final output distributions differ in total variation by at most $2\overline M\delta_{\mathrm{syn}}$.
\end{lemma}
\begin{proof}
We retain the classical history as a register. At each step, the controller selects a module conditioned on that register, so the two resulting direct-sum classical--quantum channels differ in diamond norm by at most $2\delta_{\mathrm{syn}}$, uniformly over histories. We pad terminated and rejected histories with identity channels to the common length $\overline M$. Replacing one channel at a time and applying contractivity gives the bound. Subsequent measurement and deterministic processing, including thresholds, logarithmic recovery, projection, and failure outputs, cannot increase the distance. The argument compares complete output distributions without conditioning on a rare history.
\end{proof}
Here $\nu$ is the failure allowance for one complete functional call, and each sign receives $\zeta_\ell=\nu/(8n)$. Equation~\eqref{eq:synthesis} gives $2\overline M\delta_{\mathrm{syn}}\le\nu/(32n)=\zeta_\ell/4$, even if one sign uses the full slot allowance. The same transfer estimate supplies both the local reserve in Theorem~\ref{thm:calibrated} or Proposition~\ref{prop:fixed-band} and the final guarantee. The union of the implemented sign failures has probability at most $\nu/4<\nu$. We use each pilot call's assigned functional allowance in this calculation. On statistical success, \eqref{eq:budget} and rounding give error at most $7\varepsilon_F/8$. Simulation~\cite{LowChuang19} and phase synthesis~\cite{Haah19} add only logarithmic quantum precision overhead; we account for their classical preprocessing separately.

\subsection{Explicit query accounting for the entropy instances}\label{app:explicit-resources}
We now give executable resource bounds for the composite compiler of Appendix~\ref{app:compiler}, the phase classifier of Appendix~\ref{app:instrument}, and the common entropy error allocation. We choose the largest permitted $\xi$ and $\eta_\ell$, with a fixed margin for certified rounding, and require operator-norm accuracy $2^{-b}$ for each complete approximate module. The constants $C,C_g\ge2$ below are fixed upper bounds from these constructions and do not depend on the entropy parameters. The resulting conservative segment lengths are
\begin{align}\label{eq:explicit-module-costs}
 r_\xi&=\lceil C\log(2/\xi)\rceil,\quad
 h_j=\lceil\log_2(64\pi/\delta_j)\rceil,\nonumber\\
 c_j(b)&=\left\lceil C r_\xi\left[\delta_j^{-1}
 +h_j\{b+\log(2r_\xi h_j)\}\right]\right\rceil,\nonumber\\
 D_\ell&=\left\lceil C_g\delta_\ell^{-1}
 \log(C_g/\eta_\ell)\log(C_g/(\delta_\ell\eta_\ell))\right\rceil,\nonumber\\
 C_\ell^{\mathrm{loc}}&\le C(D_\ell+1).
\end{align}
Each of the $r_\xi$ phase-estimation repetitions uses $h_j$ evolutions with total time $O(\delta_j^{-1})$. We assign error $2^{-b}/(r_\xi h_j)$ to each evolution. The clean-isometry conversion squares the internal block tolerance, and the conservative qubitization bound $O(t+\log(1/\varepsilon_{\mathrm{sim}}))$~\cite{LowChuang19} then gives $c_j(b)$. The local degree is the product of the logarithmic outer degree and the rescaling degree in \eqref{eq:amplifier}; phase precision changes known gates but not this degree. Reflections and controls refer to explicit flags and require no further input-oracle queries.

Substitution of these segment lengths into \eqref{eq:cal-band-H} and \eqref{eq:cal-counter-data}, with the parameters in \eqref{eq:band-interface-cap} and $\zeta_\ell=\nu/(8n)$, gives a finite sum for the total query count. We use that sum to select the precision without a circular dependence:
\begin{align}\label{eq:explicit-total-queries}
 \overline Q(b)&=\sum_{\ell\in\mathcal J_{\mathrm{active}}}
                  (K_{\ell,+}(b)+K_{\ell,-}(b)),\nonumber\\
 M(b)&=\left\lceil C r_\xi(n+1)^2(\overline Q(b)+1)\right\rceil,\nonumber\\
 b&=\min\{z\in\mathbb N:z\ge1,\ 2^{-z}M(z)\le\nu/(64n)\}.
\end{align}
Here $M$ bounds complete simulation, polynomial-transform, Fourier-transform, and coin slots throughout all capped executions; the factor $r_\xi(n+1)^2$ conservatively counts the slots of a base preparation. A polynomial transform is one slot and uses its fixed degree; its phase errors are budgeted internally. Gates implementing Boolean predicates are exact. Lemma~\ref{lem:module-count} justifies this bound also for slots without queries. The constant $C$ may be enlarged for this count. Since only $c_j(b)$ depends on $b$, $\overline Q(b)$ and $M(b)$ grow at most linearly in $b$, up to ceilings. Thus the integer search terminates, with $b=O(\log(nM(1)/\nu))$. The full-transcript hybrid argument then applies. For the fixed-time rows, we use $K_{\ell,\pm}^{\mathrm{FT}}$ from \eqref{eq:fixed-cap} in place of $K_{\ell,\pm}$ and include the bounded pilot calls in the total whenever they are needed.

\begin{proposition}[An explicit logarithmic bound]\label{prop:vn-explicit}
For the von Neumann instance, $R\ge2$, $0<\varepsilon<1/10$, and $0<\nu<1/3$, the preceding implementation satisfies
\begin{equation}\label{eq:vn-explicit-logs}
 \mathcal L=1+\log\frac{R}{\varepsilon\nu},\qquad
 Q_S=O\!\left(\frac R\varepsilon\mathcal L^{15/2}
                   \log\frac{2\mathcal L}{\nu}\right).
\end{equation}
The constant is independent of $R,\varepsilon,\nu$. The logarithmic exponent follows from the phase classifier, composite compiler, first-stage bound, and per-sign precision allocation fixed above. The estimate is conservative; optimizing these logarithmic factors is separate from the stated rank--precision dependence.
\end{proposition}
\begin{proof}
The cutoff gives $n=O(\log(R/\varepsilon))$. The chosen leakage and approximation tolerances have inverse values polynomial in $R/\varepsilon$, so their logarithms are $O(\mathcal L)$. Also $B_\ell\ge c\varepsilon/\log(R/\varepsilon)$, $\gamma_\ell^{-1}$ is polynomial in $R/\varepsilon$, and every segment bound at $b=1$ is polynomial in these parameters. The exact cap therefore makes $M(1)$ polynomial in $R/(\varepsilon\nu)$, proving $b=O(\mathcal L)$ without assuming the desired query exponent. Consequently
\begin{align}\label{eq:vn-log-ledger}
 n,m_\ell&=O(\mathcal L),\quad
 r_\ell=O(\log(2\mathcal L/\nu)),\nonumber\\
 c_j(b),C_j^{\mathrm{loc}}&=O(\delta_j^{-1}\mathcal L^2),\nonumber\\
 \overline{\mathcal H}_\ell&=O(\mathcal L^2\sqrt{B_\ell}/\delta_\ell),\nonumber\\
 \overline Q&=O\!\left(\frac{n^{7/2}\mathcal L^2}{\varepsilon}
 \log\frac{2\mathcal L}{\nu}\sum_\ell\frac{B_\ell G_\ell}{\delta_\ell}\right).
\end{align}
For the classifier, $h_j=O(\mathcal L)$ and $h_j=O(\delta_j^{-1})$, so \eqref{eq:explicit-module-costs} gives the displayed bound. Since $B_\ell=S_\ell$, summing $c_j\sqrt{S_j}$ geometrically and charging the persistent reference term gives the bound on $\overline{\mathcal H}_\ell$. The last line follows from the $m_\ell^{5/2}$ calibration factor in \eqref{eq:cal-relative-cost} and $e_\ell=\varepsilon/(8n)$. Finally \eqref{eq:log-sum} bounds the remaining sum by $O(Rn^2)$, yielding exponent $7/2+2+2=15/2$. The ceilings contribute only $O(n)$, already covered. This separates the layer allocation, calibration, module precision, envelope sum, and confidence factors.
\end{proof}

\subsection{Workspace, non-query operations, and preprocessing}
To make the retained-history requirement explicit, we write $n_S,n_E$ for the system and purification-environment widths. The two preparation register sets require $2(n_S+n_E)$ qubits. At scale $\delta_j=2^{-j-1}$, constant-resolution phase estimation uses $O(j+1)$ phase bits per repetition, so $O(\log(1/\xi))$ repetitions retain $O((j+1)\log(1/\xi))$ bits. Summing through $L$ gives $O((L+1)^2\log(1/\xi))$ phase records and $O(L+1)$ flags. These counts concern only the classifier history; the simulation and polynomial-processing auxiliaries must still be included.

A conservative finite implementation allocates a fresh clean workspace to each approximate-unitary slot of one unamplified preparation and retains it until that preparation is reversed. If there are $M_{\mathrm{base}}$ such slots and each needs at most $a_{\max}$ auxiliary qubits, an upper bound on the quantum workspace is
\begin{equation}\label{eq:workspace-bound}
 \begin{split}
 n_{\mathrm{work},\ell}=O\!\bigl(&n_S+n_E+(L+1)^2\log(1/\xi)\\
 &+M_{\mathrm{base}}a_{\max}+\log H_\ell\bigr).
 \end{split}
\end{equation}
The final term covers an amplitude-estimation phase register; the execution counter itself can be classical. Amplification reuses the same full preparation and its inverse on this workspace, rather than allocating a copy at each recursive call. Separate complete statistical calls may reuse it after measurement and reinitialization. This conservative allocation is sufficient for the retained-auxiliary convention of Appendix~\ref{app:instrument}.

SWAPs, Boolean controls, and reflections about explicit bit patterns have finite gate decompositions with polynomial overhead in register width. The oracle's own gate implementation is additional input data. Coefficient generation, basis conversion, and phase synthesis occur classically before their corresponding quantum calls. The admissibility condition guarantees that these descriptions can be produced to certified precision; their cost is not included in oracle-query bounds. In particular, logarithmic synthesis-accuracy overhead in quantum queries does not imply logarithmic classical preprocessing time. The general profile theorem measures oracle-query resources. For the elementary power and logarithmic profiles, computable analytic bounds provide the finite coefficient data used by the compiler.

More explicitly, let $g_U$ bound the gates for any allowed input-oracle call, and let $N_{\mathrm{known}}$ be the maximum number of elementary known operations obtained by expanding the capped circuit templates (including controls, reflections, and Fourier transforms). For a chosen finite gate set let $C_{\mathrm{synth}}(a)$ bound the cost of approximating one such operation to error $a$. Then
\begin{equation}\label{eq:gate-resource-layer}
 G_{\mathrm{total}}\le \overline Q g_U+
 N_{\mathrm{known}}\,C_{\mathrm{synth}}\!\left(\frac{2^{-b}}{2\max\{1,N_{\mathrm{known}}\}}\right).
\end{equation} We reserve half of each module's error for its analytic approximation and half for known-gate synthesis. The displayed conservative per-gate tolerance bounds the synthesis error of any module by $2^{-b}/2$ and is compatible with the same $M(b)$ transfer. Replacing internal tolerances by half their values only enlarges the fixed constants in \eqref{eq:explicit-module-costs}.
This gate bound depends on the supplied implementation of the oracle and the chosen gate set, neither of which is specified by query access. The workspace bound counts clean auxiliaries at the entry to each complete module; internal QSVT calls share those auxiliaries. The composite compiler requires classical work polynomial in its degree and precision bit length. The optional smooth quadrature construction instead guarantees polynomial work in inverse tolerance. We keep these classical and gate resources distinct from the query complexity.

\section{Supplementary Proofs for Entropy Applications}\label{app:entropy}
\subsection{Power functions, tails, and a common error recipe}\label{app:power}
We write $n=L+1$ for the number of levels. Since the application profiles are nonnegative on $(0,1]$, an upper bound $H_g$ on $F_g$ also controls their functional-weighted leakage. We use $G_*=\max_{\ell\le L}G_\ell$ for the largest local envelope. The following choices of $H_g$ follow from trace normalization and the rank promise:
\begin{equation}\label{eq:functional-bound}
 H_g=\begin{cases}
 1,&g(x)=x^{\alpha-1},\ \alpha>1,\\
 R^{1-\alpha},&g(x)=x^{\alpha-1},\ 0<\alpha<1,\\
 \log R,&g(x)=\log(1/x).
 \end{cases}
\end{equation}
By localization and \eqref{eq:leakage},
\begin{equation}\label{eq:application-leakage}
 \mathcal L_{g,\ell}\le\xi^2(H_g+G_\ell).
\end{equation}
The terminal contribution is bounded by the true spectral tail below $\delta_L$ plus $\xi^2H_g$. If a known spectral upper bound $\Lambda$ permits omission of a layer with $\delta_\ell/2>\Lambda$, all its eigenvalues lie below that band's lower edge; hence $|A_\ell|\le\xi^2H_g$. Summing the retained leakages and omitted-layer errors, including the terminal error, gives at most $(n+1)\xi^2(H_g+G_*)$ in addition to the true low-spectrum tail.

For functional error $u\in(0,1/10)$, we choose a cutoff $t\le1/8$ whose tail contribution is at most $u/128$, and take $L$ to be the least index with $\delta_L\le t$. Thus $t/2<\delta_L\le t$. The classifier, approximation, and statistical tolerances are
\begin{equation}\label{eq:application-recipe}
 \begin{split}
 \xi^2&\le\min\left\{\frac{u}{256(n+1)(1+H_g+G_*)},
                       \frac{b_{\mathrm{leak}}}{100n}\right\},\\
 \eta_\ell&\le\min\left\{\frac1{20},
            \frac{u}{128n(1+G_\ell B_\ell)}\right\},\qquad
 e_\ell=\frac{u}{8n}.
 \end{split}
\end{equation}
Here $b_{\mathrm{leak}}$ is a computable lower scale used to keep classifier leakage within the selected mass bounds.

We use $b_{\mathrm{leak}} =\min\{1,R\delta_L\}$ for the ordinary rank bounds, $b_{\mathrm{leak}} =\min_\ell b_\ell$ for the refined moment bounds, and $b_{\mathrm{leak}}=1$ for rank-independent estimation. The deterministic bias is at most $u/128+u/256+u/128$, and the statistical tolerances sum to at most $u/8$. These choices satisfy \eqref{eq:budget} with room for rounding. Uniform allocation suffices for all application bounds, including cases in which we retain only a subset of levels; the optimized allocation remains available. Failure allowances are assigned as in Appendix~\ref{app:budgets}.

\subsubsection{Power tails and the transition at order two}

We next derive the query powers in Theorem~\ref{thm:tsallis}. We use its moment tolerance $u$ and suppress only fixed-order envelope constants. For $\alpha>1$, trace normalization gives the rank-free cutoff $t_0=\min\{1/8,(u/128)^{1/(\alpha-1)}\}$ through the estimate

\begin{equation}\label{eq:power-tail-expanded}
 \sum_{0<\lambda_i<t}\lambda_i^\alpha
 \le t^{\alpha-1}\sum_{0<\lambda_i<t}\lambda_i\le t^{\alpha-1}.
\end{equation}

With rank at most $R$, counting terms instead gives $Rt^\alpha$, valid also below order one. Thus $t_R=\min\{1/8,(u/(128R))^{1/\alpha}\}$ is sufficient for every positive order. The two tail estimates are distinct: trace normalization gives the first without a rank promise, whereas the second uses a bound on the number of nonzero eigenvalues.

\begin{equation}\label{eq:power-coeff-expanded}
 a^{(0)}(\delta)=\delta^{\alpha-2},\qquad
 a^{(R)}(\delta)=\min\{1,R\delta\}\delta^{\alpha-2}.
\end{equation}

These reference coefficients express the chosen envelope bounds up to fixed-order constants; the actual branch masses are bounded by those envelopes. Under the uniform choice $e_\ell=u/(8n)$, total cost is $\Ot(n\sum_\ell a_\ell/u)$. For any fixed nonzero exponent $z$, the finite dyadic sum obeys the following bounds, whose constants depend on $z$:

\begin{equation}\label{eq:dyadic-power-sum}
 \sum_{\ell=0}^L\delta_\ell^z
 =\begin{cases}O(1),&z>0,\\O(\delta_L^z),&z<0,\end{cases}
 \qquad \sum_{\ell=0}^L\delta_\ell^0=n.
\end{equation}

The constants for $z\ne0$ need not remain bounded as $z\to0$. The equality at $z=0$ explains the boundary order. For $1<\alpha<2$, the rank-free sum is controlled by its finest scale. For $\alpha>2$, it is controlled by the coarsest scales. Substituting $\delta_L=\Theta(t_0)$ gives

\begin{equation}\label{eq:rankfree-power-expanded}
 \frac1u\sum_\ell a^{(0)}(\delta_\ell)
 =\begin{cases}
 O(u^{-1/(\alpha-1)}),&1<\alpha<2,\\
 n/u,&\alpha=2,\\
 O(1/u),&\alpha>2.
 \end{cases}
\end{equation}

The total number of levels and all precision logarithms depend only on $u$ and the fixed order in this construction. No hidden rank parameter is introduced by the classifier accuracy, compiler tolerance, or execution cap. The actual system and purification registers still determine nonquery workspace.

\subsubsection{Rank-sensitive summation and orders below one}
The rank bound changes the fine-scale coefficients. To see where its effect enters, we split the reference profile at $\delta=1/R$:
\begin{equation}\label{eq:rank-power-split}
 a^{(R)}(\delta)=\begin{cases}
 R\delta^{\alpha-1},&\delta\le1/R,\\
 \delta^{\alpha-2},&\delta\ge1/R.
 \end{cases}
\end{equation}

For $1<\alpha<2$, the coarse coefficients increase toward $1/R$, whereas the fine coefficients decrease beyond that scale. Both geometric sums are therefore bounded by their common endpoint value, up to a fixed-order factor, giving $O(R^{2-\alpha})$. If the cutoff removes the crossover, the retained sum can only be smaller. Comparing this certified bound with the rank-free one gives the minimum in Theorem~\ref{thm:tsallis}; we select the less expensive construction before execution, without estimating the unknown rank.

For $0<\alpha<1$, the fine coefficients instead increase toward the cutoff. Since $u<1/10$, $t_R<1/R$, and

\begin{align}\label{eq:belowone-power-expanded}
 \sum_\ell a^{(R)}(\delta_\ell)
 &\le O(R^{2-\alpha}+Rt_R^{\alpha-1})\nonumber\\
 &=O(Rt_R^{\alpha-1}),\nonumber\\
 \frac{Rt_R^{\alpha-1}}u&=O((R/u)^{1/\alpha}).
\end{align}

This explains why rank information is essential to this below-one construction: the function envelope grows near zero, and the tail is controlled by counting eigenvalues. The compiler remains valid because it approximates the negative power only on intervals separated from zero. Its global polynomial is bounded, so it never evaluates a singular function at a zero eigenvalue.

Every executed level has $B_\ell G_\ell/e_\ell>1$, so its displayed statistical cost already dominates one full prefix and readout execution. An additional independent $1/t_R$ setup charge is unnecessary. A skipped level has no execution cost, but its contribution remains charged to $e_\ell$. This active-level accounting also covers all-skipped instances.

Projection of a moment estimate onto its physical interval cannot increase error. For Tsallis, the affine conversion divides the error by $|1-\alpha|$. Choosing the stated $u$ and performing the final arithmetic within $\varepsilon/4$ gives Theorem~\ref{thm:tsallis}. For $R=1$, the direct zero output handles all orders.

\subsection{Entropy-weighted logarithmic leakage and the complete bound}\label{app:log}

For $R\ge2$, we use the cutoff in \eqref{eq:log-cutoff} and take $L$ to be the least index with $\delta_L\le t$. The resulting $n=L+1$ levels satisfy $t/2<\delta_L\le t$. The exact first-stop responses decompose the entropy as

\begin{equation}\label{eq:vn-exact-decomposition}
 S(\rho)=\sum_{\ell=0}^L A_\ell+A_\perp,\qquad
 A_\perp=\sum_i\lambda_i m_\perp(\lambda_i)\log(1/\lambda_i).
\end{equation}

The identity remains exact when an eigenvalue in a transition region contributes to neighboring branches, because it uses the actual coherent responses. For $v=256R/\varepsilon$, the cutoff is $t=1/(v\log v)$. The monotonicity of $x\log(1/x)$ below $e^{-1}$ and the rank bound imply

\begin{align}\label{eq:vn-tail-expanded}
 \sum_{0<\lambda_i<\delta_L}\lambda_i\log(1/\lambda_i)
 &\le Rt\log(1/t)\nonumber\\
 &=\frac{\varepsilon(\log v+\log\log v)}{256\log v}
 \le\varepsilon/128.
\end{align}

For $\lambda_i\ge\delta_L$, terminal survival entails a false-low result at the final test and has probability at most $\xi^2$. Splitting the terminal sum at $\delta_L$ therefore proves

\begin{equation}\label{eq:vn-terminal-expanded}
 0\le A_\perp\le\varepsilon/128+\xi^2 S(\rho)
 \le\varepsilon/128+\xi^2\log R.
\end{equation}

For a nonterminal branch, the pointwise response bound $m_\ell(x)\le\xi^2$ outside $I_\ell$ must be applied to the entropy-weighted sum. It gives

\begin{equation}\label{eq:weighted-log}
 \sum_{\lambda_i\notin I_\ell}\lambda_i m_\ell(\lambda_i)
       \log(1/\lambda_i)\le\xi^2 S(\rho)\le\xi^2\log R.
\end{equation}

Separately applying the same pointwise bound to $\sum_i\lambda_i=1$ bounds the unweighted off-band mass by $\xi^2$. Together with the entropy-weighted estimate, this controls the logarithmic singularity. Thus a globally bounded compiled polynomial contributes at most $G_\ell\xi^2$ outside the band. On the band, local approximation contributes at most $G_\ell\eta_\ell B_\ell$. Combining these terms gives

\begin{equation}\label{eq:vn-local-expanded}
 |A_\ell-\widetilde A_\ell|
 \le G_\ell\eta_\ell B_\ell+\xi^2(\log R+G_\ell).
\end{equation}

Here $\widetilde A_\ell=\beta G_\ell(p_{\ell,+}-p_{\ell,-})$ with $\beta=2$. We apply \eqref{eq:application-recipe} with $u=\varepsilon$, $H_g=\log R$, and $B_\ell=S_\ell$. Equations~\eqref{eq:vn-terminal-expanded} and \eqref{eq:vn-local-expanded} then bound the deterministic bias by $\varepsilon/128+\varepsilon/256+\varepsilon/128$. Estimated and skipped levels contribute at most $\sum e_\ell=\varepsilon/8$. Their sum is smaller than $\varepsilon/2$, leaving the prescribed margin for rounding.

We next sum the query costs. On coarse scales $\delta\ge1/R$, the bounds $B_\ell\le1$ and $\delta^{-1}\le R$ apply; on fine scales, we have $B_\ell\le8R\delta$. Both regions therefore satisfy $B_\ell G_\ell/\delta_\ell=O(R(1+\ell))$. The fine-scale calculation expresses the relevant cancellation: doubling the spectral resolution halves the rank-based mass bound. Consequently,

\begin{equation}\label{eq:log-sum}
 \sum_{\ell=0}^L\frac{B_\ell G_\ell}{\delta_\ell}
 =O(Rn^2),\qquad
 \sum_{\ell=0}^L Q_\ell=\Ot(R/\varepsilon).
\end{equation}

Here $n=O(\log(R/\varepsilon))$, and the chosen compiler and classifier precisions have logarithmic bit lengths. Appendix~\ref{app:budgets} transfers the successful ideal-output event to the implemented circuit using its deterministic cap. Finally, projection onto $[0,\log R]$ cannot increase error. This proves Theorem~\ref{thm:vn}.

\subsection{Relative moment accuracy for R\'enyi entropy}\label{app:renyi}
\subsubsection{Refined mass bounds}
On $x\ge\delta_\ell/2$, the inequality preceding \eqref{eq:moment-caps} gives
\begin{equation}\label{eq:moment-mass-proof}
 w_\ell\le 2^{\alpha-1}s\delta_\ell^{1-\alpha}+\xi^2.
\end{equation}
For smaller eigenvalues, the stop probability is at most $\xi^2$. We combine this estimate with the rank and unit bounds. When $\xi^2\le b_\ell$, a sufficiently large fixed $c_\alpha$ in \eqref{eq:moment-caps} makes $B_\ell$ a valid upper bound on $w_\ell$. To evaluate the cost, we compare the implemented bounds with the simpler profiles
\begin{equation}\label{eq:cap-equivalence}
 \begin{aligned}
 S_{\mathrm{ref}}(\delta)&:=\min\{1,R\delta\},\\
 B_{\mathrm{ref}}(\delta)&:=\min\{1,R\delta,s\delta^{1-\alpha}\}.
 \end{aligned}
\end{equation}
The implemented bounds satisfy $S_\ell=\Theta(S_{\mathrm{ref}}(\delta_\ell))$ and $B_\ell=\Theta(B_{\mathrm{ref}}(\delta_\ell))$. The reference profiles below have exact piecewise expressions; comparison with the implemented bounds costs only fixed-order constants.
This refinement applies only to the selected band. Smaller eigenvalues can still reach its classifier before being rejected, so the same moment cap cannot replace the survival bound $S_\ell$ in earlier stage costs.

\subsubsection{Proof of Proposition~\ref{prop:moment-scale}}
We use moment tolerance $u=hs$, its rank-dependent cutoff $t_R$, and the error allocation in \eqref{eq:application-recipe}. The promise $M_\alpha\le s$ implies that every eigenvalue is at most $\Lambda=s^{1/\alpha}$, so we retain only levels with $\delta_\ell\le2\Lambda$. Their earlier classification stages remain in the query count. The preceding leakage bound accounts for all omitted levels.

The moment bound changes from the unit bound to $s\delta^{1-\alpha}$ at the scale $\delta_\star=s^{1/(\alpha-1)}$, which satisfies $1/R\le\delta_\star\le\Lambda\le1$. Using the reference profiles in \eqref{eq:cap-equivalence}, we obtain the following three exact expressions for the cost coefficient divided by $u$:
\begin{equation}\label{eq:renyi-regimes}
 \begin{aligned}
 &\frac{\delta^{\alpha-2}\sqrt{S_{\mathrm{ref}}(\delta)B_{\mathrm{ref}}(\delta)}}{hs}\\
 &\quad=\begin{cases}
 R\delta^{\alpha-1}/(hs),&\delta\le1/R,\\
 \delta^{\alpha-2}/(hs),&1/R\le\delta\le \delta_\star,\\
 \delta^{(\alpha-3)/2}/(h\sqrt s),&\delta_\star\le\delta\le\Lambda.
 \end{cases}
 \end{aligned}
\end{equation}
The first regime follows because $s\ge R^{1-\alpha}\ge R\delta^\alpha$. The second has $S_{\mathrm{ref}}=B_{\mathrm{ref}}=1$, and the third has $S_{\mathrm{ref}}=1$, $B_{\mathrm{ref}}=s\delta^{1-\alpha}$.

The finite geometric sums \eqref{eq:renyi-low-sum}--\eqref{eq:renyi-upper-sum} below bound the first two regions by $O(nR/h)$. They bound the third by the same expression up to order three; above three its endpoint is
\begin{equation}\label{eq:high-renyi}
 \frac{s^{-3/(2\alpha)}}h
 \le\frac{R^{3/2-3/(2\alpha)}}h.
\end{equation}
Empty regions contribute zero and shared endpoints are counted once. The detailed accounting below includes the extra scales up to $2\Lambda$, the critical orders, and every omitted classifier prefix. Together with the factor $8n$ from uniform allocation, it completes the proof of Proposition~\ref{prop:moment-scale}.

\paragraph{Upper scales and retained prefixes.}

We now account for omitted bands and retained prefixes. The rank cutoff makes the contribution from $x<\delta_L$ at most $R\delta_L^\alpha\le u/128$, after the constants in the common recipe are chosen. For a discarded coarse layer with $\delta_\ell/2>\Lambda=s^{1/\alpha}$, every eigenvalue lies strictly below that layer's low threshold. Its stop response is therefore at most $\xi^2$, even though its contribution is not identically zero in the actual classifier.

\begin{align}\label{eq:renyi-omitted}
 \sum_{\ell:\delta_\ell>2\Lambda} A_\ell
 &\le n\xi^2 M_\alpha\le n\xi^2s,\nonumber\\
 A_\perp&\le R\delta_L^\alpha+\xi^2M_\alpha.
\end{align}

Both terms fit the common error recipe, which uses the weaker bound $M_\alpha\le1$ when selecting $\xi$. The omitted layers need no local transforms or statistical calls. Their classification stages, however, still occur inside every retained prefix: a retained band is defined by first stop, including survival through those coarse thresholds. All such stages remain in $\overline{\mathcal H}_\ell$. Replacing the prefix by a fresh classifier starting at the first retained scale would define different branch weights and would require a separate decomposition proof.

For the dyadic coefficient sum, let $\mathcal I_1,\mathcal I_2,\mathcal I_3$ be the retained scales in the three regions of \eqref{eq:renyi-regimes}, excluding the at most two scales above $\Lambda$. Uniform errors $e_\ell=u/(8n)$ suffice for the stated soft bounds. With this allocation, the total cost is bounded by $8n$ times the sum of the displayed coefficients, retaining their dependence on the same spectral scale. The following endpoint estimates show all parameter dependences before logarithms are suppressed.

\begin{align}\label{eq:renyi-low-sum}
 \sum_{\delta\in\mathcal I_1}\frac{R\delta^{\alpha-1}}{hs}
 &\le\frac{R^{2-\alpha}}{hs(1-2^{1-\alpha})}\nonumber\\
 &\le\frac{R}{h(1-2^{1-\alpha})}.
\end{align}

The first inequality extends the finite sum down to zero, which increases it because $\alpha>1$. The second uses $s\ge R^{1-\alpha}$. If the cutoff is above $1/R$, this region is empty and its contribution is zero. For the middle region, a geometric sum gives the following alternatives; constants depend on the fixed order and include the possible factor-two mismatch of a dyadic endpoint.

\begin{equation}\label{eq:renyi-middle-sum}
 \sum_{\delta\in\mathcal I_2}\frac{\delta^{\alpha-2}}{hs}
 \le\begin{cases}
 O(R^{2-\alpha}/(hs)),&1<\alpha<2,\\
 n/(hs),&\alpha=2,\\
 O(\delta_\star^{\alpha-2}/(hs)),&\alpha>2.
 \end{cases}
\end{equation}

For $1<\alpha<2$, the physical lower bound $s\ge R^{1-\alpha}$ bounds the first expression by $O(R/h)$. At $\alpha=2$, it gives $nR/h$. For $\alpha>2$, we have $\delta_\star^{\alpha-2}/s=1/\delta_\star\le R$. Thus the middle region contributes at most $O(nR/h)$. The same conclusion holds when $\delta_\star=1/R$, where the two endpoints coincide and the region contains at most one dyadic scale.

\begin{equation}\label{eq:renyi-upper-sum}
 \sum_{\delta\in\mathcal I_3}\frac{\delta^{(\alpha-3)/2}}{h\sqrt s}
 \le\begin{cases}
 O(1/(h\delta_\star)),&1<\alpha<3,\\
 n/(h\sqrt s),&\alpha=3,\\
 O(s^{-3/(2\alpha)}/h),&\alpha>3.
 \end{cases}
\end{equation}

The identity at the lower endpoint is $\delta_\star^{(\alpha-3)/2}/\sqrt s=1/\delta_\star$. At $\alpha=3$, the physical bound $s\ge R^{-2}$ gives $nR/h$. Above three, substituting $\Lambda=s^{1/\alpha}$ yields the upper-endpoint power $s^{-3/(2\alpha)}$. The extra retained scales between $\Lambda$ and $2\Lambda$ have the same size up to fixed-order constants. These endpoint calculations complete the sum of the three reference coefficients.

\subsubsection{Proof of Lemma~\ref{lem:pilot}: validity and termination}
We analyze the pilot conditional on all preceding calls satisfying their complete additive-error guarantees. Initially $M_\alpha\le s=1$. If the pilot continues, then
\begin{equation}\label{eq:pilot-invariant}
 M_\alpha\le\widetilde M+s/16<5s/16<s/2.
\end{equation}
Halving therefore preserves the upper promise and cannot reduce $s$ below the physical lower bound $R^{1-\alpha}$. If the pilot stops, then $M_\alpha\ge\widetilde M-s/16\ge3s/16>s/8$. Once $s\le2M_\alpha$, the estimate is at least $M_\alpha-s/16\ge7s/16>s/4$, so the pilot must stop. This occurs after at most $k=\max\{0,\lceil\log_2(1/(2M_\alpha))\rceil\}$ halvings. Since $k\le\lceil(\alpha-1)\log_2R\rceil$, fewer than $N_{\mathrm{pilot}}$ calls suffice. Each conditional failure probability is at most $\nu/(8N_{\mathrm{pilot}})$; the union bound gives total pilot failure at most $\nu/8$. The iteration and scale checks ensure termination on the remaining histories as well.

\paragraph{Finite checks and endpoint behavior.}

At each scale, the pilot's additive tolerance $s/16$ includes both local approximation and classical output rounding. We compare the rounded value with $s/4$ and stop when equality holds. On the successful event, this rule still gives $M_\alpha\ge3s/16$. No additional comparison error is needed at the threshold.

For noninteger $\alpha$, the lower moment $R^{1-\alpha}$ need not have an exact finite representation. At a current dyadic scale $s$, we compute an interval $[l,u]$ containing it with width at most $s/32$. The controller rejects the call only when $s<l$; otherwise, the fixed iteration and query limits remain in force. On a successful history, the invariant already proves that $s$ is at least the true lower moment. An accepted call with an invalid promise can therefore occur only after a failure already charged to the probability budget.

There is a further deterministic consequence of this check. Any allowed call has $s\ge l\ge R^{1-\alpha}-s/32$, hence $s\ge(32/33)R^{1-\alpha}$. Thus even a bad transcript cannot select an arbitrarily tiny tolerance before being rejected. The parameter sizes in its per-call cap remain bounded by the same powers of $R$ and $1/\varepsilon$. This supplies the bounded-transcript parameter range used in the synthesis analysis.

At $M_\alpha=R^{1-\alpha}$ the same termination argument stops before an invalid halving; no flatness test is needed. A loose rank cap only permits an earlier stop relative to its worst-case bound.

\subsubsection{Final relative accuracy and entropy conversion}
On successful pilot and final calls, projection preserves the final moment error and
\begin{equation}\label{eq:renyi-conversion}
 \begin{gathered}
 \left|\widehat M/M_\alpha-1\right|\le8h<1/2,\\
 \frac{|\log\widehat M-\log M_\alpha|}{\alpha-1}
 \le\frac{16h}{\alpha-1}\le\varepsilon/16.
 \end{gathered}
\end{equation}
Here $h=\min\{1,\alpha-1\}\varepsilon/256$. Each call is the implemented estimator, whose guarantee already includes synthesis and arithmetic. We reserve another $\varepsilon/8$ for the final logarithm and rounding. The combined failure probability is at most $\nu/4<\nu$. Since the pilot uses only logarithmically many calls at constant relative-scale precision, the final call determines the $1/\varepsilon$ power. This proves the above-one rows of Theorem~\ref{thm:renyi}. Below one, projection onto $[1,R^{1-\alpha}]$ gives $|\log\widehat M-\log M_\alpha|\le|\widehat M-M_\alpha|$, so the moment construction in Theorem~\ref{thm:tsallis} with $u=(1-\alpha)\varepsilon/4$ proves the remaining row.

\subsubsection{Proof of Proposition~\ref{prop:upper-spectrum}}
We combine the ordinary rank bounds with $\Lambda=\min\{1,C/R\}$, retaining scales $\delta\le2\Lambda$. The moment tolerance is $u=h_fR^{1-\alpha}=\min\{1,\alpha-1\}\varepsilon R^{1-\alpha}/256$, which suffices for relative estimation because $M_\alpha\ge R^{1-\alpha}$. Below $1/R$, the coefficient sum is $O(R^{2-\alpha})$. Between $1/R$ and $2C/R$, there are $O(1+\log C)$ scales, each bounded by an $(\alpha,C)$-dependent multiple of $R^{2-\alpha}$. This also covers $R\le2C$, where $R$ is bounded in terms of $C$. Dividing the resulting $O(R^{2-\alpha})$ sum by $u$ and applying the logarithmic conversion gives $\Ot(R/\varepsilon)$. The common error allocation covers the omitted levels.

\subsection{Entropy bounds for the optional fixed-time route}\label{app:fixed-entropy}
\begin{proposition}[Fixed-time entropy bounds]\label{prop:fixed-entropy}
For $R\ge2$, fixed order, and the error and confidence ranges of Section~\ref{sec:applications}, the estimator of Section~\ref{sec:master} with Proposition~\ref{prop:fixed-band} satisfies the following bounds. Order one is excluded from the moment line. The moment tolerance is $u$; the other two lines use additive entropy error $\varepsilon$.
\begin{align}\label{eq:fixed-entropy-bounds}
 Q_S^{\mathrm{FT}}&=\Ot(R/\varepsilon^{3/2}),\nonumber\\
 Q_{M_\alpha}^{\mathrm{FT}}(u)&=\begin{cases}
 \Ot(R^{3/(2\alpha)-1/2}u^{-3/(2\alpha)}),&0<\alpha<3/2,\\
 \Ot(R^{2-\alpha}/u),&3/2\le\alpha<2,\\
 \Ot(1/u),&\alpha\ge2,
 \end{cases}\nonumber\\
 Q_{{\mathrm{R}},\alpha}^{\mathrm{FT}}&=\begin{cases}
 \Ot(R/\varepsilon^{3/(2\alpha)}),&1<\alpha<3/2,\\
 \Ot(R/\varepsilon),&3/2\le\alpha\le3,\\
 \Ot(R^{3/2-3/(2\alpha)}/\varepsilon),&\alpha>3.
 \end{cases}
\end{align}
Tsallis estimation uses $u=\Theta(\varepsilon)$; R\'enyi estimation below order one uses the same moment tolerance. For $1<\alpha<2$, the rank-independent alternative gives $\Ot(u^{-1/(\alpha-1)})$. Under the additional promise $\lambda_{\max}\le C/R$, fixed-time readout gives $\Ot(R/\varepsilon)$ for $\alpha\ge3/2$. Below $3/2$, the present fixed-time construction retains the first above-one R\'enyi bound, while variable time gives the $R/\varepsilon$ spectral-promise bound. Rank one and the exactly uniform promise require no queries.
\end{proposition}
\begin{proof}
We use the common error allocation and mass bounds, with uniform $e_\ell=u/(8n)$. For a power profile, the fixed-time coefficient is $\delta^{\alpha-2}\sqrt{\min\{1,8R\delta\}}$, which becomes $O(\sqrt R\delta^{\alpha-3/2})$ below $1/R$. When $\alpha<3/2$, the finest scale $t_R=\Theta((u/R)^{1/\alpha})$ dominates and gives the first moment bound. For $3/2\le\alpha<2$, splitting the sum at $1/R$ gives $O(R^{2-\alpha})$, up to a level factor. For $\alpha\ge2$, the rank-independent construction of Appendix~\ref{app:power} applies. The earlier affine and logarithmic conversions give the Tsallis and below-one R\'enyi bounds.

For the logarithm, the fine coefficient is $O(\sqrt R\delta^{-1/2}(1+\log(1/\delta)))$. With $t=\Theta(\varepsilon/[R\log(R/\varepsilon)])$, its sum gives $\Ot(R/\sqrt\varepsilon)$ before the $1/\varepsilon$ statistical factor. Coarser levels cost at most $\Ot(R)$.

For R\'enyi order above one, the pilot of Lemma~\ref{lem:pilot} supplies a moment upper scale $s$, and the final moment tolerance is $u=hs$. As in Appendix~\ref{app:renyi}, we set $\delta_\star=s^{1/(\alpha-1)}$ and $\Lambda=s^{1/\alpha}$. The three spectral regions then have coefficients
\begin{equation}\label{eq:fixed-renyi-regions}
 \frac{G\sqrt B}{hs\delta}=
 \begin{cases}
 \Theta(\sqrt R\,\delta^{\alpha-3/2}/(hs)),&\delta\le1/R,\\
 \Theta(\delta^{\alpha-2}/(hs)),&1/R\le\delta\le \delta_\star,\\
 \Theta(\delta^{(\alpha-3)/2}/(h\sqrt s)),&\delta_\star\le\delta\le\Lambda.
 \end{cases}
\end{equation}
Empty intervals contribute zero. The finest endpoint in the first interval gives $R^{3/(2\alpha)-1/2}(hs)^{-3/(2\alpha)}$ for $\alpha<3/2$; using $s\ge R^{1-\alpha}$ bounds it by $R/h^{3/(2\alpha)}$. At and above $3/2$ that interval contributes at most $\Ot(R/h)$. On the remaining intervals $S=1$, so their sums are exactly the upper-scale sums already proved in Appendix~\ref{app:renyi}: at most $\Ot(R/h)$ up to order three and $\Ot(R^{3/2-3/(2\alpha)}/h)$ above it. The bounded pilot uses constant $h$ and the final call uses $h=\Theta(\varepsilon)$, so the same finite pilot proof applies with the fixed-time caps. For $1<\alpha<2$, setting $u=\Theta(\varepsilon R^{1-\alpha})$ also proves the relevant bounds without a pilot. With $\Lambda=C/R$, the upper-scale interval is truncated at a constant multiple of $1/R$, giving the stated spectral-promise bound for $\alpha\ge3/2$. For $1<\alpha<3/2$ the fine endpoint is unchanged, so this argument does not remove its extra precision factor. At $\alpha=3/2,2,3$, endpoint sums may add level logarithms; constants are not asserted uniform as the fixed order approaches a boundary.
\end{proof}
The gain over this stronger fixed-time comparator is therefore concentrated in the small-survival region. For von Neumann entropy it is a factor $\varepsilon^{-1/2}$ in the proved upper bounds. For Tsallis order at least $3/2$ and R\'enyi order at least $3/2$, the displayed best powers can already be recovered by fixed time. Local normalization and moment caps determine these shared powers; survival information determines the additional variable-time saving. The first-stop version remains the formal optional branch of the estimator of Section~\ref{sec:master}; Proposition~\ref{prop:adjacent-bands} separately shows that exact first-stop histories are not necessary for these fixed-time comparison powers.

\section{Accuracy Lower Bounds and Oracle Compatibility}\label{app:lower}
\subsection{A rank-two accuracy lower bound}
\begin{proposition}\label{prop:accuracy-lower}
For each fixed $\alpha>0$, $\alpha\ne1$, and sufficiently small $\varepsilon>0$, estimating any of $T_\alpha$, $S_\alpha^{\mathrm{R}}$, or $S$ to additive error $\varepsilon$ with success at least $2/3$ requires $\Omega(1/\varepsilon)$ controlled purified queries in the worst case, even at rank two. The constant for $S$ is absolute.
\end{proposition}
\begin{proof}
We restrict to $\rho_t=\operatorname{diag}(t,1-t)$ near $t=1/4$. A rotation by angle $\theta$ with $t=\sin^2\theta$, followed by copying the computational-basis label to an environment qubit, gives a valid purification oracle with a fixed unitary completion. Its forward, inverse, and controlled forms differ in operator norm by $O(|\theta-\theta'|)$. The three entropy derivatives are
\begin{equation}\label{eq:entropy-derivatives}
 \begin{split}
 \frac{d}{dt}S(\rho_t)&=\log((1-t)/t),\\
 \frac{d}{dt}T_\alpha(\rho_t)&=\frac{\alpha}{1-\alpha}
 (t^{\alpha-1}-(1-t)^{\alpha-1}),\\
 \frac{d}{dt}S_\alpha^{\mathrm{R}}(\rho_t)&=
 \frac{\alpha(t^{\alpha-1}-(1-t)^{\alpha-1})}
 {(1-\alpha)(t^\alpha+(1-t)^\alpha)}.
 \end{split}
\end{equation}
Each derivative is nonzero at $t=1/4$, where $dt/d\theta$ is also bounded away from zero. We choose two angles separated by $c_\alpha\varepsilon$ whose entropy values differ by more than $2\varepsilon$. An accurate estimate distinguishes their oracles with constant advantage. By replacing the $Q$ queries one at a time, the hybrid argument~\cite{BBBV97} bounds the final trace distance by $O(Q\varepsilon)$. The same argument covers adaptive algorithms after their measurements are recorded coherently. Constant distinguishability therefore requires $Q=\Omega(1/\varepsilon)$.
\end{proof}
Together with Theorem~\ref{thm:tsallis}, this gives fixed-order Tsallis complexity $\widetilde\Theta(1/\varepsilon)$ for every real $\alpha\ge2$, for the unrestricted class or any rank cap at least two.

\subsection{Sample-to-query lifting with controlled access}
We use the random-completion simulation of~\cite[Theorem~1.5]{Tang26}, alongside the block-encoding lifting perspective of~\cite[Theorem~3.1]{Lifting25}. For an environment of dimension at least twice the system dimension (or twice a known rank), its distribution is supported on reflection completions of a purification oracle. Forward and inverse calls can be simulated to constant total error from $O(Q^2)$ samples. Since our algorithms must succeed for every valid completion, correctness also holds on this distribution. Known zero padding supplies the required environment size.

The controlled extension follows at the sample-based projector-simulation step of that proof. For a sample state $\sigma$ on a register $P$, target $T$, control $c$, and their SWAP $F_{TP}$, we apply
\begin{equation}\label{eq:controlled-swap}
 U_h=\exp(-ih\ket1\bra1_c\otimes F_{TP}).
\end{equation}
We denote by $\mathcal E_h$ the channel obtained after tracing out $P$. Its action on an input operator $X$ satisfies
\begin{equation}\label{eq:controlled-generator}
 \mathcal E_h(X)=X-ih[\ket1\bra1\otimes\sigma,X]+O(h^2).
\end{equation}
Indeed, expanding in control blocks and using $\Tr_P[F(Y\otimes\sigma)]=\sigma Y$ and $\Tr_P[(Y\otimes\sigma)F]=Y\sigma$ proves the first-order identity. The remainder is $O(h^2)$ in diamond norm: the generator has norm at most one, and adjoining a sample and tracing it out are contractive channels. With $N_{\mathrm{sam}}$ fresh samples and $h=t/N_{\mathrm{sam}}$, the channel hybrid bound is $O(t^2/N_{\mathrm{sam}})$. For a pure projector and $t=\pi$, this simulates a controlled reflection; a known control phase chooses its sign convention. All random purification choices in the imported construction remain shared. Setting per-call error $O(1/Q)$ gives $O(Q^2)$ samples for bounded-control forward/inverse access as well.

Consequently, a sample lower bound $L_{\mathrm{sam}}$ for separated entropy hypotheses implies $\Omega(\sqrt{L_{\mathrm{sam}}})$ queries in this model. In~\cite[Theorem~5.11]{WangFramework26}, for fixed $1<\alpha<2$ and sufficiently small $\varepsilon$, the Tsallis sample lower bound is $\Omega(\varepsilon^{-2/(\alpha-1)}/\log^4(1/\varepsilon))$ when
\begin{equation}\label{eq:hard-dimension}
 d\ge1+\lfloor a_\alpha\varepsilon^{-1/(\alpha-1)}\rfloor
\end{equation}
for an appropriate positive constant $a_\alpha$. Lifting matches the rank-free precision exponent in Theorem~\ref{thm:tsallis}. A rank-constrained class must contain this hard family for the same conclusion. Applying the same lifting to~\cite[Theorem~5.3]{WangFramework26} gives the weaker logarithm-suppressed von Neumann bound $\widetilde\Omega(R/\sqrt\varepsilon+1/\varepsilon)$ for sufficiently large $R$ and small $\varepsilon$, by embedding an $R$-dimensional hard family.

\subsection{A density functional with a necessary degree--precision product}\label{app:product-example}

The entropy bounds exploit relations between function magnitude and spectral mass; smoothness or polynomial representability alone does not guarantee the same saving. We demonstrate this within the purified-state model itself. For the Chebyshev polynomial $T_D$, the profile $g_D(x)=T_D(2x-1)$ is bounded by one on $[0,1]$, but estimating its spectral expectation can require a product of degree and inverse precision.

\begin{proposition}[A necessary product for a density functional]\label{prop:product-cost}
For every even integer $D\ge2$ and $0<\varepsilon\le1/20$, the functional $F_D(\rho)=\Tr[\rho g_D(\rho)]$ has query complexity $\Theta(D/\varepsilon)$ on the promised rank-two family
\begin{equation}\label{eq:rotation-instance}
 \rho_\theta=\operatorname{diag}\!\left(
 \frac{1+\sin\theta}{2},\frac{1-\sin\theta}{2}\right),
 \qquad 0\le\theta\le\pi/3.
\end{equation}
The query model includes controlled purified preparation and its inverse, and the success probability is at least $2/3$.
\end{proposition}
\begin{proof}
Because $D$ is even, the two eigenvalues give the same profile value. The Chebyshev identity therefore yields
\begin{equation}\label{eq:density-product-target}
 F_D(\rho_\theta)=T_D(\sin\theta)=(-1)^{D/2}\cos(D\theta).
\end{equation}
For an upper bound valid for every purification of a state in this promised family, we estimate $p=\langle0|\rho_\theta|0\rangle$ to additive error $\varepsilon/(8D)$ by amplitude estimation~\cite{BHMT02}. We project the estimate onto $[1/2,(1+\sqrt3/2)/2]$ and recover $\theta=\arcsin(2p-1)$. On this interval, the derivative of the inverse map is at most four, while $\cos(D\theta)$ is $D$-Lipschitz in $\theta$. The resulting functional error is at most $\varepsilon/2$, leaving room for finite arithmetic, and the query cost is $O(D/\varepsilon)$.

For the lower bound, we use the valid purification unitary
\begin{equation}\label{eq:density-product-oracle}
 \begin{gathered}
 U_\theta=\mathrm{CNOT}_{S,E}
       [R(\pi/4-\theta/2)_S\otimes I_E],\\
 R(\phi)=\begin{pmatrix}\cos\phi&-\sin\phi\\
                          \sin\phi&\cos\phi\end{pmatrix}.
 \end{gathered}
\end{equation}
Here the controlled-NOT has control $S$ and target $E$. Acting on $\ket{00}$, this unitary has squared Schmidt coefficients equal to the two diagonal entries of \eqref{eq:rotation-instance}. We choose
\begin{equation}\label{eq:rotation-hard-pair}
 \theta_0=\frac{\pi}{2D},\qquad
 \theta_1=\theta_0+\frac{4\varepsilon}{D}.
\end{equation}
Both angles lie in the promised interval, and their target values differ by $\sin(4\varepsilon)>2\varepsilon$. The two preparation unitaries differ in operator norm by at most $|\theta_1-\theta_0|/2=2\varepsilon/D$, as do their inverses and controlled versions. The query hybrid argument~\cite{BBBV97} bounds the final distinguishability of any $Q$-query algorithm by $O(Q\varepsilon/D)$. An additive-$\varepsilon$ estimate distinguishes these separated targets with constant advantage, so $Q=\Omega(D/\varepsilon)$. Since correctness in our model is required for every allowed purification and completion, this family of completions suffices for the lower bound.
\end{proof}

Here the polynomial degree varies, and the normalized local analytic bounds cannot be taken uniformly in $D$. This is consistent with Definition~\ref{def:analytic}, whose constants depend on the specified function family. The example distinguishes an avoidable product of separately maximized algorithmic costs from a genuine product required by the input problem. In our entropy applications, the local magnitudes and probability bounds provide the additional structure that yields the improved scaling.

\subsection{Known integer moments and comparison ranges}

For a fixed integer $k\ge2$, the cyclic shift on $k$ independently prepared copies has expectation $\operatorname{Tr}(\rho^k)$. Its Hadamard test and ordinary amplitude estimation give additive moment error $v$ using $O(1/v)$ purified queries~\cite{Ekert02,BHMT02}. Taking $v=\Theta(\varepsilon R^{1-k})$ therefore gives the existing R\'enyi upper bound $O(R^{k-1}/\varepsilon)$. This is a useful additional baseline at small integer orders or very high requested precision.

\begin{equation}\label{eq:integer-renyi-prior}
 U_{\mathrm{R}}^{\mathrm{int}}(k)=\min\left\{
 \frac{R}{\varepsilon^{1+1/k}},\frac{R^{k-1}}{\varepsilon}\right\}.
\end{equation}

\begin{proposition}[Direct odd-moment baseline]\label{prop:odd-moment}
For fixed $k=2m+1\ge3$, additive-$\varepsilon$ R\'enyi estimation uses $O(R^m\varepsilon^{-1}\log(1/\nu))$ controlled purified queries.
\end{proposition}
\begin{proof}
We allocate $m$ independent density-block auxiliary registers $A_1,\ldots,A_m$, initially zero, and a common system $S$. Each $W_\rho^{(A_jS)}$ has zero block $\rho$. Projecting the distinct auxiliary registers onto zero contracts these blocks in sequence and gives
\begin{equation}\label{eq:odd-block-product}
 \begin{gathered}
 V_m=W_\rho^{(A_mS)}\cdots W_\rho^{(A_1S)},\\
 (\bra{0}^{\otimes m}\otimes I)V_m(\ket{0}^{\otimes m}\otimes I)=\rho^m,\\
 p=\| (\rho^m\otimes I_E)\ket{\psi_\rho}\|^2
   =\Tr(\rho^{2m+1})\ge R^{-2m}.
 \end{gathered}
\end{equation}
Here each $\ket0$ denotes a complete auxiliary register. One preparation uses $2m+1$ queries, including the processed purification, and its inverse has the same cost. We define success by the joint zero state of all auxiliaries at the end, without intermediate postselection. Amplitude estimation with length $O(R^m/\varepsilon)$ satisfies $|\widehat p-p|\le c_k\varepsilon p$ by \eqref{eq:cal-bhmt}, for a sufficiently small fixed $c_k$; an odd median reduces the failure probability to $\nu$. Projection onto $[R^{-2m},1]$ preserves the error bound, and $|\log(1+z)|\le2|z|$ for $|z|\le1/2$ gives the entropy guarantee. The auxiliary registers must be independent: since the swap-based encoding satisfies $W_\rho^2=I$, repeated powers on the same auxiliaries would not encode $\rho^m$.
\end{proof}
Consequently orders two and three already admit $O(R/\varepsilon)$ at constant confidence. Equation~\eqref{eq:integer-renyi-prior} records the permutation and earlier general alternatives; for odd orders it is further improved by Proposition~\ref{prop:odd-moment}. For odd $k\ge5$, $(k-1)/2\ge3/2-3/(2k)$, so this additional baseline is dominated by our multi-level term. The even-order permutation bound is also dominated for $k\ge4$. Comparing the displayed powers with the general bound $\Ot(R/\varepsilon^{1+1/\alpha})$, with logarithmic factors suppressed, again makes our term no larger when $\varepsilon\le R^{-(\alpha-3)/2}$.

The one-sided spectral promise at $C=1$ fixes all entropy values. The final subsection below gives this zero-query case and the precision lower bound for $C>1$.

\subsection{Comparison bounds and parameter ranges}\label{app:table-certificates}

The quantum-state reduction of WZL~\cite{WZL24} replaces its cutoff and power parameters by $\delta'=\delta^2,c'=c/2$ below order one, and by $\nu'=\nu^2,c'=c/2,\beta'=\beta^2$ above order one. Its Corollaries~4 and~5 then give $\widetilde O(R^{1/\alpha}/\varepsilon^{1+1/\alpha})$ and $\widetilde O(R/\varepsilon^{1+1/\alpha})$, respectively, at fixed confidence; these are the rank-based alternatives in our tables. These source parameters are unrelated to our failure budget and readout normalization.

For $\alpha>1$, the Tsallis entropy is a smooth function of the R\'enyi entropy with derivative $M_\alpha\le1$. Thus a R\'enyi additive-error estimator also gives a Tsallis estimator to the same error, supplying the rank-based alternative in Table~\ref{tab:results}~\cite{WZL24}.

For the below-one Tsallis comparison, we define
\begin{equation}\label{eq:prior-tsallis}
 U_T^{\mathrm{old}}=\min\left\{
 \frac{R^{(3-\alpha^2)/(2\alpha)}}{\varepsilon^{(3+\alpha)/(2\alpha)}},
 \frac{R^{2/\alpha-\alpha}}{\varepsilon^{1+1/\alpha}}\right\}.
\end{equation}

The first bound is the direct estimator of~\cite[Theorem~III.9]{WangEntropies24}. The second follows from the relative moment estimator underlying~\cite[Corollary~4]{WZL24}, using $M_\alpha\le R^{1-\alpha}$ and relative tolerance $\Theta(\varepsilon/R^{1-\alpha})$. Thus the comparison includes this conversion rather than only the older direct bound.

For completeness, the dimension-based below-one R\'enyi estimator of~\cite[Corollary~3]{WZL24} also yields a Tsallis baseline. With $M_\alpha\le R^{1-\alpha}$, an entropy error of order $\varepsilon/R^{1-\alpha}$ gives moment error $O(\varepsilon)$ by exponentiation, hence Tsallis error $O(\varepsilon)$. We apply that estimator with this entropy tolerance to obtain the Tsallis query bound
\begin{equation}\label{eq:prior-tsallis-dimension}
 \Ot\!\left(
 d_\rho^{(1+\alpha)/(2\alpha)}
 R^{(1-\alpha)(1+1/(2\alpha))}
 \varepsilon^{-1-1/(2\alpha)}\right).
\end{equation}
The constants in the entropy tolerance can be decreased to achieve the requested error exactly. Table~\ref{tab:results} takes the minimum of this bound and the two terms in \eqref{eq:prior-tsallis}. For below-one R\'enyi entropy on full-rank states ($d_\rho=R$), dividing our bound by the dimension-based baseline gives $R^{(1-\alpha)/(2\alpha)}\varepsilon^{(2\alpha-1)/(2\alpha)}$. Thus the earlier dimension-based algorithm is no worse for $\alpha\le1/2$; for $1/2<\alpha<1$, comparison of the displayed powers makes our term no larger when $\varepsilon\le R^{-(1-\alpha)/(2\alpha-1)}$. These comparisons describe the rank and precision powers with logarithmic factors suppressed.

All statements in this subsection concern sufficiently small additive entropy error, constant success probability, and fixed order. For dimension-dependent statements, we take a sufficiently large integer $R$ and embed the $R$-dimensional hard instance in a known subspace of the actual system. The resulting state has rank at most $R$. Any algorithm for the rank-promised class must solve this subclass. In contrast, a dimension-dependent upper bound cannot replace its dimension by an unknown rank without a support-access construction.

Theorem~1.1 and Definition~1.1 of~\cite{BoundsList25} state the lifting interface with controlled forward/inverse purified access explicitly. Combining that interface with the sample theorems below proves the query lower bounds.

To state the lifted lower bounds compactly, we define the following four expressions. The labels $\mathrm{T}$ and $\mathrm{R}$ denote Tsallis and R\'enyi entropy, respectively; the superscripts $<1$, $\mathrm{nonint}$, and $\mathrm{int}$ denote orders below one, noninteger orders above one, and integer orders at least two. Their assumptions and sample-to-query conversions are verified immediately below.
\begin{align}\label{eq:table-lower-defs}
 L_{\mathrm{T}}^{<1}&=R^{(1+1/\alpha)/2}\varepsilon^{-1/(2\alpha)}
              +R^{1-\alpha}/\varepsilon,\nonumber\\
 L_{\mathrm{R}}^{<1}&=R^{(1+1/\alpha)/2}\varepsilon^{-1/(2\alpha)}
              +R^{(1/\alpha-1)/2}/\varepsilon,\nonumber\\
 L_{\mathrm{R}}^{\mathrm{nonint}}&=R\varepsilon^{-1/(2\alpha)}
              +R^{(1-1/\alpha)/2}/\varepsilon,\nonumber\\
 L_{\mathrm{R}}^{\mathrm{int}}&=R^{1-1/\alpha}\varepsilon^{-1/\alpha}
              +R^{(1-1/\alpha)/2}/\varepsilon.
\end{align}

For $0<\alpha<1$, Theorems~5.9 and~5.11 of~\cite{WangFramework26} give, respectively, the R\'enyi and Tsallis sample bounds. Their common first term is

\begin{equation}\label{eq:sample-below-one}
 \frac{R^{1+1/\alpha}}{\varepsilon^{1/\alpha}\log^2R
              \max\{1,\varepsilon^{1/\alpha}\log^2R\}}.
\end{equation}

Their second terms are $R^{1/\alpha-1}/\varepsilon^2$ and $R^{2-2\alpha}/\varepsilon^2$. Taking square roots gives $\OtOmega(L_{\mathrm{R}}^{<1})$ and $\OtOmega(L_{\mathrm{T}}^{<1})$. The maximum in the denominator affects only the suppressed logarithms in the fixed small-error range. The square root follows from simulating queries with samples in the lifting reduction.

For noninteger $\alpha>1$, Theorem~5.9 of~\cite{WangFramework26} gives sample lower bound

\begin{equation}\label{eq:sample-noninteger}
 \Omega\!\left(
 \frac{R^2}{\varepsilon^{1/\alpha}\log^{\kappa_\alpha}R}
 +\frac{R^{1-1/\alpha}}{\varepsilon^2}\right),
 \qquad \kappa_\alpha=\frac{4\lceil\alpha\rceil}{\lceil\alpha\rceil-\alpha}.
\end{equation}

Applying the controlled-access lifting argument to this sample bound gives $\OtOmega(L_{\mathrm{R}}^{\mathrm{nonint}})$ for fixed noninteger $\alpha$; its logarithmic exponent depends on the distance to an integer. For integer $\alpha\ge2$, we use the direct query consequence in~\cite[Theorem~2.24]{BoundsList25}, based on the sample analysis of~\cite{Acharya20}; this yields $\Omega(L_{\mathrm{R}}^{\mathrm{int}})$.

For von Neumann entropy, Theorem~5.3 of~\cite{WangFramework26} gives sample lower bound \[\Omega(R^2/[\varepsilon\log^2R\max\{1,\varepsilon\log^2R\}]+\log^2R/\varepsilon^2).\]
Lifting yields $\OtOmega(R/\sqrt\varepsilon+1/\varepsilon)$. This matches the linear rank dependence of Theorem~\ref{thm:vn} at constant precision and leaves a gap in the joint rank--precision dependence. Proposition~\ref{prop:accuracy-lower} supplies the precision lower bound for fixed small ranks outside the large-$R$ statement.

\subsection{One-sided spectral promises}\label{app:spectral-promise-lower}

If $\lambda_{\max}\le1/R$ and rank is at most $R$, trace one forces exactly $R$ positive eigenvalues, all equal to $1/R$. Thus $S=S_\alpha^{\mathrm{R}}=\log R$ and $T_\alpha=(1-R^{1-\alpha})/(\alpha-1)$ require zero queries. For fixed $C>1$, a precision lower bound is possible within the promise itself. For even $R$, take $R/2$ eigenvalues equal to $(1+t)/R$ and $R/2$ equal to $(1-t)/R$, with $t$ in a neighborhood of a fixed $t_0\in(0,\min\{1,C-1\})$. Then

\begin{equation}\label{eq:promised-renyi-family}
 S_\alpha^{\mathrm{R}}(\rho_t)=\log R+
 \frac1{1-\alpha}\log\frac{(1+t)^\alpha+(1-t)^\alpha}{2}.
\end{equation}

The derivative in $t$ is nonzero at $t_0$. A known uniform preparation inside each group and a two-dimensional mixing rotation supply purification oracles with norm difference $O(|t-t'|)$. Two choices separated by $\Theta(\varepsilon)$ prove $\Omega(1/\varepsilon)$ by the same hybrid argument. This proves the precision lower bound for the promised family at fixed $R$. It suffices for the corresponding qualified entry in Table~\ref{tab:promises}.
\Needspace{13\baselineskip}
\section*{Acknowledgments}
This work is supported by the National Natural Science Foundation of China (Grant No. 62501060), the Quantum Science and Technology–National Science and Technology Major Project (Grant No. 2023ZD0300200), the Beijing Natural Science Foundation (Grant No. Z250004), NSAF (Grant No. U2330201), the National Natural Science Foundation of China (Grant No. 12361161602), and the Beijing Science and Technology Planning Project (Grant No. Z25110100810000).

The core ideas of this work originated with the authors. We developed the arguments and manuscript through multiple rounds of interaction with OpenAI Codex (\mbox{GPT-6 Astra}), which assisted with exploring proof strategies, checking mathematical arguments and references, and drafting and revising the text. The authors take responsibility for all results, proofs, citations, and text.

\bibliographystyle{IEEEtran}
\bibliography{entropy_query_revised}

\end{document}